%% file: main.tex
\documentclass{article}

\input{macros}

\title{Dynamic Contention Resolution Schemes}

\author{Moran Feldman\thanks{Computer Science Department, University of Haifa, Israel. This work was done while the author was visiting Queen Mary University of London. E-mail: \email{moranfe@cs.haifa.ac.il}.} \and Gregory Kehne\thanks{Computer Science \& Engineering, Washington University in St. Louis, Missouri. E-mail: \email{kehne@wustl.edu}.} \and Roie Levin\thanks{Department of Computer Science, Rutgers University, New Jersey. E-mail: \email{roie.levin@rutgers.edu}.} \and Sherry Sarkar\thanks{Carnegie Mellon University (CMU), Pennsylvania. E-mail: \email{academia.sarkar@gmail.com}.}}

\date{}

\begin{document}

\maketitle

\begin{abstract}
\medskip
    We introduce a low-recourse rounding paradigm for packing problems in fully dynamic settings, which we name \emph{Dynamic Contention Resolution Schemes} (DCRSs). These are dynamic analogs of (Online) Contention Resolution Schemes (or (O)CRSs) for low-recourse dynamic optimization and offer a variety of benefits. Similarly to their offline and online counterparts, DCRSs for different constraints can be combined to obtain DCRSs for the constraints' intersection. Furthermore, together with the Positive Body Chasing framework of Bhattacharya, Buchbinder, Levin, and Saranurak [FOCS 2023], DCRSs imply \emph{competitive} recourse algorithms for fully dynamic packing problems with submodular objectives: these are algorithms that, for any input sequence, incur recourse that is itself competitive with the best possible recourse for that sequence.
    
    We show the existence of $\Omega(1)$-balanced and $O(\log \rank)$-recourse DCRSs for matroid constraints, and $\Omega(1)$-balanced/$O(1)$-recourse DCRSs for matching and knapsack constraints. In particular, these yield the first non-trivial recourse bound for fully dynamic knapsack, as well as the first competitive-recourse algorithm for non-bipartite matching, and both of these apply even to monotone submodular objectives. Beyond our particular results, we view the DCRS framework as a principled step towards mechanizing the relax-and-round paradigm of approximation algorithms in the context of dynamic optimization.
\end{abstract}

\newpage
\input{intro.tex}

\input{prelim.tex}

\input{matroid_new.tex}

\input{fully_dynamic.tex}

\input{partition.tex}

\input{knapsack.tex}

\input{combiner.tex}

\input{submod.tex}

\newpage
\appendix
\input{full_prelim.tex}
\input{elided.tex}
\input{lfold_new.tex}
\input{adversary.tex}

\input{tcos-variance}
\input{polytime.tex}
\input{round-or-sep.tex}
\input{acknowledgements.tex}

\bibliographystyle{alpha}
\bibliography{dblp,refs}

\end{document}

%% file: macros.tex
\usepackage[T1]{fontenc}
\usepackage[utf8]{inputenc}
\usepackage{typearea}
\usepackage{forloop}
\usepackage[nopatch=footnote]{microtype}
\usepackage{fullpage}
\usepackage{latexsym}
\usepackage{paralist}
\usepackage{natbib}

\usepackage{mathrsfs}
\usepackage{subfig}

\usepackage{comment} 
\usepackage{bbm}
\usepackage{framed} 
\usepackage{url} 
\usepackage{booktabs}
\usepackage{amsmath,amssymb}

\usepackage{amsthm}
\usepackage{thmtools} 
\usepackage{thm-restate}
\usepackage{nicefrac}
\usepackage{calc}
\usepackage{enumerate}
\usepackage{enumitem}
\usepackage[dvipsnames]{xcolor}
\usepackage{algorithm}
\usepackage{algorithmicx}
\usepackage{graphicx}
\usepackage[font=footnotesize,labelfont=bf]{caption}
\usepackage[nobreak=true]{mdframed}
\usepackage{appendix}
\usepackage[noend]{algpseudocode}
\usepackage[colorinlistoftodos]{todonotes}
\usepackage{xr}
\usepackage{array}
\usepackage{xspace}
\definecolor{lightgray}{gray}{0.9} 
\definecolor{ForestGreen}{rgb}{0.1333,0.5451,0.1333}
\definecolor{DarkRed}{rgb}{0.8,0,0}
\definecolor{Red}{rgb}{1,0,0}
\usepackage[linktocpage=true,
pagebackref=true,colorlinks,
linkcolor=DarkRed,citecolor=ForestGreen,
bookmarks,bookmarksopen,bookmarksnumbered]{hyperref}

\usepackage[capitalise]{cleveref}
\crefrangelabelformat{enumi}{#3#1#4--#5#2#6}
\usepackage{parskip}
\usepackage{chngcntr}
\usepackage{mathtools,stackengine}
\usepackage{multirow}

\stackMath
\newcommand{\stackGeq}[1]{%
	\setbox0=\hbox{${}\mathrel{\stackon[-1pt]{\geq}{\scriptstyle\text{#1\strut}}}{}$}
	\xdef\tmpwd{\dimexpr\the\wd0\relax}
	\kern.5\tmpwd\mathclap{\box0}&\kern.5\tmpwd
}

\usepackage{nameref}

\usepackage{tikz}
\usetikzlibrary{patterns}


\allowdisplaybreaks

\DeclareMathOperator*{\ber}{Ber}

\DeclareMathOperator*{\Exp}{Exp}
\DeclareMathOperator*{\expectation}{\mathbb{E}}
\let\poly\relax
\DeclareMathOperator*{\poly}{poly}

\DeclareMathOperator*{\uni}{Uniform}
\DeclareMathOperator*{\var}{Var}
\DeclareMathOperator*{\probability}{\Pr}

\DeclareMathOperator{\rank}{rank}
\newcommand{\Rset}{R}
\newcommand{\Rdist}{{\mathtt{R}}}
\newcommand{\ISet}{{I}}
\newcommand{\Iset}{\ISet}

\newcommand{\characteristic}{{\mathbf{1}}}
\renewcommand{\Reset}{{\texttt{Reset}}\xspace}
\newcommand{\tightcolorbox}[2]{\setlength{\fboxsep}{0.5pt}\colorbox{#1}{#2}}

\newcommand{\lvl}{i}
\newcommand{\Inc}{\mathrm{Inc}}
\newcommand{\Dec}{\mathrm{Dec}}

\newcommand{\Span}{\mathrm{span}}

\newcommand\eps{\varepsilon}
\renewcommand\emptyset{\varnothing}

\newcommand\norm[1]{\left\| #1 \right\|}

\newcommand\abs[1]{\lvert #1 \rvert}
\newcommand{\expect}{\expectation\expectarg}
\DeclarePairedDelimiterX{\expectarg}[1]{[}{]}{%
	\ifnum\currentgrouptype=16 \else\begingroup\fi
	\activatebar#1
	\ifnum\currentgrouptype=16 \else\endgroup\fi
}
\newcommand{\email}[1]{{\href{mailto:#1}{#1}}}

\newcommand{\defcal}[1]{\expandafter\newcommand\csname c#1\endcsname{{\mathcal{#1}}}}
\newcommand{\defbb}[1]{\expandafter\newcommand\csname b#1\endcsname{{\mathbb{#1}}}}
\newcommand{\defvec}[1]{\expandafter\newcommand\csname v#1\endcsname{{\mathbf{#1}}}}
\newcounter{calBbCounter}
\forLoop{1}{26}{calBbCounter}{
    \edef\letter{\alph{calBbCounter}}
		\edef\Letter{\Alph{calBbCounter}}
    \expandafter\defcal\Letter
		\expandafter\defbb\Letter
		\expandafter\defvec\letter
}

\DeclarePairedDelimiterX{\nicesetarg}[1]{\{}{\}}{%
	\ifnum\currentgrouptype=16 \else\begingroup\fi
	\activatebar#1
	\ifnum\currentgrouptype=16 \else\endgroup\fi
}

\newcommand{\innermid}{\nonscript\;\delimsize\vert\nonscript\;}
\newcommand{\activatebar}{%
	\begingroup\lccode`\~=`\|
	\lowercase{\endgroup\let~}\innermid 
	\mathcode`|=\string"8000
}

\DeclarePairedDelimiterXPP{\prob}[1]{\probability}{(}{)}{}{#1}
\DeclarePairedDelimiterXPP{\probarg}[2]{\probability_{#2}}{(}{)}{}{#1}

\newcommand\opt{\textsc{Opt}\xspace}

\newcommand\optr{\textsc{Opt}_{\textsc{Rec}}\xspace}

\counterwithin{equation}{section}

\usepackage{eqparbox}
\renewcommand{\algorithmiccomment}[1]{\hfill\eqparbox{COMMENT}{{\color{BrickRed}// #1}}}

\newcommand\symdif{\hspace{0.01in}\triangle\hspace{0.01in}}

\newcommand\mst{\textsc{MinimumSpanningTree}\xspace}

\newcommand\setcov{\textsc{SetCover}\xspace}
\newcommand\kcent{$k$-\textsc{Center}\xspace}
\newcommand\kmed{$k$-\textsc{Median}\xspace}
\newcommand\facloc{\textsc{FacilityLocation}\xspace}
\newcommand\loadbal{\textsc{LoadBalancing}\xspace}
\newcommand\matching{\textsc{BipartiteMatching}\xspace}
\newcommand\gmatching{\textsc{Matching}\xspace}

\theoremstyle{plain}

\newtheorem{theorem}{Theorem}[section]
\newtheorem{definition}[theorem]{Definition}
\newtheorem{proposition}[theorem]{Proposition}

\newtheorem{lemma}[theorem]{Lemma}

\newtheorem{observation}[theorem]{Observation}
\newtheorem{claim}[theorem]{Claim}

\newtheorem{example}[theorem]{Example}

\newtheorem{corollary}[theorem]{Corollary}

\newtheorem{remark}[theorem]{Remark}



\DeclareFontShape{T1}{cmr}{m}{scit} { <-> ssub * cmr/m/scsl }{}

\algdef{SE}[WHILE]{MultilineWhile}{EndWhile}[1]{\algorithmicwhile\ #1}{\algorithmicend\ \algorithmicwhile}

\newlength{\continueindent}
\usepackage{etoolbox}
\makeatletter
\newcommand*{\ALG@customparshape}{\parshape 2 \leftmargin \linewidth \dimexpr\ALG@tlm+\continueindent\relax \dimexpr\linewidth+\leftmargin-\ALG@tlm-\continueindent\relax}
\apptocmd{\ALG@beginblock}{\ALG@customparshape}{}{\errmessage{failed to patch}}
\makeatother

\makeatletter
\def\thm@space@setup{%
	\thm@preskip=\parskip \thm@postskip=0pt
}
\makeatother

\usepackage{etoolbox}
\usepackage{tikz}
\usetikzlibrary{tikzmark,positioning, arrows.meta, fit, backgrounds}

\usetikzlibrary{calc}

\errorcontextlines\maxdimen

\newcommand{\ALGtikzmarkcolor}{black}
\newcommand{\ALGtikzmarkextraindent}{4pt}
\newcommand{\ALGtikzmarkverticaloffsetstart}{-.5ex}
\newcommand{\ALGtikzmarkverticaloffsetend}{-.5ex}
\makeatletter
\newcounter{ALG@tikzmark@tempcnta}

\newcommand\ALG@tikzmark@start{%
	\global\let\ALG@tikzmark@last\ALG@tikzmark@starttext%
	\expandafter\edef\csname ALG@tikzmark@\theALG@nested\endcsname{\theALG@tikzmark@tempcnta}%
	\tikzmark{ALG@tikzmark@start@\csname ALG@tikzmark@\theALG@nested\endcsname}%
	\addtocounter{ALG@tikzmark@tempcnta}{1}%
}

\def\ALG@tikzmark@starttext{start}
\newcommand\ALG@tikzmark@end{%
	\ifx\ALG@tikzmark@last\ALG@tikzmark@starttext
	\else
	\tikzmark{ALG@tikzmark@end@\csname ALG@tikzmark@\theALG@nested\endcsname}%
	\tikz[overlay,remember picture] \draw[\ALGtikzmarkcolor] let \p{S}=($(pic cs:ALG@tikzmark@start@\csname ALG@tikzmark@\theALG@nested\endcsname)+(\ALGtikzmarkextraindent,\ALGtikzmarkverticaloffsetstart)$), \p{E}=($(pic cs:ALG@tikzmark@end@\csname ALG@tikzmark@\theALG@nested\endcsname)+(\ALGtikzmarkextraindent,\ALGtikzmarkverticaloffsetend)$) in (\x{S},\y{S})--(\x{S},\y{E});%
	\fi
	\gdef\ALG@tikzmark@last{end}%
}

\apptocmd{\ALG@beginblock}{\ALG@tikzmark@start}{}{\errmessage{failed to patch}}
\pretocmd{\ALG@endblock}{\ALG@tikzmark@end}{}{\errmessage{failed to patch}}
\makeatother

\algblock[with]{With}{EndWith}
\algblockdefx[With]{With}{EndWith}%
[1]{\textbf{with} #1 \textbf{do}}%
{}

\makeatletter
\ifthenelse{\equal{\ALG@noend}{t}}%
{\algtext*{EndWith}}
{}%
\makeatother

\newcommand\ind[1]{\mathbbm{1}\left\{#1\right\}} 

\newcommand\nf{\nicefrac}
\newcommand\defeq{\coloneqq}

\DeclareMathOperator{\Cov}{Cov}

%% file: intro.tex
\section{Introduction}

The high-level goal of low-recourse (a.k.a.\ consistent or stable) algorithms is to maintain near-optimal solutions to combinatorial optimization problems as constraints change dynamically, while making minimal edits to these solutions over time. Though interest in the field has burgeoned among both theoreticians and practitioners over the last few decades, balancing the dueling objectives of quality and consistency has proven to be a difficult and problem-specific endeavor. Few generic tools are known, and algorithms seem to require ad hoc combinatorial insights.

Recently, 
Bhattacharya, Buchbinder, Levin, and Saranurak \cite{BBLS23} 
proposed a principled framework for solving low-recourse optimization problems, following the classical relax-and-round paradigm of approximation algorithms. Their core contribution was an algorithm that, for the special class of mixed packing/covering problems, maintains a \emph{fractional} solution (approximately) satisfying the LP relaxation of any such problem, with low recourse as measured by the total $\ell_1$ movement of the fractional solution. This reduces the design of low-recourse algorithms for the class of mixed packing/covering problems to a low-recourse rounding question. 

The authors further gave a proof of concept of their framework by showing low recourse rounding algorithms for a small number of fundamental problems: \setcov, \mst, \loadbal and \matching \cite{BBLS23}. In follow-up work building on the same framework, an overlapping set of authors showed low-recourse rounding algorithms for \kcent, \kmed, and \facloc clustering~\cite{buchbindercompetitively}. Yet these rounding algorithms are vulnerable to the same criticism as the body of work on low-recourse algorithms: few generic tools are known, and algorithms seem to require ad hoc combinatorial insights. 

In this paper, we attempt to complete the picture by providing building blocks for low-recourse rounding algorithms for down-closed (i.e., packing) problems subject to element insertions/deletions. We call these primitives \emph{Dynamic Contention Resolution Schemes} (DCRS), and they are low-recourse dynamic analogs of the \emph{Contention Resolution Schemes} (CRS) proposed in the context of submodular optimization \cite{DBLP:journals/siamcomp/ChekuriVZ14} and their cousins \emph{Online Contention Resolution Schemes} (OCRS) \cite{feldman16online}, which are used broadly in online algorithms and are especially relevant for Bayesian selection problems.

\subsection{Our Results} 
\label{ssc:our_results}

Our first contribution is the following definition, which we formalize in \cref{sec:prelim_dcrs}.
\begin{definition}[DCRS, informal] 
    For $c \in [0,1]$ and $z \geq 0$, a \emph{$(c, z)$-DCRS}
    is an online algorithm that, given an online sequence of fractional points $\vx^1, \ldots, \vx^T \in [0,1]^\cN$ in a polytope relaxation of a down-closed constraint, performs a time-correlated sample $\Rset^1, \dotsc, \Rset^T$ of these vectors,\footnote{This is a sequence of subsets such that each individual set $\Rset^t$ contains each element $e \in \cN$ with probability $x^t_e$ independently, but the sequence $\Rset^1, \ldots, \Rset^T$ is jointly correlated over time in a natural way.} 
    and then further subselects $\Iset^1 \subseteq \Rset^1, \ldots, \Iset^T \subseteq \Rset^T$ with the properties that:
    \begin{enumerate}
        \item $\Iset^t$ is feasible with probability $1$.
        \item Every element $e \in \cN$, conditioned on being in the sample $\Rset^t$, is subselected into $\Iset^t$ with probability at least $c$.
        \item The total recourse of the sets $\Iset^1, \ldots, \Iset^T$ is at most $z$ times the $\ell_1$ movement of the vectors $\vx^1, \ldots, \vx^T$.
    \end{enumerate}
    We refer to $c$ and $z$ as the \emph{balance} and \emph{recourse} of the DCRS, respectively. 
\end{definition}

With this definition, our main technical contribution is that an $(\Omega(1), O(\log \rank(\cM)))$-DCRS exists for matroid constraints.
\begin{theorem}[Matroid DCRS, informal version of \Cref{thm:mat_dcrs_formal,thm:better-partition-dcrs}]
    \label{thm:mat_dcrs_informal}
    For every $\eps > 0$, there exist $(\nicefrac{1}{4}-\eps, O(\nicefrac{1}{\eps^3}\cdot \log \rank(\cM)))$-DCRSs for matroid constraints.
	If the vectors $\vx^1, \vx^2, \dotsc, \vx^T$ form a coordinate-wise increasing sequence, the recourse of these DCRSs improves to $O(\eps^{-1})$. If $\cM$ is a partition matroid, we show an improved $(1-e^{-1}, 4)$-DCRS.
\end{theorem}

It is known that there are matroids for which there is no $c$-balanced CRS for $c>1-e^{-1}$ even in the offline setting. Thus, it is natural to ask: Which classes of matroids admit CRSs with balance tending to $1$? Recently, Chekuri, Song, and Zhang \cite{DBLP:conf/sosa/ChekuriSZ24} provided one answer to this question. They showed that $k$-fold union matroids (defined formally in \cref{sec:kfold}), which play an important role in combinatorial optimization (see, e.g., \cite{Schrijver-book}), have $(1-o(1))$-balanced CRSs, and this result was later extended to OCRSs by 
Alon, Gravin, Pollner, Rubinstein, Wang,
Weinberg, and Zhan
\cite{DBLP:conf/innovations/AlonGPRWW025}. We show an analogous result for DCRSs.

\begin{theorem}[$k$-fold Union Matroid DCRS, informal version of \Cref{thm:kfold-mat_dcrs_formal}]
    \label{thm:lfold-union_dcrs_informal}
    There exist $(1 - O(\sqrt{\log(k)/k}), O((k/\log(k))^{3/2} \cdot \log \rank(\cM)))$-DCRSs for $k$-fold union matroid constraints. If the vectors $\vx^1, \vx^2, \dotsc, \vx^T$ form a coordinate-wise increasing series, the recourse of these DCRSs improves to $O(\sqrt{k/\log(k)})$.
\end{theorem}

We also show a simple but (arguably) surprising result for knapsack constraints, which adapts ideas from \cite{DBLP:journals/siamcomp/ChekuriVZ14} and \cite{feldman16online} with a twist.
\begin{theorem}[Knapsack DCRS, informal version of \Cref{thm:dcrs-knapsack}]
    \label{thm:knap_dcrs_informal}
    There exist $(\nicefrac{1}{16},\: 2)$-DCRSs for knapsack constraints.
\end{theorem}

Finally, following \cite{DBLP:journals/siamcomp/ChekuriVZ14} and \cite{feldman16online}, we show that DCRSs for different constraints can be generically combined to give DCRSs for the intersection of the constraints.

\begin{theorem}[Combiner Theorem, informal version of \Cref{thm:combiner}]
        Suppose the constraints $\cF_1, \ldots, \cF_h$ admit DCRSs with parameters $(c_1, z_1), \ldots, (c_h, z_h)$. Then, the intersection of these constraints $\bigcap_{i \in [h]} \cF_i$ admits a $(\prod_i c_i, \: \sum_i z_i)$-DCRS.
\end{theorem}

The formal version of \Cref{thm:combiner} is more refined than the simplified version above.
In particular, it implies an $(\Omega(p^{-1}),\: O( \log n))$-DCRS for the intersection of $p$ matroid constraints and a constant number of knapsack constraints. This bound also generalizes to the case of unboundedly many matroid/knapsack constraints, provided that each element participates in at most $p$ matroid and $O(1)$ knapsack constraints (this is known as the intersection of a $p$-matchoid and $O(1)$-column-sparse packing integer program, or PIP).

\subsection{Implications}

Why should one care about DCRSs? As hinted in the introduction, when combined with the dynamic LP solver of \cite{BBLS23}, a DCRS for a down-closed constraint $\cF$ implies a low-recourse dynamic algorithm for maintaining solutions to optimization problems subject to $\cF$ over the course of element insertions and deletions. In fact, we show that we can solve this problem for a very general class of objectives, namely nonnegative monotone submodular functions, even when the objective function changes over time.

\begin{theorem}[DCRS $\Rightarrow$ Low-Recourse Algorithms, informal version of \Cref{thm:low_rec_submod}]
    Let $\cF$ be a  constraint for which there exists a $(c,z)$-DCRS. Then, for any $\beta$, there is an algorithm that gets a sequence of $T$ updates arriving online, where update $t \in [T]$ consists of an arbitrary set $A^t$ of active elements and an arbitrary non-negative monotone submodular function $f^t$. After receiving every update $t$, the algorithm produces an independent set $\Iset^t$ feasible under $\cF$ such that:
    \begin{enumerate}
        \item For all $t \in [T]$, the solution $\Iset^t$ is an approximate maximizer of $f^t$. If $f^t$ is nonnegative monotone submodular then
        \[\expect*{f^t(\Iset^t)} \geq (1-e^{-1}) \cdot c \cdot (1 - \eps) \cdot \beta  \cdot \opt^t,\] 
        and, if $f^t$ is linear, then
       \[\expect*{f^t(\Iset^t)} \geq c \cdot (1 - \eps) \cdot \beta  \cdot \opt^t.\]       
        \item The expected recourse is
        $O\left(\frac{z}{\eps} \log \frac{n}{\eps}\right) \optr^\beta$.
        \end{enumerate}
        Here, $\opt^t$ is the maximum value of $f^t$ for a feasible subset of $A_t$ and $\optr^\beta$ is the minimum recourse of any algorithm maintaining $\beta$-approximation.
\end{theorem}

\input{diagram.tex}

The relax-and-round pipeline underlying the general-purpose competitive-recourse dynamic algorithm from the last theorem is illustrated in \Cref{fig:framework_horizontal}. 
\Cref{tab:app} summarizes some specific results obtained by combining this last theorem with our DCRSs. In the following, we discuss several ways in which this theorem and the results that follow from it improve upon existing work.

\paragraph{Competitive Recourse.} Most work on low-recourse dynamic algorithms is concerned with \emph{absolute} recourse bounds. These are guarantees of the form ``After $T$ data point insertions or deletions, the algorithm incurs at most $c \cdot T$ recourse.'' Our work on the other hand extends the so-called \emph{competitive} recourse bounds championed by \cite{BBLS23}, which are refined guarantees of the form ``Over the course of the input sequence, the algorithm incurs recourse that is at most $c \cdot \optr$.'' Here $\optr$ is the minimal recourse of \emph{any} offline algorithm with full foreknowledge of the input that maintains an optimal solution at all times.\footnote{There are a small number of works prior to \cite{BBLS23} in which competitive recourse guarantees are implicit, e.g., \cite{brodal1999dynamic,bera2022new,azar2023competitive,avin2016online,avin2020dynamic}.} More generally, for any $\beta\leq 1$, we can compare our algorithms' recourse to $\optr^{\beta}$, which is the optimal recourse of an offline algorithm that maintains a $\beta$-approximate solution at all times. 

\paragraph{New Settings.} Our results extend to hitherto unstudied settings. For example, ours is the first low-recourse algorithm for knapsack constraints! One reason this problem has been overlooked is that non-trivial absolute recourse algorithms for fully-dynamic knapsack simply do not exist.\footnote{Consider a situation in which there is one item of size $1$ and value $C \gg 1$, and there are $1/\eps$ items of size $\eps$ and total value $1$. If the large element is toggled on the algorithm must hold it instead of small items, but if it is toggled off the algorithm must evict it and fetch at least $O(1/\eps)$ small items. Repeated toggling forces arbitrarily large recourse.} We circumvent this issue by studying \emph{competitive} recourse algorithms for this setting, and we are the first to do so.

We generalize the \matching result of \cite{BBLS23}\footnote{\cite{BBLS23} also study the \mst problem, but the comparison to our work is not apt: their goal is to find a min-weight matroid base (specifically a spanning tree), while we aim for a max-weight independent set.} to intersections of arbitrary matroids and even $p$-matchoids. In particular, we obtain the first competitive-recourse algorithms for (non-bipartite) \gmatching. We also get competitive-recourse algorithms for dynamically maintaining large arborescences, which can be represented as intersections of graphic and partition matroids. This problem was recently studied in the random order model by Dahlmeier and Hershkowitz \cite{DBLP:conf/waoa/DahlmeierH25}. It should also be noted that the result of \cite{BBLS23} was only for maximum \emph{cardinality} \matching, and did not extend to general linear objectives, let alone monotone submodular functions.

Competitive-recourse algorithms for submodular optimization were studied only very recently by Buchbinder, Naor and Wajc \cite{buchbinder2025chasingsubmodularobjectivessubmodular} for the special case of partition matroids. They achieve a tighter $(1-e^{-1})$-approximation for this case by elegantly combining body-chasing and rounding (we guarantee only $(1-e^{-1})^2$-approximation), but since their approach does not go through contention resolution schemes, it does not immediately generalize to matching.\footnote{To implement our algorithm in polynomial time without losing additional constant factors, we use a result from \cite{buchbinder2025chasingsubmodularobjectivessubmodular} to (approximately) separate  the relevant polytopes.} 
Several papers have explicitly considered absolute-recourse submodular optimization over cardinality constraints~\cite{DBLP:conf/icml/DuettingFLNZ24,DBLP:conf/stoc/DuttingFLNSZ25}. Absolute recourse bounds for dynamic submodular maximization also follow from low-update time dynamic algorithms when each time step involves a bounded number of updates to the solution (e.g. for cardinality~\cite{DBLP:conf/stoc/ChenP22,DBLP:conf/nips/LattanziMNTZ20} and matching~\cite{DBLP:conf/icalp/BanihashemBGHJM25} constraints), and from some online/streaming algorithms (e.g.\ for matroid  constraints~\cite{DBLP:journals/mp/ChakrabartiK15,DBLP:conf/icalp/ChekuriGQ15,DBLP:conf/icalp/FeldmanLNSZ22,DBLP:journals/mor/HarshawKFK22}).

\paragraph{Better bounds.} Besides significantly generalizing \cite{BBLS23}, we also improve several of their parameters, even in the special case of max-cardinality \matching that they consider. Specifically, they require $O(\log^2 n)$-competitive recourse, and in addition, they lose also a small inverse polynomial \emph{absolute} recourse term. We improve the recourse guarantee to $O(\log n)$, and we remove the additive term, yielding a true competitive recourse bound (we do however suffer a worse approximation ratio of $(1-e^{-1})^2 - \eps \approx 0.4$, while \cite{BBLS23} show a $(1-\eps)$-approximation). 

\begin{table}
    \centering
    \footnotesize{
    \def\arraystretch{1.3}

    \begin{tabular}{|>{\centering}m{0.20\textwidth}|>{\centering}m{0.15\textwidth}|>
    {\centering}m{0.18\textwidth}|>{\centering}m{0.18\textwidth}|} 
        \hline 
        \textbf{Problem} & \textbf{Approximation} & \textbf{Competitive Recourse} & \textbf{Reference}
        \tabularnewline
        \hline 
        \hline 
        Matroid & $\frac{1}{4}-\eps$ & $O\big(\frac{1}{\eps^{4}}\log^2\frac{n}{\eps}\big)$ & \Cref{thm:mat_dcrs_formal}
        \tabularnewline
        \hline 
        Partition Matroid & $1-e^{-1}-\eps$ & $O\big(\frac{1}{\eps}\log\frac{n}{\eps}\big)$  & \Cref{thm:better-partition-dcrs}
        \tabularnewline
        \hline 
        Matching & $(1-e^{-1})^2-\eps$ & $O\big(\frac{1}{\eps}\log\frac{n}{\eps}\big)$ & \Cref{thm:better-partition-dcrs,thm:combiner}
        \tabularnewline
        \hline 
        Knapsack & $\frac{1}{16} - \eps$ & $O\big(\frac{1}{\eps}\log\frac{n}{\eps}\big)$ & \Cref{thm:dcrs-knapsack}
        \tabularnewline
        \hline 
        $p$-Matchoid and $O(1)$-sparse packing PIP & $\Theta(1/p)$ & $O(p^3 \log^2 n)$ & \Cref{thm:mat_dcrs_formal,thm:dcrs-knapsack,thm:combiner}
        \tabularnewline
        \hline 
        Matching in $p$-Hypergraph & $\Theta(1/p)$ & $O(\log n)$ & \Cref{thm:better-partition-dcrs,thm:combiner}
        \tabularnewline
        \hline
        $k$-Fold Union Matroid & $1 - O\Big(\sqrt{\frac{\log k}{k}}\Big)$ & $O\Big(\!\left(\frac{k}{\log k}\right)^{3/2} \log^2 n\Big)$ & \Cref{thm:kfold-mat_dcrs_formal}
        \tabularnewline
        \hline 
    \end{tabular}
    }
    {\footnotesize{}\caption{\label{tab:app}
        Summary of our fully dynamic algorithms with competitive recourse, obtained by composing the relevant DCRS with \Cref{thm:low_rec_submod}. All results are for linear objective functions, but translate to monotone submodular functions with an additional factor $(1-e^{-1})$ loss in the approximation.
    }
    }{\footnotesize \par}
\end{table}

\subsection{Techniques and Outline}

We begin in \Cref{sec:matroid} with our main result, which is a construction of DCRSs for matroid constraints. Our jumping-off point is the chain construction of Feldman, Svensson, and Zenklusen~\cite{feldman16online}, originally designed for \emph{online} contention resolution. At high level, their idea is to recursively ``protect'' elements that are very likely to be spanned (the CRS requirement is that every element should have good probability to appear in the output solution if it is active), and then to take protected elements with higher priority over non-protected elements. The result is a chain decomposition of the ground set $\cN = N_0 \supsetneq N_1 \supsetneq \ldots \supsetneq N_\ell$, and given the chain, 
the (O)CRS can greedily add active elements to each level independently; the chain guarantees independence with respect to $\cM$ regardless of how each level is resolved.
The main contribution of Feldman et al.\ is to show that greedily protecting elements in danger of being spanned is a strategy that terminates; by virtue of this termination, every element is unlikely to be spanned if active, and thus the desired balance guarantee holds.

It is natural to try to maintain a similar chain dynamically in way that is stable with respect to changes to the input vector $\vx$, and indeed this is what we do. Unfortunately, off-the-shelf, the standard construction fails in a few ways. For one, small changes in the input vector $\vx$ can force significant changes to the chain. Second, the number of layers of the chain may be large a priori, and hence moving a single element in the chain can trigger a cascade of many other element moves. Our first natural remedy is to add ``friction'' to the chain dynamics: we simply increase the threshold for protecting elements above the strictly necessary parameter chosen by Feldman et al. This forces the expected rank of subsequent levels in the chain to shrink geometrically, and hence the number of levels is always logarithmically bounded.\footnote{
This geometric behavior was noted by Feldman, Svensson, and Zenklusen \cite{DBLP:conf/soda/FeldmanSZ26}, and arguably exposes some intuition as to why the original construction of \cite{feldman16online} terminates.} This change alone already suffices to imply low recourse if the input vector $\vx$ is assumed to be monotonically increasing: in this case elements only move to higher levels, and the friction mechanism forces the average height of elements of representative bases for the chain levels to be constant.

Nevertheless, when $\vx$ can increase and decrease arbitrarily, we cannot assume elements only graduate to higher levels, and the story is significantly more challenging. Our solution is to periodically rebuild the chain above a given level once the rank of this level deviates significantly from the $\vx$-mass restricted to that level (i.e.\ the expected size of the active set restricted to the level). Each rebuild incurs recourse that is proportional to the rank of the level being rebuilt, and hence the crux of the argument is to charge such rebuilds to $\ell_1$ movement of the $\vx$ vector. To do this, we show that to trigger a rebuild at level $\ell$, the $\vx$-vector restricted to level $\ell$ must have incurred \emph{downward} $\ell_1$ movement on the order of its rank. Since downward movement in a single coordinate contributes to the rebuild-potential of up to logarithmically many levels containing it, this is where we lose a logarithmic term in our recourse bounds. A final minor difficulty we have, in contrast to the original OCRS setting of Feldman et al., is that we cannot extract an independent set from our chain using a naive memoryless strategy, because this might incur too much recourse. Instead, we have to manually translate edits to the chain into edits to the independent set.  

In \Cref{sec:partition} we show a simplified construction with improved $O(1)$ competitive recourse for  partition matroids (which we can then translate to $O(1)$-recourse DCRS for Matching using \cref{sec:combine}). In \cref{sec:kfold} we provide finer guarantees for $k$-fold union matroids. The paper of Chekuri, Song, and Zhang~\cite{DBLP:conf/sosa/ChekuriSZ24} indirectly provides a CRS for this class by bounding the \emph{correlation gap};\footnote{\cite{DBLP:journals/siamcomp/ChekuriVZ14} showed that a correlation gap of $\alpha$ implies the existence of an $\alpha$-balanced CRS, and vice versa.} however, since the implicit CRS is constructed through LP duality, it seems challenging to adapt this CRS to the dynamic setting. Instead, we build on the work of \cite{DBLP:conf/innovations/AlonGPRWW025} showing OCRSs for $k$-fold uniform matroids, which in turn modified the chain construction OCRS of \cite{feldman16online} to use what they name the  \emph{generalized occupancy indicator} function in lieu of the rank function. We can adapt our chain construction analogously, in which case balance guarantees follow directly from the bicriterion concentration inequality argument of \cite{DBLP:conf/innovations/AlonGPRWW025}. The outstanding difficulty is once again to bound the recourse, and we do so by generalizing our analysis from \cref{sec:matroid} to use the occupancy function. 

Next, in \cref{sec:knapsack} we provide DCRSs for knapsack constraints. The CRS of \cite{DBLP:journals/siamcomp/ChekuriVZ14} processes elements in fixed, \emph{decreasing-size} order, accepting each active element into the solution so long as it fits. This seems promising to use off-the-shelf, since unlike the chain construction for the matroid case which must depend on the $\vx$ vector (see \cite{DBLP:conf/sosa/FuLTTWW022}), this ordering for the knapsack CRS is fully oblivious to the $\vx$-vector. Unfortunately, we have traded one ill for another: a single large item entering/leaving can now trigger the eviction/reloading of an unbounded number of small items.\footnote{By contrast, the matroid case becomes easy once the chain is fixed, because the matroid exchange axiom guarantees the output set changes by at most $2$ elements per element change to the active set.} To remedy the situation, unlike \cite{DBLP:journals/siamcomp/ChekuriVZ14}, we use \emph{increasing}-size order. This ensures that every element that enters/leaves the active set only triggers a single extra edit. To guarantee the balance of the CRS, we also need a randomized filtering step, which also appears in \cite{feldman16online}.

In \cref{sec:combine}, we show how to combine monotone DCRSs of different constraints to get DCRSs for intersections of these constraints. The proof of this result reuses arguments from the analogous result of 
\cite{DBLP:journals/siamcomp/ChekuriVZ14} for CRSs, but it has to pay special heed to the recourse. In \cref{sec:submod}, we formally prove the application of DCRSs to dynamic optimization problems with submodular objectives. The algorithm for this application uses the framework of~\cite{BBLS23} to produce low-recourse fractional solutions online and then feeds these to the DCRS rounding algorithm. Finally, though it is not a primary focus of our paper, 
in \cref{sec:adversary-lb} we demonstrate strong impossibilities for non-trivial competitive recourse bounds against adaptive adversaries,
and in \cref{sec:polytime,sec:round-or-sep} we argue that our algorithms can be adapted to run in polynomial time. 

\subsection{Other Related Work}

There is by now a long line of work on absolute recourse algorithms for a host of combinatorial optimization problems such as Steiner tree under terminal updates \cite{DBLP:journals/siamdm/ImaseW91,DBLP:journals/siamcomp/GuG016,DBLP:conf/soda/GuptaK14,lkacki2015power,DBLP:conf/focs/GuptaL20}, load balancing under job arrivals \cite{DBLP:journals/jcss/AwerbuchAPW01,GKS14,KLS23}, set cover under element insertions/deletions \cite{GKKP17,DBLP:conf/stoc/AbboudA0PS19,  DBLP:conf/focs/BhattacharyaHN19,DBLP:conf/focs/GuptaL20, DBLP:conf/soda/BhattacharyaHNW21,DBLP:conf/esa/AssadiS21}, facility location under client updates \cite{BLP22,guo2020facility}, and many fully dynamic graph problems under edge updates like edge orientation \cite{brodal1999dynamic,sawlani2020near,bera2022new}, graph coloring \cite{solomon2020improved}, maximal independent sets \cite{assadi2018fully}, and spanners \cite{baswana2012fully,bhattacharya2022simple}. In very recent work, Buchbinder, Naor, and Wajc
\cite{buchbinder2025chasingsubmodularobjectivessubmodular} show competitive recourse algorithms for maximizing a submodular function subject to partition matroid constraints, a special case of our main application: their focus is on polynomial time (approximate) separation oracles for submodular constraints, and they do not study dynamic contention resolution. 

Our DCRS is the dynamic relative of the CRS introduced by 
\cite{DBLP:journals/siamcomp/ChekuriVZ14} in the context of offline submodular maximization, and of the OCRS introduced by
\cite{feldman16online} in the context of Bayesian selection. In nearby work to ours,
\cite{DBLP:conf/sosa/FuLTTWW022} asked whether there exist \emph{oblivious} CRS/OCRS, for which the decisions of which elements of the active set $R$ to subselect are independent of the underlying vector $\vx$ from which $R$ is sampled. Unfortunately, their answer is negative even for the special case of graphic matroids. One can view our DCRSs and associated recourse guarantees as a relaxation of the obliviousness requirement, or as a quantitative upper bound on how non-oblivious CRSs in such settings must be. Finally,
\cite{avadhanula23fully} study a dynamic Bayesian selection problem related to OCRS with a title similar to ours; however the resemblance is only in name.

%% file: diagram.tex
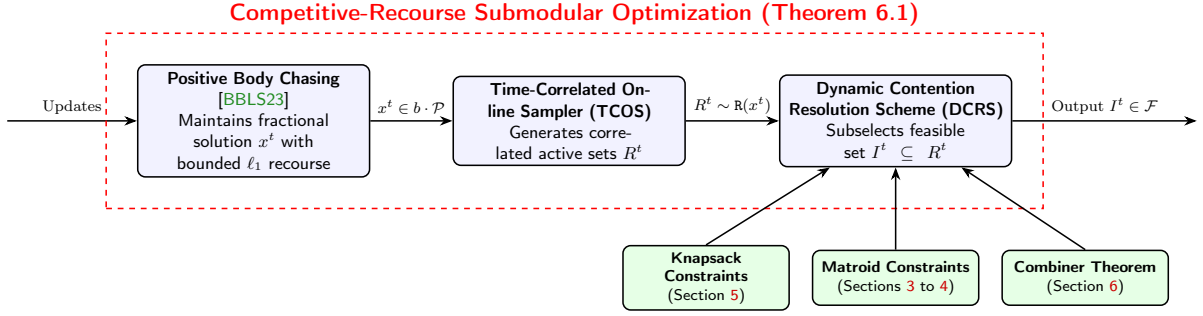
\begin{figure}
\centering
\resizebox{\textwidth}{!}{
\begin{tikzpicture}[
    >=Stealth,
    block/.style={rectangle, draw=black, thick, fill=blue!5, text width=4cm, align=center, rounded corners, minimum height=1.5cm, font=\sffamily\small},
    contrib/.style={rectangle, draw=black, thick, fill=green!10, text width=2.8cm, align=center, rounded corners, minimum height=1cm, font=\sffamily\footnotesize},
    arrow/.style={->, thick}
]

\node[block] (bbls) {\textbf{Positive Body Chasing} \\ {\cite{BBLS23}} \\ Maintains fractional \\ solution $x^t$ with bounded $\ell_1$ recourse};

\draw[arrow] ++(-4.5cm,0) -- node[above, font=\footnotesize] {Updates} (bbls);

\node[block, right=1.5cm of bbls] (tcos) {\textbf{Time-Correlated Online Sampler (TCOS)} \\ Generates correlated active sets $R^t$};

\node[block, right=1.7cm of tcos] (dcrs) {\textbf{Dynamic Contention Resolution Scheme (DCRS)} \\ Subselects feasible \\ set $I^t \subseteq R^t$};

\draw[arrow] (bbls) -- node[above, font=\footnotesize] {$x^t \in b \cdot \mathcal{P}$} (tcos);
\draw[arrow] (tcos) -- node[above, font=\footnotesize] {$R^t \sim \Rdist(x^t)$} (dcrs);
\draw[arrow] (dcrs) -- node[above, font=\footnotesize] {Output $I^t \in \mathcal{F}$} ++(5.5cm,0);

\node[draw=red, thick, dashed, inner sep=16pt, fit=(bbls) (tcos) (dcrs), label={[red, font=\sffamily\bfseries\large]above:Competitive-Recourse Submodular Optimization (Theorem 6.1)}] (framework) {};

\node[contrib, below=1.5cm of dcrs] (matroid) {\textbf{Matroid Constraints} \\ (\cref{sec:matroid,sec:partition,sec:kfold})};
\node[contrib, left=0.4cm of matroid] (knapsack) {\textbf{Knapsack Constraints} \\ (\cref{sec:knapsack})};
\node[contrib, right=0.4cm of matroid] (combiner) {\textbf{Combiner Theorem} \\ (\cref{sec:combine})};

\draw[arrow] (matroid.north) -- (dcrs.south);
\draw[arrow] (knapsack.north) -- ([xshift=-1.2cm]dcrs.south);
\draw[arrow] (combiner.north) -- ([xshift=1.2cm]dcrs.south);

\end{tikzpicture}
}
\caption{The pipeline for competitive-recourse dynamic submodular optimization. }
\label{fig:framework_horizontal}
\end{figure}

%% file: prelim.tex
\section{Preliminaries}
\label{sec:prelim_dcrs}

For brevity, we provide an abridged preliminaries section here. Readers familiar with matroids, submodular functions and contention resolution can safely start here; note only that we assume without loss of generality that the matroids we consider are free of self-loops. 
We give extended background material in \cref{sec:full_prelim}. 

\paragraph{Basic Notation.}

We assume in this paper that we are faced with an \emph{oblivious} online adversary that fixes the input sequence \emph{prior} to the realizations of the active sets $\Rset^1, \ldots, \Rset^T$. This is in contrast to the so-called \emph{almighty} adversary that Feldman at al.\ assume in the OCRS context \cite{feldman16online}. In \cref{sec:adversary-lb}, we show no nontrivial results are possible even against an adaptive adversary that at time $t$ has access to the outcomes at times $1,\ldots, t-1$.

\paragraph{Time-Correlated Samplers.} A CRS picks a feasible subset of a random set $\Rset$ distributed like $\Rdist(\vx)$ (we denote such a set by $\Rset \sim \Rdist(\vx)$). Our DCRS task is to choose feasible subsets of a sequence of $\Rset^t$, each of which is individually distributed like $\Rdist(\vx^t)$. 
While there is only one way to independently sample $\Rset$ according to element marginals $\vx^t$, there are many ways of constructing sequences of such samples that are correlated across time (but, for any fixed $t$, independent across elements). We abstract the process of producing this sequence of $\Rset^t$ as follows.

\begin{definition}[Time-correlated online sampler]
\label{def:tcos}
    An \emph{$\alpha$-time-correlated online sampler ($\alpha$-TCOS)} is a randomized online algorithm that, given a sequence of vectors $\vx^1, \ldots, \vx^t \in [0,1]^\cN$ arriving online, upon arrival of each $\vx^t$ produces a subset $\Rset^t \subseteq \cN$ such that
    \begin{enumerate}
        \item $\Rset^t$ includes each $e \in \cN$ independently with probability $x^t_e$, 
        \item $\expect*{\abs{\Rset^t \symdif \Rset^{t-1}}} \leq \alpha \cdot \norm{\vx^t - \vx^{t-1}}_1$,
    \end{enumerate}
    where we fix $\vx^0 \defeq 0$ and use $\symdif$ to denote the symmetric difference between sets. 
\end{definition}
It is important that we have some sort of time-correlated sample if we hope to achieve a low recourse CRS---note that rounding $\vx^t$ to $\Rset^t$ independently for each $t \in [T]$ is not a TCOS for any $\alpha$.
\cite{BBLS23} use the following simple 1-TCOS to convert fractional competitive-recourse algorithms to integral ones:
\begin{example}[Threshold TCOS]
\label{ex:thresh-TCOS}
    At the outset, sample a uniform random threshold $\lambda_e \sim U[0,1]$ for each $e \in \cN$.
    Then, for each arriving vector $\vx^t$, let $\Rset^t = \{e \mid x^t_e \geq \lambda_e\}$.
\end{example}
An alternative 1-TCOS makes memoryless decisions:
\begin{example}[Markovian TCOS]
\label{ex:markov-TCOS}
    Start by setting $\Rset^0 \leftarrow \emptyset$. 
    For each subsequent $\vx^t$, let $\delta_e \defeq x^t_e - x^{t-1}_e$. 
    If $\delta_e \geq 0$ and $e \in \Rset^{t-1}$ then include $e$ in $\Rset^t$; if $e \not\in \Rset^{t-1}$ then include $e$ in $\Rset^t$ with probability $\frac{\delta_e}{1-x_e^{t-1}}$.
    If $\delta_e < 0$ and $e \not \in \Rset^{t-1}$, then $e \not \in \Rset^t$; if $\delta_e < 0$ and $e \in \Rset^{t-1}$, then $e \not \in \Rset^t$ with probability $\frac{-\delta}{x_e^{t-1}}$.
\end{example}
Note that, as $1$-TCOSs, these are both instance-optimal in terms of expected $\{\Rset^t\}$ recourse:
\begin{observation}
\label{obs:TCOS-1-impossibility}
    For any TCOS and any sequence $\{\vx^t\}_t$, $\expect{\abs{\Rset^t \symdif \Rset^{t-1}}} \geq \norm{\vx^t - \vx^{t-1}}_1$.
\end{observation}
\begin{proof}
    This follows from Property 1 of TCOSs, that each $\Rset^t \sim \Rdist(\vx^t)$. 
    Suppose $\delta = \abs{x_e^t - x_e^{t-1}}>0$; then $\abs{\prob{e \in \Rset^t} - \prob{e \in \Rset^{t-1}}} = \delta$ and even an optimal coupling of $[e \in \Rset^{t-1}]$ and $[e \in \Rset^{t}]$ has total variation distance $\delta$, and so $\prob{e \in \Rset^{t-1} \symdif \Rset^{t}} \geq \delta$.
    Summing over $e\in \cN$ yields the claim.
\end{proof}

\paragraph{Dynamic Contention Resolution Schemes.} We formally refer to a constraint as the collection $\cF \subseteq 2^{\cN}$ of sets that are feasible under this constraint. For example, given a matroid $(\cN, \cI)$, the corresponding constraint is $\cF = \cI$. Given this notation, we are ready to define our main object of study.

\begin{definition}[DCRS]
\label{def:dcrs}
    For $b, c \in [0,1]$ and $z \geq 0$, a $(b, c, z)$-DCRS $\pi$ for a down-closed constraint $\cF \subseteq 2^{\cN}$ and a polytope $P \subseteq [0, 1]^\cN$ relaxing $\cF$ is an online algorithm of the following form. In time step $t \in [T]$, the DCRS $\pi$ receives a point $\vx^t \in b \cdot P$ and a time-correlated sample $\Rset^t$ of $\vx^t$. The DCRS then outputs a subset $\Iset^t \subseteq \Rset^t$ with the properties:
    \begin{enumerate}
        \item $\Iset^t \in \cF$ with probability $1$.
        \item $\prob{e \in \Iset^t} \geq c \cdot x_e^t$ for every element $e \in \cN$ and time $t \in [T]$.
        \item $\sum_{t \in [T]} \expect*{\abs{\Iset^t \triangle \Iset^{t-1}}} \leq z \cdot \sum_{t \in [T]} \expect*{\abs{\Rset^t \triangle \Rset^{t-1}}}$.
    \end{enumerate}
	Notice that $\Iset^t$ might depend on the inputs of the DCRS from previous time steps. Therefore, we sometimes refer to the output $\Iset^t$ of the DCRS $\pi$ at time step $t$ as $\pi_{\vx^{\leq t}}(\Rset^{\leq t})$, where $\vx^{\leq t} \defeq (\vx^1, \vx^2, \dotsc, \vx^t)$ and $\Rset^{\leq t} \defeq (\Rset^1, \Rset^2, \dotsc, \Rset^t)$. It is sometimes also useful to treat the last set in this vector separately. Thus, we sometimes use the notation $(\Rset^{\leq t-1}, \Rset^t)$ to refer to the vector $\Rset^{\leq t}$.
\end{definition}

\begin{remark}
    Our definitions intentionally decouple the choice of TCOS from the DCRS. For one, this gives the DCRS the flexibility to work with an exogenous time-correlated sampler that the designer may not control. Secondly, sometimes one may prefer one TCOS over another for subtle reasons. For example, though both the Markovian TCOS and the Threshold TCOS achieve the optimal recourse parameter $\alpha = 1$, the recourse of the Markovian TCOS has significantly smaller variance. We discuss details in \cref{app:tcos}.
\end{remark}

Given a $(b, c, z)$-DCRS $\pi$, for every time step $t$, the random function $\pi_{\vx^{\leq t}}(\Rset^{\leq t - 1}, \cdot)$ can be viewed as a CRS. A property we use multiple times in this paper is that this CRS is $(b,c)$-balanced for every time step $t$. Accordingly, we sometimes refer to a $(b,c,z)$-DCRS as \emph{$(b,c)$-balanced with recourse $z$}. Another useful property of CRSs is monotonicity. A CRS $\pi'$ is monotone if for every fixed vector $\vx \in [0, 1]^\cN$, and any $e \in \Rset \subseteq \Rset' \subseteq \cN$, it holds that $\prob{e \in \pi'_\vx(\Rset)} \geq \prob{e \in \pi'_\vx(\Rset')}$. Using the same logic as before, we say that the DCRS $\pi$ is \emph{monotone} if $\pi_{\vx^{\leq t}}(\Rset^{\leq t - 1}, \cdot)$ is a monotone CRS for every time step $t$.

We complete this section by observing that, just as for offline CRS, one can use a $(b,c,z)$-DCRS of a polytope $P$ to round points $\vx \in P$ by downsampling $\Rset^t$ with probability $b$ (i.e. for each $e$, dropping $e$ from all $\Rset^t$ with probability $1 - b$), which yields a set distributed like $\Rdist(\vy)$ for the vector $\vy = b \cdot \vx \in b P$. 
To summarize:

\begin{remark} 
\label{rmk:remove_b}
    Any $(b,c,z)$-DCRS can be converted to a $(1,bc,bz)$-DCRS.
\end{remark}

%% file: matroid_new.tex
\section{DCRS for Matroid Constraints} \label{sec:matroid}

Our goal in this section is to prove the following main theorem:
\begin{restatable}[Matroid DCRS]{theorem}{matroiddcrsformal}
\label{thm:mat_dcrs_formal}
    For any $b \in (0,1)$ and any $\eps > 0$, every matroid $\cM$ admits a monotone $(b, 1-b - \eps, O(\nicefrac{1}{\eps^3}\cdot \log \rank(\cM)))$-DCRS.
    If the vectors $\vx^1, \vx^2, \dotsc, \vx^T$ form a coordinate-wise increasing sequence, the recourse of this DCRS improves to $O(\eps^{-1})$.
\end{restatable}

In other words, we want an algorithm that over all steps $t \in [T]$ maintains a dynamic random independent subset $\Iset^t$ of $\Rset^t$ with a small expected number of changes with respect to $\sum_{t = 1}^T \expect{\abs{\Rset^t \symdif \Rset^{t-1}}}$, where each $\Rset^t$ is distributed according to $\Rdist(\vx^t)$.
We will generally think of these $\{\Rset^t\}$ as coming from an $\alpha$-TCOS for some constant $\alpha \geq 1$, in which case the expected change to $\Iset^t$ is small with respect to the change in $\{\vx^t\}$. 

We maintain two collections of sets.
\begin{itemize}
    \item The first collection is a chain construction $\mathcal{S}$. For convenience, we will track a \emph{representative basis} $S_i$ for each level $i$ of the chain, such that our chain decomposition is
    \[ \emptyset = \Span(S_{\ell+1}) \subseteq \Span(S_\ell) \subseteq \Span(S_{\ell - 1}) \subseteq \hdots \subseteq \Span(S_1) \subseteq \Span(S_0) = \cN.\] 
    We assume $S_0$ is an arbitrary fixed base of the matroid. We will later show that the index $\ell$ of the last non-empty level is always $O(\log \rank(\cN))$. 
    \item The second collection is $\Iset_0 \subseteq \Span(S_0) \setminus \Span(S_1), \Iset_1 \subseteq \Span(S_1)\setminus  \Span(S_2), \dotsc, \Iset_{\ell} \subseteq \Span(S_\ell)$, such that $\Iset_i \cup S_{i + 1} \in \cI$.\footnote{In matroid theory terms, $\Iset_i$ is independent in $\cM/{S_{i+1}}\vert_{\Span(S_i)}$, i.e. the matroid obtained from $\cM$ by contracting $S_{i + 1}$ and restricting to $\Span(S_{i})$.} We can think of these as the pieces of the final independent set our contention resolution scheme outputs, since our algorithm will output $\Iset := \bigcup_i \Iset_i$. 
\end{itemize}
\medskip

Feasibility follows from the fact that $\cS$ is a chain decomposition of $\cM$, and we will need to argue that our construction satisfies \emph{balance}, \emph{monotonicity}, and \emph{low recourse}.

\subsection{Warm-Up: the Increase-Only Case}

\label{sec:warmup}

We first give a simplified algorithm for the special case that $\vx$ only increases coordinate-wise.
For simplicity, we will let the simplified algorithm also contain a $1$-TCOS within it, meaning it gets the $\{\vx^t\}$ online and samples the sets $\{\Rset^t\}$ itself. 

The algorithm maintains the chain $\cS$ as follows. Originally, all $\{S_i\}_i$ and $\{I_i\}_i$ are empty, except for $S_0$. After every update step $\vx \gets \vx^t$, our algorithm updates $\mathcal{S}$ by pushing elements ``up'' a level when they have a high probability of being spanned (as determined by the parameter $\tau \in (b, 1)$), and then updates $\Iset$ by accounting for both changes to $\mathcal{S}$ and changes to $\Rset$. Given this process, the output is always determined to be $\Iset = \bigcup_{i} \Iset_i$.

For convenience, define $\vx(i) := \vx\vert_{\Span(S_i)}$ to be the part of $\vx$ relevant to level $i$.

\begin{algorithm}
	\caption{\textsc{Increase-Only Matroid DCRS}}
	\label{alg:increase-only}
	\begin{algorithmic}[1]
        \Require matroid $\cM$, $b \in (0,1)$, $\tau \in (b,1)$, online sequence $\{\vx^t\}$.
		\State{Initialize $S_i \gets \emptyset$ for all $i \geq 1$ and $\Iset_i \gets \emptyset$ for all $i \geq 0$ (recall that $S_0$ is a base of $\cM$).}
        \State{Sample $\lambda_e \sim \mathrm{Unif}([0,1])$ for all $e \in \cN$.}
        \For{\textbf{each} $t = 1, 2, \ldots, T$}
            \State{$\vx \gets \vx^t \in b \cdot P_{\cM}$ arrives online.} 
            \State{$\Rset \gets \{e \in \cN : \lambda_e \leq x_e\}$. \label{line:inc-Rset-def}}
            \State $i \leftarrow 1$.
            \While{$|S_{i-1}| > 0$}
            \algorithmiccomment{\textit{Move likely-spanned elements up to $S_i$, update $\Iset_{i-1}$ appropriately.}}
                \While{there is $e \in \Span(S_{i - 1}) \setminus \Span(S_i)$ obeying $\prob{e \in \Span(\Rdist(\vx(i-1)) \cup S_{i})} \geq \tau$\label{line:inc-add_condition}} 
                    \State Add $e$ to $S_{i}$ and remove $\Iset_{i-1} \leftarrow \Iset_{i-1} \setminus \{e\}$.
                    \If{$S_i \cup \Iset_{i-1}$ is no longer independent}
                        \State{There exists $e' \in \Iset_{i-1}$ in the unique cycle in $S_i \cup \Iset_{i-1}$; remove $e'$ from $I_{i-1}$. \label{line:inc-I-removal}} 
                    \EndIf
                \EndWhile
            \State $i \leftarrow i+1$.
            \EndWhile
            \State $i \leftarrow 1$
            \While{$|S_{i-1}|>0$}
                \algorithmiccomment{\textit{Add new elements from $\Rset$ to $\Iset_i-1$}.}
    	        \While{there is $e \in \Rset \cap (\Span(S_{i-1}) \setminus (\Span(S_{i} \cup \Iset_{i-1})))$}
                    \State Add $e$ to $\Iset_{i-1}$ \label{line:inc-I-add}.
                \EndWhile
            \State $i \leftarrow i+1$.
            \EndWhile
            \State \Return $\bigcup_{i} \Iset_i$.
        \EndFor
        \end{algorithmic}
\end{algorithm}

Since \Cref{alg:increase-only} is only a warm up, we do not provide the full analysis. In particular, we will save the formal proof of balance (Lemma~\ref{lem:warmup_balance}) for later in the context of the more general algorithm (see \Cref{ssc:balance}). As written, \Cref{alg:increase-only} is also technically non-monotone, and some additional technical details are necessary to make our full algorithm monotone.

\begin{lemma}[Balance] \label{lem:warmup_balance}
    When \Cref{alg:increase-only} terminates, the probability that element $e$ in level $i$ is included in $I_i$ (and therefore also in the final independent set) given that $e$ is active is at least $1 - \tau$. 
\end{lemma}

The rest of this section is devoted to bounding the expected recourse of \Cref{alg:increase-only}.
The following lemma establishes a key structural property we need for that purpose.

\begin{lemma} 
\label{lem:growth_bound}
    Given a vector $\vy \in [0, 1]^\cN$, consider a process that builds a set $A\subseteq \cN$ by starting with $A = \emptyset$ and adding elements to $A$ one-at-a-time. If every $e\in\cN$ added to $A$ obeys both $e \not \in \Span(A_e)$ and $\prob{e \in \Span(\Rdist(\vy) \cup A_e)} \geq \tau$, where $A_e$ is $A$ immediately before $e$ is considered, then the final size of $A$ is at most $\|\vy\|_1 / \tau$.
\end{lemma}
\begin{proof}
Let us study the evolution of the expectation $\expect{\rank(\Rdist(\vy) \vert A)}$ during the execution of the process. Notice that the value of this expectation at the beginning of the process, when $A = \emptyset$, is $\expect{\rank(\Rdist(\vy))} \leq \expect{\abs{\Rdist(\vy)}} = \norm{\vy}_1$. Furthermore, we claim that the expectation $\expect{\rank(\Rdist(\vy) \vert A)}$ decreases by at least $\tau$ whenever an element is added to $A$. To see that, consider an arbitrary element $e$ added to $A$, and recall that the value of the set $A$ before the addition of $e$ is $A_e$. Thus, the decrease in $\expect{\rank(\Rdist(\vy) \vert A)}$ following the addition of $e$ is
\begin{align*}
	\expect{\rank(\Rdist(\vy) \vert A_e)} -{}& \expect{\rank(\Rdist(\vy) \vert A_e + e)} \\
	={} &
	\rank(A_e + e) - \rank(A_e) + \expect{\rank(\Rdist(\vy) \cup A_e) - \rank(\Rdist(\vy) \cup A_e + e)} \\
	={} &
	1 - \prob{e \not \in \Span(\Rdist(\vy) \cup A_e)}
	=
	\prob{e \in \Span(\Rdist(\vy) \cup A_e)}
	\geq
	\tau
	\enspace,
\end{align*}
where the second equality holds since $e \not \in \Span(A_e)$.

To summarize, $\expect{\rank(\Rdist(\vy) \vert A)}$ is at most $\norm{\vy}_1$ at the beginning of the process, and decreases by at least $\tau$ whenever an element is added to $A$. Therefore, its final value is upper bounded by $\norm{\vy}_1 - \tau \abs{A}$. Since the expectation $\expect{\rank(\Rdist(\vy) \vert A)}$ is always non-negative, this upper bound must be also non-negative, i.e., it must hold that $\norm{\vy}_1 - \tau \abs{A} \geq 0$. The lemma now follows by rearranging this inequality.
\end{proof}

An important corollary of the last lemma is that both the $\vx$-mass and the rank of subsequent levels decrease geometrically.
\begin{corollary}
    \label{cor:geom_dec}
    For all $i \geq 0$ we have
        $\norm{\vx(i)}_1 \leq \frac{b}{\tau} \norm{\vx(i-1)}_1$, as well as $|S_i| \leq \frac{b}{\tau} |S_{i-1}|$.
\end{corollary}

\begin{proof}
Since $\vx$ is coordinate-wise increasing, every element $e$ added to $S_i$ obeyed, before being added, the inequality $\prob{e \in \Span(\Rdist(\bar{\vx}(i-1)) \cup S_{i})} \geq \tau$, where $\bar{\vx}(i - 1)$ is the final value of $\vx(i - 1)$. Thus, applying \Cref{lem:growth_bound} with $\vy = \bar{\vx}(i-1)$ and $A = S_i$, we get 
\[ 
    \norm{\bar{\vx}(i)}_1 \leq b \cdot \rank(\Span(S_{i})) \leq  b \cdot |S_{i}| \leq \frac{b}{\tau} \cdot \norm{\bar{\vx}(i-1)}_1 \leq \frac{b}{\tau} (b\cdot |S_{i-1}|),
\]
where the first and last inequalities hold since $\vx \in b \cdot P_\cM$, and the second to last follows from \cref{lem:growth_bound}.
\end{proof}

To bound now the cumulative recourse of $\{\Iset^t\}_t$, note that the number of inserts and deletes is at most twice the number of inserts. Further, we look at how many chances a given element has in being inserted. For every $i$, one can verify that $\Span(\Iset_{i} \cup S_{i + 1})$ only grows during the execution of \Cref{alg:increase-only}. Therefore, every element $e \in \Span(S_{i})$ can be inserted into $\Iset_{i}$ at most once, and even that insert can only happen when $e$ becomes active at some point, i.e. with probability at most $x_e^T$. Thus the total amortized recourse is bounded by
\begin{align*}
    \sum_{t \in [T]} \expect{\abs{\Iset^t \symdif \Iset^{t-1}}} 
    &\leq 2 \sum_{t \in [T]} \expect{\abs{\Iset^t \setminus \Iset^{t-1}}} \leq 2 \sum_{e \in \mathcal{N}} \sum_{\substack{i: \\ e \in \Span(S_i)}} x^{T}_e = 2 \sum_{i \in [\ell]} \norm{\vx^T(i)} \\
    &\leq 2 \sum_{i=0}^\infty \left(\frac{b}{\tau}\right)^{i} \norm{\vx^T}_1 \leq \frac{2\tau}{\tau-b} \norm{\vx^T}_1
    = \frac{2\tau}{\tau-b} \sum_{t \in [T]} \norm{\vx^t - \vx^{t-1}}_1 \\
    &= \frac{2\tau}{\tau-b} \sum_{t \in [T]} \expect{\abs{\Rset^t \symdif \Rset^{t-1}}}
    \enspace.
\end{align*} 
Here the second-to-last inequality is due to \cref{cor:geom_dec}. 
Setting $\tau = b + \eps$ yields the $O(\eps^{-1})$ recourse bound of the increase-only case.

%% file: fully_dynamic.tex
\subsection{The Fully Dynamic Scheme}

The main challenge in extending \cref{alg:increase-only} to the dynamic setting is that the geometric rank decrease property may become violated when $\vx$ is allowed to decrease. To remedy the situation, we introduce a ``reset'' condition, whereby if any $S_i$ grows too large with respect to $\vx(i-1)$, then we reconstruct from scratch the suffix of the chain starting at $S_i$. The crux is to show that the cost of resets can be charged to downward $\ell_1$ movement of the vector $\vx$. 

After every dynamic update, \Cref{alg:matroid-DCRS} performs two kinds of local moves until a fixed-point is reached. 
The first (\Cref{line:main_S_loop}) reorganizes the chain decomposition $\{S_i\}$. 
This maintenance of the chain depends only on how $\vx$ has changed, and is independent of $\Rset$. 
The second (\Cref{line:main_T_loop}) updates the output based on changes to $\Rset$ and the $\{S_i\}$, and is not directly dependent on $\vx$. 
The changes to $\{\Iset_i\}$, which comprise the recourse of the algorithm, are attributable to either changes in the $\{S_i\}$ or changes in $\Rset$. Both these types of changes can be bounded by $\vx$ movement.

\cref{alg:matroid-DCRS} additionally makes use of an arbitrary order $\sigma$ over the elements (in contrast to \cref{alg:increase-only}). This order serves only a tie-breaking role but is necessary to guarantee monotonicity.

\begin{algorithm}[ht]
	\caption{\textsc{MatroidDCRS}}
	\label{alg:matroid-DCRS}
	\begin{algorithmic}[1]
        \Require matroid $\cM$, $b \in (0,1)$, $\tau \in (b,1)$, $r \in (1, \tau/b)$, online sequences $\{\vx^t\}$ and $\{\Rset^t\}$.
        \State{$\sigma \gets$ arbitrary ordering of $\cN$.}
		\State{Initialize $S_i \gets \emptyset$ and $\Iset_i \gets \emptyset$ for all $i \in \mathbb{Z}_+$ and $\Iset_0 \gets \emptyset$.}
        \For{\textbf{each} $t = 1, 2, \ldots, T$}
            \State{$\vx \gets \vx^t \in b \cdot P_{\cM}$ and $\Rset \gets \Rset^t$ arrive online.}
            \State $i \leftarrow 1$.
            \While{$|S_{i-1}| > 0$\label{line:main_S_loop}}
            \algorithmiccomment{\textit{Move likely-spanned elements to $S_i$, update $\Iset_{i-1}$, \Reset if necessary.}}
                    \If{$|S_i| > r \cdot \|\vx(i-1)\|_1 / \tau$\label{line:early_reset_condition}}
                       \State{\Reset($i$).}
                    \EndIf
                    \While{there is $e \in \Span(S_{i - 1}) \setminus \Span(S_i)$ obeying $\prob{e \in \Span(\Rdist(\vx(i-1)) \cup S_{i})} \geq \tau$ \label{line:add_condition}}
                        \State Add $e$ to $S_{i}$ and remove $\Iset_{i-1} \leftarrow \Iset_{i-1} \setminus \{e\}$.\label{line:add}
                        \If{$S_i \cup \Iset_{i-1}$ is no longer independent\label{line:T_remove_condition}}
                            \State{There exists $e' \in \Iset_{i-1}$ in the unique cycle in $S_i \cup \Iset_{i-1}$; remove $e'$ from $\Iset_{i-1}$.\label{line:T_remove}}
                        \EndIf
    		              \If{$|S_i| > r \cdot \|\vx(i-1)\|_1 / \tau$\label{line:reset_condition}}
                          \State \Reset($i$).
                        \EndIf
                    \EndWhile
                    \State $i \leftarrow i+1$.
            \EndWhile
            \State $i \leftarrow 1$.
            \While{$|S_{i-1}|>0$ \label{line:main_T_loop}}
                \algorithmiccomment{\textit{Add new elements from $\Rset$ to $\Iset_i$}.}
            	\State Update $\Iset_{i-1} \gets \Iset_{i-1} \cap \Rset$. \label{line:non-R_remove}
    	         \While{there is $e \in \Rset \cap (\Span(S_{i-1}) \setminus (\Span(S_{i} \cup \Iset_{i-1})))$ \label{line:T_add_condition}}
                    \State Add the $\sigma$-first such $e$ to $\Iset_{i-1}$.\label{line:add_T}
                \EndWhile
                \State $i \leftarrow i+1$.
            \EndWhile
            \State \Return $\cup_{i = 0}^{\infty} \Iset_i$.  \label{line:return}
        \EndFor
        \Statex
        \Procedure{Reset}{$i$}
            \For{every level $j \geq i$}
                \State Set $S_j \gets \emptyset$. \label{line:reset_removal_S}
                \State Set $\Iset_{j} \gets \emptyset$. \label{line:reset_removal_I}
            \EndFor
        \EndProcedure
        \end{algorithmic}
\end{algorithm}

We begin the analysis of \cref{alg:matroid-DCRS} by showing that it always terminates.

\begin{lemma} 
\label{lem:terminate}
    \cref{alg:matroid-DCRS} must always terminate.
\end{lemma}
\begin{proof} For every time step $t \in [T]$ within the outer for loop, we need to argue that the while loops starting on \cref{line:main_S_loop,line:main_T_loop} terminate. The loop starting on \cref{line:main_T_loop} adds each of the $n$ elements to some $\Iset_i$ at most once, so we focus on understanding the other loop.

We will rely on the following fact from \cref{sec:warmup}. Define $d(i) := \max_j(j - i \ : \ S_j\neq \emptyset)$ (the number of levels above $i$ after which all levels are empty) and consider the evolution of $d(i)$ as a function of the iteration $i$ of the while loop from \cref{line:main_S_loop} and of the state of the chain $S_1, S_2, \ldots$ etc. If the loop ever reaches a level $i \geq 0$ such that $d(i) = 0$, then \cref{lem:growth_bound} guarantees that the loop will terminate. 
This proves the lemma for the case that $\Reset(i)$ is invoked at any point during the algorithm because immediately after this invocation, we indeed have $d(i)=0$. Thus, it remains to handle the case that \Reset is never invoked. 

If \Reset is never invoked and $d(i) >0$, then $d(i)$ decreases by $1$ after each iteration of \cref{line:main_S_loop} because the algorithm only promotes elements from level $i$ to $i+1$, and all other levels remain unchanged during this iteration. Thus, eventually $d(i)$ becomes $0$ also when \Reset is not invoked.
\end{proof}

The following lemma summaries additional important properties of \cref{alg:matroid-DCRS}.

\begin{restatable}{lemma}{correctnessprops}
\label{lem:properties}
\cref{alg:matroid-DCRS} maintains the invariants that, for every $i \geq 1$,
\begin{compactenum}[\bf (a)]
	\item $S_i \in \cI$,\label{item:S_independent}
	\item $\Span(S_i) \subseteq \Span(S_{i - 1})$, \label{item:laminar}
	\item $\Iset_{i-1} \cup S_{i} \in \cI$, and \label{item:I_independent}
	\item $\Iset_{i - 1} \subseteq \Span(S_{i - 1}) \setminus \Span(S_{i})$. \label{item:I_inclusion}
\end{compactenum}
\end{restatable}

We leave the straightforward proof \cref{lem:properties} to \cref{app:proofs}.
Here, we show that this lemma implies that the output set of our algorithm is feasible.

\begin{corollary} \label{cor:feasible}
\Cref{alg:matroid-DCRS} returns an independent subset of $\Rset^t$.
\end{corollary}
\begin{proof}
Recall that \Cref{alg:matroid-DCRS} returns $\Iset = \cup_i \Iset_i$ on \Cref{line:return}.
Parts~\eqref{item:laminar} and~\eqref{item:I_inclusion} of \cref{lem:properties} guarantee together that the loop on Line~\ref{line:main_T_loop} iterates over all the indexes $i$ corresponding to a non-empty set $\Iset_i$. For every such $i$, \cref{line:non-R_remove} removes from $\Iset_i$ every element that does not belong to $R$. Additionally, the condition on \cref{line:T_add_condition} explicitly makes sure that only elements of $R$ are added to $\Iset_i$. Thus, by the end of the iteration of \cref{alg:matroid-DCRS}, $\Iset = \bigcup_{i} \Iset_i$ is a subset of $\Rset = \Rset^t$. The rest of this proof is devoted to showing that $\Iset$ is also independent in $\cM$.

For every $i \geq 0$ and integer $0 \leq i' < i$, Parts~\eqref{item:laminar} and~\eqref{item:I_inclusion} of \cref{lem:properties} guarantee that $\Iset_{i} \subseteq \Span(S_{i}) \subseteq \Span(S_{i' + 1})$. Thus, by the submodularity of the rank function, we have
\[
	\rank(\Iset_i \mid \cup_{i' = i + 1}^{\ell} \Iset_{i'})
	\geq
	\rank(\Iset_i \mid \Span(S_{i + 1}))
	=
	\rank(\Iset_i \mid S_{i + 1}),
\]
where the equality holds since $\rank(A \cup \Span(B)) = \rank(A \cup B)$ for every two sets $A$ and $B$. Using this inequality and writing $\rank(\Iset)$ as a telescoping sum,
\[
	\rank(\Iset)
	=
	\sum_{i = 0}^\infty \rank(\Iset_i \mid \cup_{i' = i + 1}^{\ell} \Iset_{i'})
	\geq
	\sum_{i = 0}^\infty \rank(\Iset_i \mid S_{i+1})
	=
	\sum_{i = 0}^\infty |\Iset_i|
	=
	|\Iset|,
\]
where the second equality holds by Part~\eqref{item:I_independent} of \cref{lem:properties} and the fact that $\Iset_i \cap S_{i+1} = \emptyset$ (which is implied by Part \eqref{item:I_inclusion}), and the last equality holds since Part~\eqref{item:I_inclusion} of \cref{lem:properties} implies that the sets $\Iset_0, \Iset_1, \dotsc, $ etc.\ are disjoint. This completes the proof that $\Iset$ is independent.
\end{proof}

\subsection{Balance} \label{ssc:balance}

In this section we analyze the balance of our algorithm. For every element $e \in \cM$, let us denote by $\lvl(e)$ the maximum $i$ such that $e \in \Span(S_{i})$.

\begin{proposition} 
\label{prop:matroid-DCRS-balance}
    \Cref{alg:matroid-DCRS} has a balance of $1 - \tau$.
\end{proposition}
\begin{proof}
Consider an arbitrary element $e \in \cM$. In iteration $\lvl(e) + 1$ of the loop starting on \cref{line:main_S_loop} of \cref{alg:matroid-DCRS}, the inner loop on \cref{line:add_condition} terminated, which implies that $e$ does not obey the condition of this inner loop. However, the definition of $\lvl(e)$ implies that $e \in \Span(S_{\lvl(e)}) \setminus \Span(S_{\lvl(e) + 1})$, and so the part of the condition that $e$ violates must be the inequality $\prob{e \in \Span(\Rdist(\vx(\lvl(e))) \cup S_{\lvl(e) + 1})} \geq \tau$, and thus 
\begin{align*}
	\prob{e \not \in \Span((R \cap \Span(S_{\lvl(e)}) - e) \cup S_{\lvl(e) + 1})}
	\geq{} &
	\prob{e \not \in \Span((R \cap \Span(S_{\lvl(e)})) \cup S_{\lvl(e) + 1})} \\
	={} &
	\prob{e \not \in \Span(\Rdist(\vx(\lvl(e))) \cup S_{\lvl(e) + 1})}
	>
	1 - \tau
	,
\end{align*}
where the equality holds since one can verify that while the values of the sets $\Iset_0, \Iset_1, \dotsc, \Iset_{\ell}$ depend on the value of $\Rset$, the construction of the sets $S_1, S_2, \dotsc, S_\ell$ (and thus, also $\lvl(e)$) is independent of the value of $\Rset$. Notice now that since Part~\eqref{item:I_inclusion} of \cref{lem:properties} guarantees that $\Iset_{\lvl(e)}$ is a subset of $\Span(S_{\lvl(e)})$ and \cref{cor:feasible} guarantees that $\Iset_{\lvl(e)} \subseteq \Iset \subseteq \Rset$,
\begin{align*}
	\prob*{e \in \Rset \text{ and } e \not \in \Span((\Iset_{\lvl(e)}{ } - e) \cup S_{\lvl(e) + 1})} 
	&\geq \prob*{e \in \Rset \text{ and } e \not \in \Span((\Rset \cap \Span(S_{\lvl(e)}) - e) \cup S_{\lvl(e) + 1})} \\
	&=
	\prob*{e \in \Rset} \cdot \prob{e \not \in \Span((\Rset \cap \Span(S_{\lvl(e)}) - e) \cup S_{\lvl(e) + 1})}
	\\
	&\geq
	x_e \cdot (1 - \tau),
\end{align*}
where the equality holds since the membership of each element in $\Rset \sim \Rdist(\vx)$ is independent.

To complete the proof of the proposition, it remains to observe that the event that both $e \in \Rset$ and $e \not \in \Span((\Iset_{\lvl(e)} - e) \cup S_{\lvl(e) + 1})$ implies that $e \in \Iset_{\lvl(e)}$, and thus $e \in \Iset$. Assume towards a contradiction that this is not the case; then by the definition of $\lvl(e)$, $e \in \Span(S_{\lvl(e)})$ and so
\[
	e
	\in
	\Span(S_{\lvl(e)}) \setminus \Span((\Iset_{\lvl(e)} - e) \cup S_{\lvl(e) + 1})
	=
	\Span(S_{\lvl(e)}) \setminus \Span(\Iset_{\lvl(e)} \cup S_{\lvl(e) + 1})
	,
\]
which contradicts the fact that in iteration number $\lvl(e)$ of the loop starting on \cref{line:main_T_loop} of \cref{alg:matroid-DCRS}, the inner loop on \cref{line:T_add_condition} terminated without adding $e$ to $\Iset_{\lvl(e)}$.
\end{proof}


\subsection{Recourse} \label{sec:recourse}

We are ready to prove the recourse bound. Let $\Inc(T)$ and $\Dec(T)$ be the total amount of increase (resp. decrease) to $\vx$, i.e.
\begin{equation}
\label{eq:inc_dec_def}
    \Inc(T) := \sum_{t = 1}^T \sum_{e \in \cN} \max(0, x^t_e - x^{t-1}_e), \qquad\quad
    \Dec(T) := \sum_{t = 1}^T \sum_{e \in \cN} \max(0, x^{t - 1}_e - x^{t}_e)
    . 
\end{equation}
Note that $\sum_{t = 1}^T \norm{\vx^t - \vx^{t-1}}_1 = \Inc(T) + \Dec(T)$. 

In analogy to the increase-only case, we start by showing that the number of nonempty levels is always logarithmically bounded.

\begin{lemma} \label{lem:set_size_bound}
When \cref{alg:matroid-DCRS} terminates the processing of $\vx^t$, $|S_i| \leq (rb/\tau)^{i - j} \cdot |S_j|$ for every two integers $0 \leq j \leq i$. In particular, for $\ell = \lceil \log_{\tau / (rb)} \rank(\cM)\rceil$, we have $|S_{\ell + 1}| \leq (rb/\tau)^{\ell + 1} \cdot |S_0| = (rb/\tau)^{\ell + 1} \cdot \rank(\cM) < 1$.
\end{lemma}
\begin{proof}
For every $i \in [\ell]$, \cref{line:early_reset_condition} and \cref{line:reset_condition} of \cref{alg:matroid-DCRS} guarantee together that when iteration $i$ of the loop starting on \cref{line:main_S_loop} terminates, we have the inequality 
\[
	|S_i|
	\leq
	\frac{r \cdot \|\vx(i-1)\|_1}{\tau}
	\leq
	\frac{r \cdot (b \cdot |S_{i - 1}|)}{\tau}
	=
	\frac{rb}{\tau} \cdot |S_{i - 1}|
	,
\]
where the second inequality holds since $\vx \in b \cdot P_\cM$. Since $S_i$ and $S_{i - 1}$ do not change after this point until the processing of $\vx^t$ by \cref{alg:matroid-DCRS} ends, the last inequality holds also when this processing terminates. Furthermore, iterating this inequality over a range of levels implies that for every two integers $0 \leq j \leq i \leq \ell$
\[
	\abs{S_i}
	\leq
	\Big(\frac{rb}{\tau}\Big)^{i - j} \cdot \abs{S_j}
	,
\]
which completes the proof of the first part of the lemma. Towards the second part, notice that
\[
	\abs{S_{\ell + 1}}
	\leq
	\Big(\frac{rb}{\tau}\Big)^{\ell + 1} \cdot |S_0|
	=
	\Big(\frac{rb}{\tau}\Big)^{\ell + 1} \cdot \rank(\cM)
	\leq
	\Big(\frac{rb}{\tau}\Big)^{\log_{\tau / (rb)} \rank(\cM) + 1} \cdot \rank(\cM)
	=
	\frac{rb}{\tau}
    <
    1
	.
	\qedhere
\]
\end{proof}

We get the following immediate corollary.

\begin{corollary} \label{cor:non-empty_count}
    When \cref{alg:matroid-DCRS} terminates the processing of $\vx^t$ the number of nonempty levels is upper bounded by $\ell = \lceil \log_{\tau / (rb)} \rank(\cM)\rceil$.    
\end{corollary}

Henceforth, we use $\ell := \lceil \log_{\tau / (rb)} \rank(\cM)\rceil$ to refer to the above bound on the number of nonempty levels.

\begin{lemma} \label{lem:remaining}
When Algorithm~\ref{alg:matroid-DCRS} terminates, $\sum_{i = 1}^\ell |S_i| \leq \frac{\Inc(T)}{\tau - rb}$.
\end{lemma}
\begin{proof}
Let $J$ be the time interval between the last moment in which $S_1$ was the empty set and the termination of the execution of Algorithm~\ref{alg:matroid-DCRS}, and let $\bar{\vx}$ be the coordinate-wise maximum of the values $\vx$ takes throughout the interval $J$. The definition of the interval $J$ implies that $\Reset(1)$ is never invoked during it, and thus, \cref{alg:matroid-DCRS} changes $S_1$ during this interval only via \cref{line:add}. This line adds to $S_1$ only elements $e$ that prior to their addition obeyed both $e \not \in \Span(S_1)$ and
\[
	\prob{e \in \Span(\Rdist(\bar{\vx}) \cup S_1)}
	\geq
	\prob{e \in \Span(\Rdist(\vx) \cup S_1)}
	\geq
	\tau
\]
(here $\vx$ denotes $\vx^t$ at the step when $e$ is added to $S_1$). This process satisfies the conditions of \cref{lem:growth_bound} with $\vy = \bar{\vx}$, and so we get that the size of the set $S_1$ at the end of interval $J$, i.e.\ when \cref{alg:matroid-DCRS} terminates, obeys
\[
	|S_1|
	\leq
	\frac{\|\bar{\vx}\|_1}{\tau}
	\leq
	\frac{\Inc(T)}{\tau}
	,
\]
where the second inequality holds since $\vx^0=0$ and so $\Inc(T)$, the total increase in $\vx$, is an upper bound on $\|\bar{\vx}\|_1$.
Thus by \cref{lem:set_size_bound},
\begin{align*}
	\sum_{i = 1}^\ell |S_i|
	\leq{} &
	\sum_{i = 1}^\ell \Big(\frac{rb}{\tau}\Big)^{i - 1} \cdot |S_1|
    \leq
	\frac{|S_1|}{1 - rb/\tau}
	\leq
	\frac{\Inc(T) / \tau}{1 - rb/\tau}
	=
	\frac{\Inc(T)}{\tau - rb}
	.
	\qedhere
\end{align*}
\end{proof}

We now bound the number of element removals from the sets $S_1, S_2, \dotsc, S_\ell$.
\begin{lemma} 
\label{lem:removed}
    The total number of elements \Cref{alg:matroid-DCRS} removes from the sets $S_1, S_2, \dotsc, S_\ell$ is upper bounded by $\frac{\ell \cdot \Dec(T)}{(\tau - b)(1 - 1/r)}$. Thus,
    \[
        \sum_{t \in [T]} \sum_{i \in [\ell]} \abs{S^{t-1}_i \setminus S^t_i } \leq \frac{\ell \cdot \Dec(T)}{(\tau - b)(1 - 1/r)}
        .
    \]
\end{lemma}
\begin{proof}
Fix $i \in [\ell]$. In the following, we will refer to every time in which $\Reset(i')$ is invoked for some $i' \leq i$ as an \emph{emptying} of $S_i$. Furthermore, if $i' = i$ we call this a \emph{primary emptying} of $S_i$. Let $k$ be the number of times that $S_i$ is emptied, and for every integer $j \in [k]$, denote by $s_{j}$ the time of the $j$th emptying of $S_i$. It is also useful to denote by $s_{0}$ the time at which the algorithm begins.

For every $j \in [k]$, let $\hat{S}_{i - 1, j}$ and $\hat{S}_{i, j}$ denote the sets $S_{i - 1}$ and $S_i$, respectively, immediately before time $s_{j}$.
Observe that between reset times $s_{j-1}$ and $s_j$, no elements are removed from $S_i$.
Let $\bar{\vx}^{(i - 1, j)}$ be the coordinate-wise maximum of the values taken by the vector $\vx(i - 1)$ between times $s_{j - 1}$ and $s_{j}$. 
Every element added to $S_i$ in the time range we consider was added on \Cref{line:add} and so satisfies the \Cref{line:add_condition} condition 
\[
    \tau \leq \probarg{e \in \Span(\Rset \cup S_{i})}{\Rset \sim \Rdist(\vx(i-1))} 
    \leq \probarg{e \in \Span(\Rset \cup S_{i})}{\Rset \sim \Rdist(\bar \vx^{(i-1,j)})} 
     .
\]
Thus, the construction of $\hat S_{i,j}$ satisfies \cref{lem:growth_bound} for $\vy = \bar{\vx}^{(i-1,j)}$, and so $\abs{\hat{S}_{i, j}} \leq \norm{\bar{\vx}^{(i - 1, j)}}_1/\tau$.

By the definition of the {\Reset} procedure in \cref{alg:matroid-DCRS}, every emptying of $S_{i - 1}$ is accompanied by an emptying of $S_i$ at the same time. Thus, $S_{i - 1}$ is not emptied between times $s_{j - 1}$ and $s_j$, which implies that between these times elements are only added to $S_{i - 1}$. Now denote by $\hat{\vx}^{(i - 1, j)}$ the vector $\vx(i-1)$ immediately before time $s_j$. Since between times $s_{j - 1}$ and $s_j$ elements are only added to $S_{i - 1}$, any coordinate of $\hat{\vx}^{(i - 1, j)}$ that is lower than the corresponding coordinate of $\bar{\vx}^{(i - 1, j)}$ implies that the difference of value must have been due to some decrease in $\vx$ between times $s_{j - 1}$ and $s_{j}$. Thus, if we extend \eqref{eq:inc_dec_def} to let $\Dec(s_{j - 1}, s_{j})$ be the total size of the decreases between these times, i.e.
\[
    \Dec(s_{j - 1}, s_{j}) := \sum_{\substack{t:\text{ processing of } \vx^t \\ \text{starts between times}\\\text{$s_{j - 1}$ and $s_j$}}} \ \sum_{e \in \cN} \ \max(0, x^{t - 1}_e - x^{t}_e)
    ,
\]
then we get
\begin{equation} \label{eq:non-primary_emptying}
	|\hat{S}_{i, j}|
	\leq
	\frac{\norm{\bar{\vx}^{(i - 1, j)}}_1}{\tau}
	\leq
	\frac{\|\hat{\vx}^{(i - 1, j)}\|_1 + \Dec(s_{j - 1}, s_{j})}{\tau}
	\leq
	\frac{b \cdot |\hat{S}_{i - 1, j}| + \Dec(s_{j - 1}, s_{j})}{\tau}
	,
\end{equation}
where the last inequality holds since $\hat{\vx}^{(i - 1, j)}=\vx(i-1)$ at time $s_{j}$ and $\vx \in b \cdot P_\cM$ at all times.
If the emptying of $S_i$ at time $s_j$ is primary, then $\Reset(i)$ is invoked at this time, which by \Cref{line:early_reset_condition,line:reset_condition} implies $\abs{\hat{S}_{i, j}} > r \cdot \|\hat{\vx}^{(i - 1, j)}\|_1 / \tau$ and hence
\[
	|\hat{S}_{i, j}|
	\leq
	\frac{\|\hat{\vx}^{(i - 1, j)}\|_1 + \Dec(s_{j - 1}, s_{j})}{\tau}
	<
	\frac{\tau \cdot |\hat{S}_{i, j}| / r + \Dec(s_{j - 1}, s_{j})}{\tau}
	=
	\frac{|\hat{S}_{i, j}|}{r} + \frac{\Dec(s_{j - 1}, s_{j})}{\tau}
	,
\]
which by rearrangement gives
\begin{equation} \label{eq:primary_emptying}
	\abs{\hat{S}_{i, j}}
	<
	\frac{\Dec(s_{j - 1}, s_{j})}{\tau(1 - 1/r)}
	.
\end{equation}
Adding Inequality~\eqref{eq:non-primary_emptying} for every non-primary emptying of $S_i$ and Inequality~\eqref{eq:primary_emptying} for every primary emptying of $S_i$, we get
\begin{align*}
	\sum_{j = 1}^{k} |\hat{S}_{i, j}|
	\leq{} &
	\sum_{j = 1}^k \characteristic\Big[\begin{array}{ll}\text{the $j$-th emptying of}\\\text{$S_i$ was non-primary}\end{array}\Big] \cdot \frac{b \cdot |\hat{S}_{i - 1, j}|}{\tau} +  \sum_{j = 1}^{k} \frac{\Dec(s_{j - 1}, s_{j})}{\tau(1 - 1/r)}\\
	\leq{} &
	\frac{b}{\tau} \cdot \sum_{j = 1}^k \characteristic\Big[\begin{array}{ll}\text{the $j$-th emptying of}\\\text{$S_i$ was non-primary}\end{array}\Big] \cdot |\hat{S}_{i - 1, j}| + \frac{\Dec(T)}{\tau(1 - 1/r)}
	,
\end{align*}
where the second inequality holds since every decrease in $\vx$ is counted by $\Dec(s_{j - 1}, s_{j})$ for at most one value $j \in [k]$.

Denote now by $E_i$ the total number of element removals from $S_i$. Then, the leftmost side of the last inequality is $E_i$, and the first term on the rightmost side of this equality is $(b/\tau) \cdot E_{i - 1}$ because the fact that an emptying of $S_i$ at time $s_j$ is non-primary implies that $S_{i - 1}$ was also emptied at this time. Hence, the last inequality can be rewritten as
\[
	E_i \leq \frac{b}{\tau} \cdot E_{i - 1} + \frac{\Dec(T)}{\tau(1 - 1/r)}
	.
\]
Since this inequality holds for every $i \in [\ell]$, it implies an upper bound on $E_i$. Iterating yields
\begin{align*}
	E_i
	\leq{} &
	\sum_{i' = 1}^{i} \Big(\frac{b}{\tau}\Big)^{i' - 1} \cdot \frac{\Dec(T)}{\tau(1 - 1/r)} + \Big(\frac{b}{\tau}\Big)^{i} \cdot E_0\\
	\leq{} &
	\sum_{i' = 0}^{\infty} \Big(\frac{b}{\tau}\Big)^{i'} \cdot \frac{\Dec(T)}{\tau(1 - 1/r)}
	=
	\frac{1}{1 - b/\tau} \cdot \frac{\Dec(T)}{\tau(1 - 1/r)}
	=
	\frac{\Dec(T)}{(\tau - b)(1 - 1/r)}
	,
\end{align*}
where the second inequality uses the fact that $E_0 = 0$ as no elements are ever removed from $S_0$.
The lemma now follows by observing that the number of element removals from all the sets $S_1, S_2, \dotsc, S_\ell$ together is $\sum_{i = 1}^\ell E_i$.
\end{proof}

\begin{corollary} 
\label{cor:S_i_additions}
    The total number of additions of elements to the sets $S_1, S_2, \dotsc, S_\ell$ satisfies
    \[
        \sum_{t \in [T]} \sum_{i \in [\ell]} \abs{S^t_i \setminus S^{t-1}_i} 
        \leq \frac{\Inc(T)}{\tau - rb} + \frac{\ell \cdot \Dec(T)}{(\tau - b)(1 - 1/r)}
        .
    \]
\end{corollary}

\begin{proof}
Every element added to any of the sets $S_1, S_2, \dotsc, S_\ell$ must either remain in this set when the algorithm terminates or correspond to an element removal at some step along the way.
By \cref{lem:remaining}, only $\frac{\Inc(T)}{\tau/r - b}$ elements remain in the sets $S_1, S_2, \dotsc, S_\ell$ upon termination; 
and by \cref{lem:removed}, the total number of removals from the sets $S_1, S_2, \dotsc, S_\ell$ is at most $\frac{\ell \cdot \Dec(T)}{(\tau - b)(1 - 1/r)}$.
\end{proof}

So far, we have bounded the number of changes in the sets $S_1, S_2, \dotsc, S_\ell$. However, the recourse is the change in $\Iset = \cup_{i = 0}^{\ell} \Iset_i$. Therefore, we now need to bound the number of changes in the sets $\Iset_0, \Iset_1, \dotsc, \Iset_{\ell}$.

\begin{lemma} 
\label{lem:T_removals}
    The total number of removals from the output set is bounded by
    \[
        \sum_t \abs{\Iset^{t-1} \setminus \Iset^{t}} 
        \leq \sum_t \abs{\Rset^{t-1} \setminus \Rset^{t}} + \frac{2\cdot \Inc(T)}{\tau - rb} + \frac{3\ell \cdot \Dec(T)}{(\tau - b)(1-1/r)}
        .
    \]
\end{lemma}
\begin{proof}
    We prove the stronger fact that $\sum_t \sum_i \abs{\Iset^{t-1}_i \setminus \Iset^{t}_i}$ obeys this bound.
    Consider how \cref{alg:matroid-DCRS} removes elements from the sets $\Iset_0, \Iset_1, \dotsc, \Iset_{\ell}$.
    For fixed level $i$, elements are removed from $I_i$ either 
    (a) due to additions to $S_{i-1}$ (\Cref{line:add,line:T_remove}), or
    (b) due to removals from $\Rset$ (\Cref{line:non-R_remove}), or 
    (c) as part of a $\Reset$ operation (\Cref{line:early_reset_condition,line:reset_condition}).
    This is an exhaustive (and exclusive) list of all removals from $\Iset_i$ across all moments in the iteration of \Cref{alg:matroid-DCRS}, and so denoting these removal counts (over all $i \in [\ell]$) by $C^{(a)}, C^{(b)}, C^{(c)}$, we have
    \[
        \sum_t \abs{\Iset^{t-1} \setminus \Iset^{t}}  \leq \sum_t \sum_i \abs{\Iset^{t-1}_i \setminus \Iset^{t}_i}
        \leq C^{(a)} + C^{(b)} + C^{(c)}
        .
    \]

    First, consider removals of type (a).
    If \cref{line:add} or \cref{line:T_remove} of \cref{alg:matroid-DCRS} remove a single element $e'$ from one of the sets $\Iset_0, \Iset_1, \dotsc, \Iset_{\ell}$, then \cref{line:add} added an element $e$ to one of the sets $S_1, S_2, \dotsc, S_{\ell}$. Thus, the number of elements removed by these lines from the sets $\Iset_0, \Iset_1, \dotsc, \Iset_{\ell}$ is upper bounded by twice the total number of elements added to the sets $S_1, S_2, \dotsc, S_{\ell}$, which is upper bounded in \cref{cor:S_i_additions} by
    \begin{equation}
    \label{eq:Ca-bound}
        C^{(a)} \leq 2 \sum_{t} \sum_{i} \abs{S^t_i \setminus S^{t-1}_i} 
        \leq 2 \left[\frac{\Inc(T)}{\tau - rb} + \frac{\ell \cdot \Dec(T)}{(\tau - b)(1 - 1/r)}\right]
        .
    \end{equation}

    Next consider removals of type (b).
    \cref{line:non-R_remove} removes elements $e \in \Iset_i$ that no longer belong to $\Rset$. Since \cref{alg:matroid-DCRS} adds elements to $\Iset_i$ only on \cref{line:add_T}, which appears after \cref{line:non-R_remove}, every element $e$ removed from $\Iset_i$ by \cref{line:non-R_remove} must have been in $\Iset_i \setminus R$ already at the beginning of the processing of $\vx^t$. However, \Cref{line:non-R_remove} and the condition on \Cref{line:add_condition} together guarantee that $\Iset_i \setminus \Rset$ was empty when the processing of $\vx^{t-1}$ terminated. Thus, every $e \in \Iset_i \setminus R$ removed by \cref{line:non-R_remove} must have been removed from $\Rset$ since the previous execution of \cref{alg:matroid-DCRS}.
    Hence
    \begin{equation}
    \label{eq:Cb-bound}
        C^{(b)} \leq \sum_t \abs{\Rset^{t - 1} \setminus \Rset^{t}}.
    \end{equation}

    Finally consider removals of type (c), due to resets.
    Observe that 
    whenever $\Reset$ removes all of $\Iset_i$ , it also removes all of $S_{i}$ (see \Cref{line:reset_removal_S,line:reset_removal_I}). By Parts~\eqref{item:I_independent} and~\eqref{item:I_inclusion} of \cref{lem:properties}, $S_{i}, \Iset_i \in \cI$ and $\Iset_i \subseteq \Span(S_{i})$, which together imply $|\Iset_i|
		=
		\rank(\Iset_i)
		\leq
		\rank(\Span(S_{i}))
		=
		\rank(S_{i})
		=
		|S_{i}|$.
    Thus, the number of elements removed by $\Reset$ from the sets $\Iset_i$ is upper bounded by the number of elements removed by $\Reset$ from the sets $S_i$. By \Cref{lem:removed}, we get
    \begin{equation}
    \label{eq:Cc-bound}
        C^{(c)}
        =
        \bigg\{\begin{array}{c}\text{\# of elements removed} \\ \text{from the sets $\Iset_i$ by \Reset}\end{array}\bigg\}
        \leq
        \bigg\{\begin{array}{c}\text{\# of elements removed} \\ \text{from the sets $S_i$ by \Reset}\end{array}\bigg\}
        \leq \frac{\ell \cdot \Dec(T)}{(\tau - b)(1 - 1/r)}
        .
    \end{equation}
    Combining \Cref{eq:Ca-bound,eq:Cb-bound,eq:Cc-bound} proves the claim.
\end{proof}

We are now ready to bound the recourse of our matroid DCRS.
\begin{proposition}\label{prop:recourse_bound}
    The total expected recourse of \Cref{alg:matroid-DCRS} is at most 
    \[ 
        O\left(\frac{r\ell}{(\tau - rb)(r - 1)}\right) \cdot \sum_{t = 1}^T \expect*{\abs{\Rset^t \symdif \Rset^{t-1}}}
        .
    \]
    In the incremental model, i.e. when $\Dec(T) = 0$, the expected recourse improves to 
    \[
        O\left(\frac{1}{\tau - rb}\right) \cdot \sum_{t = 1}^T \expect*{\abs{\Rset^t \symdif \Rset^{t-1}}} 
        .
    \]
\end{proposition}

\begin{proof}
    The number of additions to $\Iset$ is at most the number of deletions plus however many elements remain when \cref{alg:matroid-DCRS} terminates. Since $\Iset = \cup_{i = 0}^{\ell} \Iset_i \subseteq \Rset$ and $\expect{\abs{R^T}} = \norm{\vx^T}_1$, we can bound the expected number of elements that remain by $\Inc(T)$. Thus, the total expected number of additions and deletions is at most $\Inc(T)$ plus two times the expected number of deletions. 
    Deploying the bound derived in \cref{lem:T_removals} and taking the expectation,
    \begin{align*}
    	\sum_t \expect*{\abs{\Iset^t \symdif \Iset^{t-1}}}
        &\leq 2 \sum_t \expect*{\abs{\Rset^t \setminus \Rset^{t-1}}} + \left(\frac{4}{\tau - rb} +1 \right) \cdot \Inc(T) + \frac{6\ell}{(\tau - b)(1-1/r)}\cdot \Dec(T) \\
        &= 2 \sum_t \expect*{\abs{\Rset^t \setminus \Rset^{t-1}}} + O\left( \frac{r\ell}{(\tau - rb)(r - 1)}\right) \cdot (\Inc(T) + \Dec(T)) , \\
        &=O\left( \frac{r\ell}{(\tau - rb)(r - 1)}\right) \cdot \sum_t \expect*{\abs{\Rset^t \symdif \Rset^{t-1}}}
        ,
    \end{align*}
    where the last line follows from the earlier observation that $\Inc(T) + \Dec(T) = \sum_t \norm{\vx^t - \vx^{t-1}}_1$, together with \Cref{obs:TCOS-1-impossibility}. When $\Dec(T) = 0$, we get the improved bound
		    \begin{align*}
    	\sum_t \expect*{\abs{\Iset^t \symdif \Iset^{t-1}}}
        &\leq 2 \sum_t \expect*{\abs{\Rset^t \setminus \Rset^{t-1}}} + \left(\frac{4}{\tau - rb} +1 \right) \cdot \Inc(T) \\
        &= 2 \sum_t \expect*{\abs{\Rset^t \setminus \Rset^{t-1}}} + O\left( \frac{1}{\tau - rb}\right) \cdot \Inc(T) , \\
        &=O\left( \frac{1}{\tau - rb}\right) \cdot \sum_t \expect*{\abs{\Rset^t \symdif \Rset^{t-1}}}
        .
				\qedhere
    \end{align*}
\end{proof}

\subsection{Monotonicity}
\label{sec:matroid-dcrs-monotone}

Recall from \Cref{sec:prelim_dcrs} the definition of a monotone CRS, and that a DCRS is monotone if it is a monotone CRS at every step $t \in [T]$ (conditioned on the history $\Rset^{\leq t-1}$).
This property will be useful for combining DCRS for various constraint families, as we will see in \Cref{sec:combine}.

We establish that our matroid DCRS is monotone.

\begin{proposition}
\label{prop:matroid-dcrs-monotone}
    \Cref{alg:matroid-DCRS} is monotone.
\end{proposition}
\begin{proof}
    We will consider the DCRS at an arbitrary step $t \in [T]$, and consider its behavior when the active set is $\Rset^t = \Rset$ vs. $\Rset^t = \Rset'$, where $\Rset \subseteq \Rset' \subseteq \cN$.
    Fix a common history $\vx^{\leq t}$ and $\Rset^{\leq t-1}$ and corresponding (deterministic) $\{S^{t-1}_i\}$ and $\{\Iset^{t-1}_i\}$, which are the level bases and the level outputs at the end of step $t-1$.
    If we denote our DCRS by $\pi$, then $\Iset = \cup_i \Iset_i \defeq \pi_{\vx^{\leq t}}(\Rset^{\leq t - 1}, \Rset)$ and $\Iset' = \cup_i \Iset_i' \defeq \pi_{\vx^{\leq t}}(\Rset^{\leq t - 1}, \Rset')$ are the outputs of our DCRS for round-$t$ active sets $\Rset$ and $\Rset'$, respectively.
    Since our $\pi_{\vx^{\leq t}}(\Rset^{\leq t - 1}, \cdot)$ is deterministic, it suffices to show that $\Iset'_i \cap \Rset \subseteq \Iset_i$; that is, that if $e \in \Rset$ is not included in $\Iset_i$, then $e$ is not included in $\Iset'_i$ either.

    Our first observation is that the level bases $\{S^{t}_i\}$ computed in step $t$ depend only on the previous bases $\{S^{t-1}_i\}$ and on $\vx^{t}$, but not on $\Rset^t$. 
    This is because the decision of whether to reset (\cref{line:early_reset_condition,line:reset_condition}) depends only on $S_i$, $S_{i-1}$, and $\vx^t$, and all element additions to $S_i$ happen on \Cref{line:add}, whose gate condition (\Cref{line:add_condition}) depends only on $S_i$, $S_{i-1}$, and $\vx^t$.
    Therefore the CRS $\pi_{\vx^{\leq t}}(\Rset^{\leq t - 1}, \cdot)$ operates with a well-defined collection of level bases $\{S_i^t\}$, independent of its input $\Rset^t$.

    Our next observation is that up to \Cref{line:main_T_loop}, the modifications that $\pi_{\vx^{\leq t}}(\Rset^{\leq t - 1}, \cdot)$ makes to the $\{\Iset_i\}$ are independent of $\Rset^t$. 
    In particular, the only modifications made to derive each $\Iset_i$ from the $\Iset^{t-1}_i$ up to that point are removals on \Cref{line:add,line:T_remove}, which depend (according to \Cref{line:T_remove_condition}) only on the $S_i$'s, which we have established evolve independently of $\Rset^t$.

    Therefore $\Iset_i$ and $\Iset_i'$ are the same at the beginning of the loop that starts on \Cref{line:main_T_loop}. This loop is then equivalent to making a pass over the elements according to the order $\sigma$, and greedily adding each element $e$ in this order to $\Iset_i$ (resp. $\Iset_i'$) if $e \in \Rset \cap (\Span(S_{i-1}) \setminus \Span(S_i \cup \Iset_i))$ (resp. $e \in \Rset' \cap (\Span(S_{i-1}) \setminus \Span(S_i \cup \Iset_i'))$).
    We claim that whenever an element $e$ is considered, it holds that $\Span(S_i \cup \Iset_i) \subseteq \Span(S_i \cup \Iset_i')$; and thus, $e\in \Rset$ is added to $\Iset_i'$ only when it is added to $\Iset_i$, which guarantees the desired inclusion $\Iset_i' \cap \Rset \subseteq \Iset_i$.

    So why is $\Span(S_i \cup \Iset_i) \subseteq \Span(S_i \cup \Iset_i')$ for all of the $\Iset_i$ and $\Iset_i'$ at the moments at which each $e$ is considered on the loop?
    This is clearly true initially, since $\Iset_i=\Iset_i'$ at the outset.
    For each $e$, consider all $f \in \Iset_i$ that precede $e$ in $\sigma$ order.
    Clearly, $f \in \Rset \cap \Span(S_{i-1}) \subseteq \Rset' \cap \Span(S_{i-1})$.
    Thus, if $f \not \in \Span(S_i \cup \Iset_i')$ when $f$ is considered, then $f$ is added to $\Iset_i'$, which guarantees that $f \in \Span(S_i \cup \Iset_i')$ after it is considered.
    Therefore, when $e$ is considered, it must hold that $\Iset_i \subseteq \Span(S_i \cup \Iset_i')$, and hence also $\Span(S_i \cup \Iset_i) \subseteq \Span(S_i \cup \Iset_i')$, which is what we wanted to prove.
\end{proof}

\subsection{Setting the Parameters}
\label{sec:matroid-conclusion}

We conclude by combining the above guarantees for \Cref{alg:matroid-DCRS} and choosing parameters to make a self-contained statement of our main result for matroids, restated here.

\matroiddcrsformal*

\begin{proof}
    \Cref{prop:matroid-dcrs-monotone} establishes that \Cref{alg:matroid-DCRS} is monotone, and \Cref{prop:matroid-DCRS-balance} establishes that the balance is $1 - \tau$, which is $1 - b - \eps$, when we set $\tau = b + \eps$ (we assume $\eps + b < 1$, as otherwise \cref{thm:mat_dcrs_formal} guarantees zero balance, which is trivial). To get a bound also on the recourse, we need to choose a value also for $r$, which we set to $1 + \frac{\eps}{1 + b}$. This choice of $r$ yields
\[
	\ell
	=
	\lceil \log_{\tau / (rb)} \rank(\cN)\rceil
	=
    	\lceil \log_{1 + \eps/((1 + b)(rb))} \rank(\cN)\rceil
	\leq
	\lceil \log_{1 + \eps/(2b)} \rank(\cN)\rceil
	=
	O(b\eps^{-1} \log \rank(\cN))
	,
\]
and thus, leads to a recourse of 
\begin{align*}
    O\left(\frac{r\ell}{(\tau - rb)(r - 1)}\right) &= O\left(\frac{(1 + \eps/(1 + b)) (b\eps^{-1} \log \rank(\cN))}{(b + \eps - (1 + \eps(1 + b)^{-1})b)(1 + \eps/(1 + b) - 1)}\right) \\
    &= O\left(\frac{(1 + \eps/(1 + b)) (b\eps^{-1} \log \rank(\cN))}{\eps^2 / (1 + b)^2}\right) = O\left(\eps^{-3} \log \rank(\cN)\right)
    .
\end{align*}
In the increase-only model, the recourse improves to 
\[
    O\left(\frac{1}{\tau - rb}\right) = O\left(\frac{1}{b + \eps - (1 + \eps/(1 + b))b}\right) = O\left(\frac{1 + b}{\eps}\right) = O\left(\frac{1}{\varepsilon}\right). \qedhere
\]
\end{proof}

%% file: partition.tex
\section{Partition Matroids}
\label{sec:partition}

In this section we show a simple DCRS for partition matroids with improved approximation factor of $(1-e^{-1})$ and $O(1)$ recourse bound, which matches the results of \cite{buchbinder2025chasingsubmodularobjectivessubmodular} (up to the exact constant in the recourse). Our technique has the advantage that it generalizes to ($p$-hypergraph) matchings, which are ($p$-)matchoids constructed from partition matroids.

A partition matroid $\cM$ is defined by a partition $(P_1, \ldots, P_h)$ of $\cN$, and a capacity $c_P$ for every part $P$ of the partition. A subset $I \subseteq \cN$ is defined to be independent if $|I \cap P| \leq c_P$ for every part $P$.

\begin{theorem}[Partition Matroid DCRS]
\label{thm:better-partition-dcrs}
For any $b \in [0, 1]$, there is a monotone $(b, \frac{1-e^{-b}}{b}, 4)$-DCRS for partition matroids with arbitrary capacities. In particular, for $b=1$, the balance is at least $1-1/e$.
\end{theorem}
\begin{proof}
    The algorithm is simple. For every $e \in \cN$, let $P$ be the part containing $e$. Then: 
    \begin{itemize}
        \item Assign $e$ a random \emph{priority} $y_e \sim \uni(0,1)$, and let $\sigma$ be the ordering of $\cN$ according to increasing priority. Note that $\sigma$ is uniformly random and fixed for all time steps $t \in [T]$. 
        \item Assign $e$ a random \emph{acceptance threshold} $z_e \sim \Exp(y_e / c_P)$ to $e$. We say $e$ is \emph{alive at time $t$} if both $e \in R^t$ and $z_e \geq x^t_e$. Observe that for every time $t \in [T]$, conditioned on $e \in R^t$, the probability that $e$ is alive is $\exp(-y_e x^t_e/c_P)$. Call $A^t$ the set of alive elements at time $t$.
        \item Assign $e$ a uniformly random \emph{bucket} among $B^P_1, B^P_2, \dotsc, B^P_{c_P}$. Let $B^P_e$ denote the bucket to which $e$ was assigned.
    \end{itemize}
    
     At every time $t$, process the elements in the order $\sigma$, and accept each element $e \in \cN$ into the solution $I^t$ if (a) $e$ is alive, and (b) no other elements from the same bucket have been accepted so far, i.e. $|I^t \cap B^P_e| = 0$ at this point.

     \begin{claim}
        \label{claim:alive_by_tau_prob}
         For a fixed time $t$, part $P$, element $e \in P$ and threshold $\tau \in [0, 1]$, the probability that both $e$ is alive at time $t$ and $y_e \leq \tau$ is $c_P \cdot (1-\exp(-\tau \cdot x_e / c_P))$.
     \end{claim}
     \begin{proof}
         Conditioned on $y_e = y$, element $e$ is alive at time $t$ with probability $x_e \cdot \exp(-y \cdot x_e / c_P)$ (because $e \in R^t$ and $z_e \geq x^t_e$ are independent events). Integrating, we get
         \[
            \prob{e \in A^t \, \land \, y_e \leq \tau} = \int_0^\tau \prob{e \in A^t \mid y_e = y} \, dy = \int_0^\tau x_e \cdot e^{-y x_e/c_P} \, dy = c_P\left(1- e^{-\tau x_e/c_P}\right). \qedhere
        \]
     \end{proof}

     \begin{claim}
         For a fixed time $t$ and element $e \in \cN$, conditioned on $y_e = y$, the probability $e$ is accepted into $I^t$ is at least $x_e \cdot \exp(-b y)$.
     \end{claim}
     \begin{proof}
        We say that an element $e'$ in the same part $P$ as $e$ is \emph{blocking} $e$ if (a) $e'$ is alive, (b) $e'$ has higher priority than $e$, and (c) $e'$ is also assigned to bucket $B_e^P$. Since the assignment of bucket is independent of priority and alive status, by \cref{claim:alive_by_tau_prob} and the uniform choice of bucket, the probability that $e'$ does not block $e$ is $1 - \frac{1}{c_P} \cdot c_P(1-\exp(-y x_{e'} / c_P)) = \exp(-y x_{e'} / c_P)$.
       
        Element $e$ is accepted into $I^t$ if $e$ is alive (which happens with probability $x_e \exp(-y x_e / c_P)$) and no other element $e'$ blocks $e$. Since the events that $e'$ blocks $e$ are independent across choices of $e'$, and independent from the event that $e$ is alive, we get
        \begin{align*}\prob{e \in I^t \mid y_e = y}
        &= x_e \exp(-y x_e / c_P) \cdot \prod_{\substack{e' \in P \\ e' \neq e}} \exp\left(-\frac{y x_{e'}}{c_P}\right) = x_e \exp\left(- \sum_{e' \in P} \frac{y x_{e'}}{c_P}\right) \geq x_e \exp(-b \cdot y),
        \end{align*}
        where we used the fact that $\sum_{e' \in P} x_{e'} \leq b\cdot c_P$ because $\vx \in b\cdot P_\cM$.
     \end{proof}

    From here we can easily conclude our balance claim: the probability $e \in \cN$ is accepted into $I^t$ is
    \[\prob{e \in I^t} = \int_0^1 \prob{e \in I^t \mid y_e = y} \, dy \geq \int_0^1 x_e \cdot e^{-b y}\: dy = x_e \cdot \frac{(1-e^{-b})}{b}.\]

    For the recourse analysis, we first claim that for any time $t$ and element $e$, we have $\prob{e \in A^t \symdif A^{t-1}} \leq 2\cdot \prob{e \in R^t \symdif R^{t-1}}$. Let us explain why this is the case when $\delta \defeq x_e^t - x_e^{t-1} \geq 0$ (the $\delta < 0$ case is analogous). In the case we consider,
    \[
        \prob{e \in A^t \symdif A^{t-1}} \leq \prob{e \in R^t \symdif R^{t-1}} + \prob{z_{e} \in [x_e^{t-1}, x_{e}^{t}]} \leq 2\cdot \prob{e \in R^t \symdif R^{t-1}},
    \] 
    where the second inequality holds since the fact that $c_P \geq 1$ implies that
		\begin{align*}
			\prob{z_{j} \in [x_e^{t-1}, x_{e}^{t}]} ={} & \int_0^1 [e^{-x_e^{t - 1} y / c_P} - e^{-x_e^t y / c_P}] \, dy = c_P \cdot \bigg[\frac{e^{-x_e^{t} / c_P} - 1}{x_e^t} - \frac{e^{-x_e^{t - 1} / c_P} - 1}{x_e^{t - 1}}\bigg] \\ ={} & -c_P \cdot \int_{x_e^{t - 1}}^{x_e^t} \frac{e^{-x / c_P}(x/c_P + 1) - 1}{x^2} \,  dx \leq -c_P \cdot \int_{x_e^{t - 1}}^{x_e^t} \frac{(1 - x / c_P)(x/c_P + 1) - 1}{x^2}  \, dx\\ ={} & -c_P \cdot \int_{x_e^{t - 1}}^{x_e^t} (- 1/c^2_P)  \, dx 
            =\frac{\delta}{c_P} \leq \delta.
		\end{align*}
    Observing now that each element $e \in A^t \symdif A^{t-1}$ can trigger at most a single change of another element in $\Iset$ (if $e$ is inserted into $A^t$ it pushes out at most one element, and if $j$ is removed from $A^{t-1}$, it causes at most one new element to be inserted), we get
    \[
        \expect{\abs{I^t \symdif I^{t-1}}} \leq 2\cdot\expect{\abs{A^t \symdif A^{t-1}}} \leq 4\cdot  \expect{\abs{R^t \symdif R^{t-1}}}.
    \]
    
		To verify monotonicity, consider arbitrary fixed $e \in \Rset' \subseteq \Rset$. The acceptance of $e$ depends on the realizations of $\{y_i\}_i$ and $\{z_i\}_i$. However, for every particular realization of these values, if $e$ is accepted to the solution when $\Rset$ is the realization of the active set, then when $\Rset'$ is the realization of the active set, only fewer alive elements precede $e$ in $\sigma$, and thus $e$ is still accepted into the solution.
\end{proof}

%% file: knapsack.tex
\section{DCRS for Knapsack Constraints}
\label{sec:knapsack}
 
In this section we move on and show a DCRS for knapsack constraints. A knapsack constraint $\cK$ is a linear packing constraint $\sum_{i \in [n]} s_i x_i \leq 1$, where the coefficient $s_i$ of every element $i \in [n]$ belongs to $[0,1]$. This coefficient is also known as the \emph{size} of element $i$. We denote by $\cP_\cK$ the polytope of fractionally feasible solutions for the knapsack constraint, i.e., vectors in $[0, 1]^n$ that obey the constraint.

\begin{theorem}[Knapsack DCRS]
\label{thm:dcrs-knapsack}
    For any $b \in [0,\nf{1}{2}]$, there are monotone $(b,\: (1-2b)/2,\: 2)$-DCRSs for knapsack constraints.
\end{theorem}

Note that in particular we get constant-competitive recourse. This is surprising since adding or removing a single large element to the solution a priori could trigger a large number of small-element deletions to restore feasibility, or a large number of small-element insertions to restore balance. This is in contrast to the matroid case, where the exchange axiom guarantees (for a fixed chain) that a single insertion/deletion only triggers one additional insertion/deletion.

To overcome this hurdle, we process the elements in \emph{non-decreasing} order of size. This may seem strange because the original CRS of \cite[Lemma 4.15]{DBLP:journals/siamcomp/ChekuriVZ14} proceeds in precisely the reverse non-\emph{increasing} size order, and this order is important for the analysis! However, it turns out that the analysis really requires only that large elements (e.g.\ those with size $\geq 1/2)$ be considered before all small elements (e.g.\ those with size $\leq 1/2$); the internal ordering of large/small elements does not matter. Hence we can use an idea which appears in \cite[Theorem 2.9]{feldman16online} of choosing between large and small elements randomly at the outset, throwing out the unchosen elements, and then proceeding in our desired non-decreasing order.

\begin{proof}[Proof of \Cref{thm:dcrs-knapsack}]
    Given knapsack constraint $\cK$ and a vector $\vx \in P_\cK$,
    let $\cN_{\text{big}} := \{e \in \cN \mid s_e > \nicefrac{1}{2}\}$ be the subset of big elements.

    At the outset of the online sequence, the DCRS designates the big elements to be alive  with probability $1/2$, and otherwise sets the small elements $\cN \setminus \cN_{\text{big}}$ to be alive. Then, at every time step $t \in [T]$, the DCRS scans the elements of $\cN$ in (fixed) \emph{non-decreasing} size order $\sigma$ and greedily accepts each element $e$ to the solution $\Iset^t$ if (i) it is alive, (ii) it belongs to $\Rset^t$, and (iii) adding it to $\Iset^t$ does not violate the feasibility of this set. In their proof \cite[Theorem 2.9]{feldman16online}, Feldman, Svensson, and Zenklusen show that this process is $(b,(1-2b)/2)$-balanced even if the order in which the elements are scanned is determined adversarially. Therefore we achieve the same balance parameters for our specific choice of non-decreasing size order.

    To bound the recourse of our DCRS, note that when an element enters or leaves the active set, because we scan the elements in a non-decreasing size order, this can cause at most one later item to enter or leave $\Iset^t$. 
    Hence,
    \[
        \sum_{t \in [T]} \expect*{\abs{\Iset^t \symdif \Iset^{t-1}}} \leq 2 \sum_{t \in [T]} \expect*{\abs{\Rset^t \symdif \Rset^{t-1}}}.
    \]
		
	It remains to verify that our DCRS is also monotone. Fix a realization $\Rset^t$ and consider an element $e \in \cN$ that belongs to $\Iset^t$ given  $\Rset^t$. It suffices to argue that for any $e' \in \Rset^t \setminus \{e\}$, if the realization were instead $\Rset' = \Rset^t \setminus \{e'\}$, then $e$ would still be included in the new output $\Iset'$. If $e' > e$, the elements in the solution when $e$ is considered remain unchanged, i.e. $\Iset' \cap \sigma_{<e} = \Iset^t \cap \sigma_{<e}$. (Here $e' > e$ denotes that $e$ precedes $e'$ in order $\sigma$, and $\sigma_{<e}$ denotes the elements preceding $e$ in $\sigma$.) If $e' < e$, then it cannot be the case that there is some $e' < e'' < e$ such that $e'' \in \Iset' \setminus \Iset^t$, because by virtue of the non-decreasing order of $\sigma$, if $e'' \in \Iset'$ then also $e'' \in \Iset^t$ because $e''$ is considered before $e$, and $e$ fits in $\Iset^t$. Hence the knapsack is only less full at the point at which $e$ is processed, i.e. $\Iset' \cap \sigma_{<e} \subseteq \Iset^t \cap \sigma_{<e}$. Thus there is still room for $e$, and $e \in \Iset'$.
\end{proof}
The claimed approximation ratio of $1/16$ in \Cref{tab:app} is derived by choosing $b = 1/4$, which is the value $b \in [0,\nf{1}{2}]$ maximizing $b \cdot \frac{1-2b}{2}$.

%% file: combiner.tex
\section{Combining Constraints}

\label{sec:combine}

Next, we show a combiner theorem which allows us to construct DCRS for the intersections of constraints when each one of these constraints individually admits a monotone DCRS.
In particular, this theorem allows us to construct DCRS for $k$-matchoids (intersections of an arbitrary number of matroids where every element participates in at most $k$ of the matroids), and $h$-sparse PIPs (packing integer programs with column sparsity at most $h$).

\begin{theorem}[Combiner Theorem]
\label{thm:combiner}
Consider $h$ constraints $\cF_1, \cF_2, \dots, \cF_h$ over ground sets $\cN_1, \cN_2, \dotsc,\allowbreak \cN_h$ such that each element $e$ appears in at most $p$ of the ground sets. Let us denote by $C_e$ the set of indices of the ground sets containing $e$, and let $P_i$ be a polytope relaxing $\cF_i$. Suppose there is a monotone $(b, c_i, r_i)$-DCRS $\pi^i$ for each constraint $\cF_i$ and polytope $P_i$, and there exist bounds $c$ and $r$ such that for each element $e \in \cN$, $\prod_{i \in C_e} c_{i} \geq c$ and $\sum_{i \in C_e} r_{i} \leq r$. Then, there is a monotone $(b, c, r)$-DCRS $\pi$ for the constraint $\cF = \{S \subseteq \cN \mid \forall_{i \in [h]} S \cap \cN_i \in \cF_i\}$ and the polytope $P = \{\vx \in [0, 1]^\cN \mid \vx|_{\cN_i} \in P_i\}$ relaxing it. In particular, one can always choose $c \defeq \prod_{i = 1}^h c_i$ and $r \defeq \sum_{i = 1}^h r_i$. Moreover, if all the DCRSs $\pi^i$ are implementable in polynomial time for all $i \in [h]$, then so is $\pi$.
\end{theorem} 

\begin{proof}
Let us define the output $\Iset^t$ of the DCRS $\pi$ at $t$ as follows. For every $e \in \cN$, $e$ belongs to $\Iset^t$ if
\[
	e \in \bigcap_{i \in C_e} \pi^i_{\vx^{\leq t}}(R^{\leq t})
	.
\]
Recall that $\pi^i_{\vx^{\leq t}}(R^{\leq t - 1}, \cdot)$ is a monotone $(b_i, c_i)$-balanced CRS for every $i \in [h]$, and thus, by \cite[Lemma 1.6]{DBLP:journals/siamcomp/ChekuriVZ14}, the set $\Iset^t$ defined above can be viewed as a monotone CRS, which implies that $\pi$ is a monotone DCRS. Furthermore, by the proof of the same \cite[Lemma 1.6]{DBLP:journals/siamcomp/ChekuriVZ14}, the probability that an element $e \in \cN$ belongs to $\Iset^t$ is at least 
$\prod_{i \in C_e} c_{i} \geq c$, and hence, $\pi$ is $(b, c)$-balanced.

It remains to observe that the recourse of the DCRS $\pi$ is bounded. 
Letting $R_i \defeq R\cap \cN_i$, we have
\begin{align*}
	\sum_{t=1}^T \bE[|\Iset^t \symdif \Iset^{t-1}|]
	&\leq \sum_{t=1}^T \sum_{i=1}^h \bE[|\pi^i_{\vx^{\leq t}}(R^{\leq t}) \symdif \pi^i_{\vx^{\leq t - 1}}(R^{\leq t - 1})|] 
    \leq \sum_{i=1}^h r_i \cdot \sum_{t=1}^T \bE[|R_i^t \symdif R_i^{t-1}|] \\
    &= \sum_{t=1}^T \sum_{e \in \cN} \left( \sum_{i \in C_e} r_i \right) \Pr[e \in R^t \symdif R^{t-1}|] 
	\leq r \cdot \sum_{t=1}^T \bE[|R_i^t \symdif R_i^{t-1}|] .
	\qedhere
\end{align*}
\end{proof}

We conclude the section with a note about the (asymptotically) optimal setting of parameters for various settings.
\begin{enumerate}
    \item For a single partition matroid, we set $b = 1$ and use \Cref{thm:better-partition-dcrs} to get a $(1-e^{-1},4)$-DCRS.
    \item For matching, we again set $b = 1$ and use \cref{thm:combiner} to get a $((1-e^{-1})^2,8)$-DCRS.
    \item For $p$-matchoid constraints, we set $b = \Theta(1/p)$ and $\eps = \Theta(1/p)$ to get a $(\Omega(1/p), O(p^3 \log n))$-DCRS. 

    Note on where $p^3$ in the recourse comes from: Given a sequence of points $\vx^1,\ldots, \vx^T$ in the $p$-matchoid polytope $P$ and corresponding TCOS $\Rset^1,\ldots, \Rset^T$, we first downsample each coordinate with probability $b = \Theta(1/p)$, which yields active sets $\widehat \Rset^1,\ldots, \widehat \Rset^T$ distributed according to $\widehat \vx^1 = b\cdot \vx^1, \ldots, \widehat \vx^T = b\cdot \vx^T$ which lie in the scaled polytope $b\cdot P$.
    Combining the $(b,c,z)$-DCRS for the individual matroids (\Cref{thm:mat_dcrs_formal}) using \cref{thm:combiner}, we get a DCRS that outputs a sequence of feasible solutions $I^1, \ldots, I^T$, and these have expected recourse
    \[\sum_t \expect{\abs{I^t \symdif I^t}} \leq pz \cdot \sum_t \expect{\abs{\widehat R^t \symdif \widehat R^t}} = pzb \cdot \sum_t \expect{\abs{R^t \symdif R^t}}.\]
    When $b = \Theta(1/p)$ and $\eps = \Theta(1/p)$, this indeed yields a $(\Omega(1/p), O(p^3 \log n))$-DCRS. 
    
    In the special case of $p$-uniform hypergraph matching, where the $p$-matchoid is comprised only of partition matroids, again setting $b = \Theta(1/p)$, we get an improved $(\Omega(1/p), O(1))$-DCRS. 
\end{enumerate}

%% file: submod.tex
\section{Low-Recourse Submodular Optimization}

\label{sec:submod}
We conclude with our flagship application of DCRSs to competitive-recourse dynamic submodular optimization. 

\begin{theorem}[DCRS $\Rightarrow$ Low-Recourse Algorithms]
\label{thm:low_rec_submod}
    Let $\cF$ be a down-closed constraint and $\cP \subseteq [0, 1]^\cN$ be a down-closed relaxation of $\cF$ for which there exists a monotone $(b,c,z)$-DCRS. Then, for any constant $\beta \in [0, 1]$, there is an algorithm that gets a sequence of $T$ updates arriving online, where update $t \in [T]$ consists of an arbitrary set $A^t$ of active elements and an arbitrary non-negative monotone submodular function $f^t$, and after receiving every update $t$, produces an independent set $I^t \in \cF$ such that:
    \begin{enumerate}
        \item For all $t \in [T]$, the solution $I^t$ is an approximate maximizer of $f^t$. If $f^t$ is nonnegative monotone submodular then
        \[\expect*{f^t(I^t)} \geq (1-e^{-1}) \cdot bc \cdot (1 - \eps) \cdot \beta  \cdot \opt^t,\] 
        and, if $f^t$ is linear, then
       \[\expect*{f^t(I^t)} \geq bc \cdot (1 - \eps) \cdot \beta  \cdot \opt^t.\]   
        \item The expected recourse is
        \[\sum_{t \in [T]} \expect*{\lvert I^t \symdif I^{t-1}\rvert} \leq O\left(\frac{z}{\eps} \log \frac{n}{\eps}\right) \optr^\beta.\]
        \end{enumerate}
        Here, $\opt^t$ is the maximum value obtained by $f^t$ over all $S \in 2^{A^t} \cap \cF$, and $\optr^\beta$ is a lower bound on the recourse of any algorithm maintaining $f(I^t) \geq \beta  \cdot \opt^t$ for all $t \in [T]$.
\end{theorem}

Before giving the proof, we briefly describe the \emph{positive body chasing} framework introduced by \cite{BBLS23} to maintain low $\ell_1$ movement fractional solutions. In their setting, we are given a sequence of bodies $K_t = \{ \vx^t \in \mathbb{R}_+^n | C^t\vx^t \geq 1, P^t \vx^t \leq 1 \}$ revealed online, where $C^t$ and $P^t$ are non-negative matrices. The authors of \cite{BBLS23} designed an online algorithm for this problem with the following guarantee.
\begin{theorem}[\cite{BBLS23}]
	\label{thm:postive_bodies}
	For any $\eps \in (0, 1]$, there is an $O(\nicefrac{1}{\eps} \log (\nicefrac{d}{\eps}))$-competitive algorithm for chasing positive bodies in $\ell _1$ such that $\vx^t \in K_t^{1 -\eps} = \{ \vx^t \in \mathbb{R}_+^n \mid C^t\vx^t \geq 1-\eps, P^t\vx^t \leq 1 \}$ at time $t$, and $d$ is the maximal number of non-negative coefficients in a covering constraint.\footnote{The original theorem statement in \cite{BBLS23} violates the packing constraints instead of the covering constraints (i.e. $C^t\vx^t \geq 1, P^t\vx^t \leq 1+\eps$) but one can easily convert between the two versions by scaling $\vx$ by a factor of $1+\eps$.}
\end{theorem}

Violating the constraints by a factor of $1-\eps$ was necessary for \cite{BBLS23} to show \cref{thm:postive_bodies}: the violation will translate to a small loss in the approximation factor for us, and will not make a qualitative difference.

\begin{proof}[Proof of \cref{thm:low_rec_submod}]

The algorithm is to use the positive body chasing framework of \cite{BBLS23} to maintain a fractional solution to a relaxation of the problem of optimizing $f^t$, and then to apply the given DCRS for $\cF$.
	
	Since $\cF$ is a down-closed constraint system, the polytope $\cP$ can be expressed with packing constraints $Px \leq 1$ for some matrix $P$ with nonnegative entries. We now formulate the following positive body chasing problem. At time $t$, define the polytope
	\begin{align}
        K_t := \left\{\vx \in \mathbb{R}_{\geq 0}^{\cN} \: \Big\vert  \:
	    P\vx \leq 1, \quad f^t(S) + \sum_{e \in \cN} f^t_S(\{e\}) \, x_e \geq \beta \cdot \opt^t \quad \forall S \subseteq \cN, \quad x_e \leq 0 \quad \forall e \not \in A^t \right\}. \label{line:polytope}
    \end{align}
	  Observe that this is a pure packing/covering constraint system. Furthermore, if we denote by $f^*_t$ and $F^t$ the covering and multilinear extensions of $f^t$, respectively, then the exponential set of covering constraints in $K_t$ enforces that $f^*_t(\vx) \geq \beta \cdot \opt^t$ for every $\vx \in K_t$, which implies $F^t(\vx) \geq (1-e^{-1}) \cdot \beta \cdot \opt^t$ by \cite[Lemma 3.8]{vondrak2007submodularity}.

        Applying the positive body chasing framework of \cite{BBLS23} to $\{K_t\}_t$ produces in an online fashion vectors $\vx^t \in K_t^{1-\eps}$ such that $\sum_{t \in [T]} \|\vx^t-\vx^{t-1}\|_1 \leq  O(\nicefrac{1}{\eps} \log (\nicefrac{n}{\eps})) \cdot \optr^{\beta}$, where $\optr^{\beta}$ is the minimum value the left hand side of this inequality can take for any sequence of vectors $\vx^t \in K_t$. Note that $\optr^{\beta}$ is a lower bound on the recourse of any algorithm maintaining an integral solution $I^t$ respecting the constraint $\cF$ and $f(I^t) \geq \beta  \cdot \opt^t$ for all $t \in [T]$, because $K_t$ is a relaxation of the set of solutions maximizing $f^t$ up to approximation factor $\beta$. Since each vector $\vx^t$ output by \cite{BBLS23} lies in $K_t^{1-\eps}$, by definition we get that $P\vx^t \leq 1$, only active coordinates are nonzero, and $f_t^*(\vx^t) = \min_{S \subseteq \cN} \{f^t(S) + \sum_{e \in \cN} f^t_S(\{e\}) \, x^t_e\} \geq (1-\eps)\cdot \beta \cdot \opt^t$. Again by \cite[Lemma 3.8]{vondrak2007submodularity}, the last condition implies that $F^t(\vx^t) \geq (1-\eps) \cdot (1-e^{-1}) \cdot \beta \cdot \opt^t$.

    	To show the first part of the lemma, we apply the monotone $(b,c,z)$-DCRS whose existence is assumed by the theorem to the sequence $b\cdot \vx^1, \ldots, b \cdot \vx^T$ and return the sets $I^1, \ldots, I^T$ it produces.
        Since $I^t$ is the output of a monotone $(b,c)$-CRS for $\vx^t$, by \cite[Theorem 1.3]{DBLP:journals/siamcomp/ChekuriVZ14}, $\expect{f(I^t)} \geq bc \cdot F(\vx^t) \geq bc (1-e^{-1}) \cdot (1-\eps) \cdot \beta \cdot \opt^t$, as desired.
      
        The second part of the theorem follows directly from the DCRS guarantee that $\sum_t \bE[|I^t \symdif I^{t-1}|] \leq z \cdot \sum_{t\in[T]} \|\vx^t -\vx^{t-1}\|_1 \leq O\left(\frac{z}{\eps} \log \frac{n}{\eps}\right) \optr^{\beta}$.

        To get the improved guarantee for linear $f_t$, one needs to change the definition of $K_t$ for such functions to the simpler polytope 
	    \begin{align}
            K_t := \left\{\vx \in \mathbb{R}_{\geq 0}^{\cN} \quad \Big\vert  \quad
            P\vx \leq 1, \quad F^t(\vx) \geq \beta \cdot \opt^t, \quad x_e \leq 0 \quad \forall e \not \in A^t \right\},
        \end{align}
        which directly implies that $F^t(\vx) \geq \beta \cdot \opt^t$ for every $\vx \in K_t$, and thus, saves the $(1-e^{-1})$ loss factor. The polytope $K_t^{1 - \eps}$ needs to be changed in similar way, and one can verify that all the properties that we need from the polytopes $K_t$ and $K_t^{1 - \eps}$ are preserved by these modifications. \end{proof}

        The astute reader might observe that for monotone submodular functions, $K_t$ involves an exponential number of constraints even when $P$ is specified by a polynomial number of constraints. Though this is not a focus of our paper, we show in \cref{sec:round-or-sep} that, nevertheless, there is an oracle that runs in polynomial time and (approximately) separates $K_t$, which is sufficient to obtain the same guarantees as above.

%% file: full_prelim.tex
\section{Extended Preliminaries}
\label{sec:full_prelim}

\paragraph{Matroids.}

A \emph{matroid} $\cM$ is the pair $\cM = (\cN, \cI)$ 
where $\cN$ is a finite ground set of $n$ elements, and $\cI \subseteq 2^{\cN}$ is a family of \emph{independent} sets that respects the matroid axioms: $\cI$ is nonempty, down-closed,\footnote{The set $\cI$ is down-closed if $S \in \cI$ implies $T \in \cI$ for all $T \subseteq S$.} and satisfies the property that for all $A, B \in \cI$ with $\abs{A} < \abs{B}$, there is some element $e \in B\setminus A$ for which $A \cup \{e\} \in \cI$. 
The \emph{rank} function of $\cM$, denoted by $\rank_\cM\colon 2^{\cN} \rightarrow \mathbb{N}$, is the function $\rank_\cM(S) = \max \{ |I| : I \subseteq S, I \in \cI \}$, i.e. the cardinality of the largest independent subset of $S$. The rank function $\rank_\cM$ is a monotone submodular function (defined below), and has the additional properties that $\rank_\cM(\emptyset) = 0$ and that adding a single element to a set $S$ always increases $\rank_\cM(S)$ by either $0$ or $1$. It is well-known that the polytope 
\[
    P_{\cM} \defeq \{ \vx \in \mathbb{R}^\cN_+: \ \vx(S) \leq \rank_\cM(S) \ \forall S \subseteq \cN, \ x_e \geq 0 \  \forall e \in \cN\}
\]
is the convex hull of the characteristic vectors of the independent sets of $\cM$.

The span of a set $S \subseteq \cN$ with respect to a matroid $\cM$, denoted by $\Span_\cM(S)$, is the set of elements in $\cN$ whose addition to $S$ would not change $\rank_\cM(S)$. When the matroid $\cM$ is clear from the context, we drop the suffix $\cM$ from $\rank_\cM$ and $\Span_\cM$. The rank of the entire ground set $\cN$ is called the \emph{rank of the matroid} $\cM$, and is denoted by $\rank(\cM) \triangleq \rank_\cM(\cN)$. An independent set whose size is $\rank(\cM)$ is called a \emph{base} of $\cM$. By the augmentation axiom of matroids, an independent set $S \in \cI$ is a base of $\cM$ if and only if it is an inclusion-wise maximal independent set. Similarly, a \emph{circuit} of a matroid is an inclusion-wise minimal dependent set---i.e. a set that is not independent, but removing any single element from it produces an independent set. The following is a useful property related to circuits.
\begin{lemma}[(39.35) in~\cite{Schrijver-book}] \label{lem:circuit}
Given an independent set $S$ of a matroid $\cM = (\cN, \cI)$ and an arbitrary element $e \in \cN$. The set $S \cup \{e\}$ contains at most a single circuit of $\cM$.
\end{lemma}

For simplicity of exposition, in this paper we assume $\cM$ contains no \emph{self-loops}, which are circuits of size $1$ (a self-loop element $e$ will never appear in any independent set, and must have $x_e = 0$ for any $\vx \in P_\cM$, hence it can safely be ignored). 

A richer class of down-closed constraint families $\cI$ on a ground set $\cN$ are those which can be represented as a \emph{matroid intersection}, meaning there is a collection of matroids $\{\cM_\alpha\}$ with common $\cN$ such that $A \in \cI$ if and only if $A \in \cI_\alpha$ for all $\cM_\alpha = (\cN, \cI_\alpha)$.
It is also useful to refine this definition to reflect the number of matroids that meaningfully constrain each element $e \in \cN$. 
A down-closed family $(\cN, \cI)$ is a \emph{$p$-matchoid} if it is given by a collection of matroids $\{\cM_\alpha=(\cN_\alpha, \cI_\alpha)\}$ such that $\cN_\alpha \subseteq \cN$, each $e \in \cN$ appears in at most $p$ of the $\cN_\alpha$, and $A \subseteq \cN$ is independent if and only if $A \cap \cN_\alpha \in \cI_\alpha$ for all of the $\cM_\alpha$.

\paragraph{Submodular Functions.}

A set function $f\colon 2^{\cN} \rightarrow \mathbb{R}$
over a ground set $\cN$ is \textit{submodular} if $f(A \cap B) + f(A \cup B) \leq f(A) + f(B)$
for any $A, B \subseteq \cN$.  It is \textit{monotone} if $f(A) \leq f(B)$
for all $A \subseteq B \subseteq \cN$. The \textit{contraction} of $f\colon 2^{\cN} \rightarrow \mathbb{R}$
onto $\cN \setminus T$ is defined as
\[
    f_T(S) = f(S \mid T) \defeq f(S \cup T) - f(T).
\]
If $f$ is submodular then $f_T$ is also submodular for any
$T \subseteq \cN$.
In this work submodular functions are assumed to be nonnegative, unless otherwise stated.
The \textit{multilinear extension} of $f$ is denoted by $F\colon[0,1]^{\cN} \rightarrow \mathbb{R}$ and defined as
\[
	F(\vx) \defeq \expect{f(\Rdist(\vx))} = \sum_{S \subseteq \cN} \prod_{e \in S} x_e \mspace{-9mu} \prod_{e' \in \cN \setminus S} \mspace{-9mu} (1 - x_{e'}) f(S), 
\]
where $\Rdist(\vx)$ contains each element $e \in \cN$ independently with probability $x_e$. The \textit{covering extension} of $f$ (defined but not named in \cite{vondrak2007submodularity}) is:
\[
	f^*(\vx) \defeq \min_{S \subseteq \cN} \bigg\{f(S) + \sum_{e \in \cN} f(\{e\} \mid S) \, x_e\bigg\}. 
\]

It is known (\cite[Lemma 3.8]{vondrak2007submodularity}) that $f^*(\vx) \geq F(\vx) \geq (1-e^{-1}) \cdot f^*(\vx)$ for all $\vx \in [0,1]^\cN$.

\paragraph{Offline Contention Resolution Schemes.}

Chekuri, Vondr\'ak, and Zenklusen \cite{DBLP:journals/siamcomp/ChekuriVZ14} introduced the following concept, which they named Contention Resolution Schemes.

\begin{definition}[CRS] For $b, c \in [0,1]$, a \emph{$(b, c)$-balanced contention resolution scheme (CRS)} for a polytope $P_\cF$ that relaxes a down-closed constraint $\cF$ is an algorithm $\pi$ that receives as input a point $\vx \in b \cdot P_{\cF}$ and a set $\Rset \subseteq \cN$ whose elements are sampled independently according to $\vx$. The algorithm outputs a subset $\pi_\vx(\Rset) \subseteq \Rset$ with the properties that
\begin{enumerate}
    \item $\pi_\vx(\Rset) \in \cF$ with probability $1$.
    \item $\probarg{e \in \pi_\vx(\Rset)}{\Rset \sim \Rdist(\vx)} \geq c \cdot x_e$ for every element $e\in \cN$.
\end{enumerate}
\end{definition}

A concept introduced by \cite{DBLP:conf/sosa/FuLTTWW022} is that of \emph{oblivious} CRS, which is a (possibly randomized) CRS obeying the property that for all $\Rset'$, the distribution of $\pi_\vx(R')$ is the same for all input vectors $\vx$.

\paragraph{Chain Decompositions for Matroids}

A height $\ell$ \emph{chain decomposition} $\cS$ for a matroid $\cM$ over ground set $\cN$ is a sequence $\cN = \cN_0 \supseteq \cN_1 \supseteq \ldots \supseteq N_\ell \supseteq N_{\ell+1} = \emptyset$. 

%% file: elided.tex
\section{Deferred Proofs} \label{app:proofs}

\correctnessprops*

\begin{proof}[Proof of \cref{lem:properties}]
Notice that all the properties stated by the lemma are trivially satisfied when our algorithm is initialized, i.e., when $S_i = \emptyset$ for every $i \geq 1$ and $\Iset_i = \emptyset$ for every $i \geq 0$. We prove below that no operation of \cref{alg:matroid-DCRS} can violate the properties of the lemma, assuming these properties held at all times prior to this operation. For the purpose of this proof, we view an execution of the procedure {\Reset} as single operation.

\paragraph{Property~\eqref{item:S_independent}: $S_i \in \cI$.}

$S_i$ is modified by \cref{alg:matroid-DCRS} in two places. One place is the {\Reset} procedure. If this procedure modifies $S_i$, it makes this set empty, and thus independent. The other place is \cref{line:add}. The condition on \cref{line:add_condition} of the algorithm guarantees that every element $e$ added to $S_i$ by \cref{line:add} does not belong to $\Span(S_i)$ prior to the addition, and therefore, adding it to $S_i$ keeps this set independent.

\paragraph{Property~\eqref{item:laminar}: $\Span(S_i) \subseteq \Span(S_{i - 1})$.}

Again, we notice that $S_i$ and $S_{i - 1}$ are modified by \cref{alg:matroid-DCRS} in two places. One place is the {\Reset} procedure. This procedure can either make $S_i$ or both $S_i$ and $S_{i-1}$ empty; in both cases, the property is maintained. The other place is \cref{line:add}, which can add elements to the sets $S_i$ and $S_{i - 1}$. Clearly, the addition of elements to $S_{i - 1}$ cannot cause the property to become violated. Addition of elements to $S_i$ is more problematic. However, \cref{line:add_condition} guarantees that every element $e$ added to $S_i$ belongs to $\Span(S_{i - 1})$, which means that the span of $S_i$ remains contained within the span of $S_{i - 1}$ after the addition.

\paragraph{Property~\eqref{item:I_independent}: $\Iset_{i - 1} \cup S_i \in \cI$.}

The set $\Iset_{i - 1}$ is modified in multiple places in \cref{alg:matroid-DCRS}. In most of these places the modification consists only of removal of elements, which cannot cause the property to be violated. The sole exception is \cref{line:add_T}, which adds elements to $\Iset_{i - 1}$. However, the condition on \cref{line:T_add_condition} of the algorithm explicitly ensures every element added by \cref{line:add_T} to $\Iset_i$ does not belong (prior to its addition) to $\Span(S_{i} \cup \Iset_{i-1})$. Thus its addition keeps $S_{i} \cup \Iset_{i - 1}$ independent.

Let us now consider the places in \cref{alg:matroid-DCRS} that can modify the sets $S_{i}$. There are two such places. One place is the {\Reset} procedure, but this procedure can only remove elements from $S_i$ (by making it empty), and thus, it cannot cause the property to be violated. The other place is \cref{line:add}. An addition of an element to $S_i$ by this line can transiently make $S_i \cup \Iset_{i-1}$ dependent, but we claim that \cref{line:T_remove} immediately restores the independence of $S_i \cup \Iset_{i - 1}$. Indeed, if \cref{line:add} makes $S_i \cup \Iset_{i - 1}$ dependent, then by Lemma~\ref{lem:circuit}, the set $S_i \cup \Iset_{i - 1}$ contains only one circuit $C$. The circuit $C$ cannot be fully contained by $S_i$ because Property~\eqref{item:S_independent} guarantees that $S_i$ is independent. Thus, there is an element $e' \in C \cap \ISet_{i - 1}$, and the removal of such an element by \cref{line:T_remove} removes the sole circuit of $S_i \cup \ISet_{i - 1}$, and thus, makes it independent.

\paragraph{Property~\eqref{item:I_inclusion}: $\Iset_{i-1} \subseteq \Span(S_{i-1}) \setminus \Span(S_{i})$.}

We have already proved that the property $\Iset_{i - 1} \cup S_{i} \in \cI$ is preserved, which in particular implies that $\Iset_{i - 1}$ does not contain any elements of $\Span(S_{i}) \setminus S_i$. The set $\Iset_{i - 1}$ also does not contain any elements of $S_i$ since \cref{alg:matroid-DCRS} explicitly removes from $\Iset_{i - 1}$ any element added to $S_i$. Thus, we concentrate on proving that $\Iset_{i-1}$ remains a subset of $\Span(S_{i-1})$.

The set $\Iset_{i-1}$ is modified in multiple places in \cref{alg:matroid-DCRS}. In most of these places the modification consists only of removal of elements, which cannot cause a violation for the inclusion $\Iset_{i-1} \subseteq \Span(S_{i-1})$. The sole exception is \cref{line:add_T}, which adds elements to $\Iset_{i-1}$. However, the condition on \cref{line:T_add_condition} of the algorithm explicitly ensures every element added by \cref{line:add_T} belongs to $\Span(S_{i-1})$.
We finally consider lines where \cref{alg:matroid-DCRS} may modify the sets $S_{i-1}$. There are two such places. One place is the {\Reset} procedure. Whenever this procedure modifies $S_{i-1}$ it sets both $S_{i-1}$ and $\Iset_{i-1}$ to be empty, and so the inclusion $\Iset_{i-1} \subseteq \Span(S_{i - 1})$ is preserved. The other place is \cref{line:add}, but this line adds an element to $S_{i-1}$, and thus also preserves the inclusion.
\end{proof}

%% file: lfold_new.tex
\section{DCRS for \texorpdfstring{$k$}{k}-Fold Union Matroids Constraints} 
\label{sec:kfold}

Following a recent line of work \cite{DBLP:conf/sosa/ChekuriSZ24,DBLP:conf/innovations/AlonGPRWW025} on (online) contention resolution schemes for a class of matroids called \emph{$k$-fold unions}, in this section we show that our approach from \Cref{sec:matroid} can be made to provide DCRSs with balance $(1-o_k(1))$ for this class.

\begin{definition}[$k$-fold matroid union]
\label{def:ell-fold-union}
    Given a matroid $\cQ = (\cN, \cI_\cQ)$ and an integer $k \geq 1$, the \emph{$k$-fold union} of $\cQ$ is the matroid 
    \[
        \cQ^k := \underbrace{\cQ \lor \cQ \lor \ldots \lor \cQ}_{k \text{ times}} = (\cN, \cI_\cQ^k),
    \] 
    where $\cI_\cQ^k := \{I_1 \cup I_2 \cup \ldots \cup I_k \mid I_1, I_2, \ldots, I_k \in \cI\}$. In other words, the independent sets of $\cQ^k$ are unions of $k$ independent sets of $\cQ$.
\end{definition}

A simple example: the $k$-uniform matroid is nothing other than the $k$-fold union of the $1$-uniform matroid, and hence $k$-fold union matroids are generalizations of $k$-uniform matroids.

As mentioned above, in this section we assume the matroid $\cM$ for which we design a DCRS is equal to $\cQ^k$ for some matroid $\cQ = (\cN,\cI_\cQ)$. Let us also denote by $\cQ_*$ the extended matroid $(\cN \times [k], \cI_*)$, where $\cI_*$ contains a set $S \subseteq 2^{\cN}\times [k]$ if
\[
	|S \cap (\{u\} \times [k])| \leq 1
	\quad
	\forall\; u \in \cN
	\qquad
	\text{and}
	\qquad
	\{u \in \cN \mid |S \cap (\{u\} \times [k])| = 1\} \in \cI_\cQ
	.
\]
In the following, we use $\Span^*$ and $\Span_\cQ$ to refer to the span with respect to the matroids $\cQ_*^k$ and $\cQ$, respectively. Similarly, we use $\rank^*$ and $\rank_\cQ$ to denote the rank functions of these matroids. The following is a useful observation about these definitions
\begin{observation} \label{obs:equivalence_multi_k}
For every set $S$ that independent in $\cQ$, it holds that
\[
	\Span^*(S \times [k]) = \Span_\cQ(S) \times [k]
	.
\]
\end{observation}
\begin{proof}
Consider an arbitrary element $(e, i) \not \in \Span^*(S \times [k])$. By the choice of $(e, i)$, the set $S \times [k] + (e, i)$ is independent in $\cQ_*^k$, and thus, there is a way to partition $S \times [k] + (e, i)$ into $k$ sets $A_1, A_2, \dotsc, A_k$ that are independent in $\cQ_*$. Assume without loss of generality that $(e, i) \in A_1$. For every element $e' \in S$, since $S \times [k]$ contains $(e', i')$ for all $i' \in [k]$, the definition of $\cQ_*$ guarantees that each one of the sets $A_i$ contains $(e', i')$ for exactly one value $i' \in [k]$. In particular, this is true for $A_1$. Thus,
\[
	S + e
	=
	\{e' \in \cN \mid (e', i') \in A_1\}
	\in
	\cI_\cQ
	,
\]
where the membership in $\cI_\cQ$ follows from the fact that $A_1$ is independent in $\cQ_*$. Thus, $e \not \in \Span_\cQ(S)$, and $(e, i) \not \in \Span_\cQ(S) \times [k]$. Since this is true for an arbitrary element $(e, i) \not \in \Span^*(S \times [k])$, it implies $\Span^*(S \times [k]) \supseteq \Span_\cQ(S) \times [k]$.

To see that inclusion holds also in the reverse direction, notice that for every element $e \not \in \Span_\cQ(S)$, we have that $S + e \in \cI_\cQ$. This implies that $S \times [k] + (e, i)$ is independent in $\cQ_*^k$ because $S \times [k] + (e, i)$ can be partitioned into the $k$ sets $S \times [1] + (e, i), S \times [2], S \times [3], \dotsc, S \times [k]$, and all these sets belong to $\cQ_*$.
\end{proof}

Using the above-defined notation, we can now state, as \cref{alg:lfold-DCRS}, our DCRS for $k$-Fold Union matroids. Compared to \cref{alg:matroid-DCRS}, \cref{alg:lfold-DCRS} has an extra parameter $d \in (0, 1)$. The algorithm denotes by $R_d^t$ a subset of $R^t$ that is a $d\alpha$-time correlated sample of the vector $d \cdot \vx$. Notice that such a sample can be obtained from $\{\Rset^t\}$ in multiple ways. For example, one can preselect a set $D \sim \Rdist(d \cdot \characteristic_\cN)$, and then set $R_d^t = R^t \cap D$. Alternatively, one can add every element that is added to $R^t$ also to $R_d^t$ with probability $d$, and remove from $R_d^t$ every element that is removed from $R^t$. For clarity, we highlight the most important changes from \cref{alg:matroid-DCRS}.
Note that we have also replaced $R$ and $\Span$ with $R_d$ and $\Span_\cQ$ respectively and modified some of the parameter ranges.

\begin{algorithm}[h]
	\caption{\textsc{$k$-foldUnionDCRS}}
	\label{alg:lfold-DCRS}
	\begin{algorithmic}[1]
        \Require matroid $\cM = \cQ^k$, $b \in (0, 1)$, $\tau \in (bk, k)$, $r \in (1, \tau / (bk))$, online sequences $\{\vx^t\}$ and $\{\Rset^t\}$
	      \State{$\sigma \gets$ arbitrary ordering of $\cN$.}
				\State{\tightcolorbox{lightgray}{$\ell \gets \lceil \log_{\tau / (rbk)} \rank(\cM)\rceil$.}}
				\State{Initialize $S_i \gets \emptyset$, $\Iset_i \gets \emptyset$ for all $i\in[\ell]$ and $\Iset_0 \gets \emptyset$.}
        \For{\textbf{each} $t = 1, 2, \ldots, T$}
    		\State $\vx \gets \vx^t \in b \cdot P_{\cM}$ and $\Rset_d \gets \Rset^t_d$ arrive online.
    		\For{$i = 1, 2, \ldots, \ell$ \label{line:fold_main_S_loop}}
                \If{$|S_i| > r \cdot \|\vx|_{\Span_\cQ(S_{i - 1})}\|_1 / \tau$\label{line:fold_early_reset_condition}}
                   \State{\Reset($i$).}
                \EndIf\smash{\hspace{5mm}{\color{lightgray}\rule[-26.7mm]{137.8mm}{16.3mm}}}
                \MultilineWhile{
                \parbox[t]{0.7\linewidth}{$\exists \: e \in \Span_\cQ(S_{i - 1}) \setminus \Span_\cQ(S_i)$ obeying \\ {\tightcolorbox{lightgray}{$\bE[\rank^*(e \times [k] \mid (\Rdist(\vx|_{\Span_\cQ(S_{i - 1})}) \times [1]) \cup (S_{i} \times [k]))] \leq k - \tau$}~~\algorithmicdo} \label{line:fold_add_condition}}}
										\While{there is $e' \in \Iset_{i - 1}$ such that $\rank^*((e', 1) \mid ((\Iset_{i - 1} - e') \times [1]) \cup ((S_i + e) \times [k])) = 0$\label{line:fold_T_remove_condition}}
                    \State Remove $e'$ from $\Iset_{i - 1}$.\label{line:fold_T_remove}
                    \EndWhile
                    \State {Add $e$ to $S_{i}$.\label{line:fold_add}}
                    \If{$|S_i| > r \cdot \|\vx|_{\Span_\cQ(S_{i - 1})}\|_1 / \tau$\label{line:fold_reset_condition}}
                        \State \Reset($i$). 
                    \EndIf
                \EndWhile
            \EndFor
            \For{$i = 0, 1, \ldots, \ell$ \label{line:fold_main_T_loop}}
            	\State 	Update $\Iset_i \gets \Iset_i \cap R_d$.\label{line:fold_non-R_remove}
                \While{there exists \tightcolorbox{lightgray}{$e \in R_d \cap \Span_\cQ(S_{i})$ such that $(e, 1) \not \in \Span^*((S_{i+1} \times [k]) \cup (\Iset_i \times [1]))$}\label{line:fold_T_add_condition}}
                	\State Add the $\sigma$-first such $e$ to $\Iset_i$.\label{line:fold_add_T}
                \EndWhile
            \EndFor
            \Return $\cup_{i=0}^\ell \Iset_i$
		\EndFor 
            \Statex
        \Procedure{Reset}{$i$}
            \For{Level $j=i,\ldots, \ell$}
                \State Set $S_j \gets \emptyset$.
                \State Set $\Iset_{j} \gets \emptyset$.
            \EndFor
        \EndProcedure
	\end{algorithmic}
\end{algorithm}

The analysis of \cref{alg:lfold-DCRS} is based on the following lemma, which is a counterpart of \cref{lem:growth_bound}.

\begin{lemma} 
\label{lem:fold_growth_bound}
    Given a vector $\vy \in [0, 1]^\cN$, consider a process that constructs a set $A$ by starting with $A = \emptyset$, and then adding elements to $A$ one by one. If every element $e \in \cN$ added to $A$ by this process obeys both $e \not \in \Span_\cQ(A_e)$ and $\bE[\rank^*(e \times [k] \mid (\Rdist(\vy) \times [1]) \cup (A_e \times [k]))] \leq k - \tau$, where $A_e$ is the set $A$ before $e$ is added to it, then the final size of the set $A$ is at most $\|\vy\|_1 / \tau$.
\end{lemma}
\begin{proof}
Let us study the evolution of the expectation $\bE[\rank^*(\Rdist(\vy) \times [1] \mid A \times [k])]$ during the execution of the process. Notice that the value of this expectation at the beginning of the process, when $A = \emptyset$, is $\bE[\rank^*(\Rdist(\vy) \times [1])] \leq \bE[|\Rdist(\vy)|] = \|\vy\|_1$. Furthermore, we claim that the expectation $\bE[\rank^*(\Rdist(\vy) \mid A \times [k])]$ decreases by at least $k\tau$ whenever an element is added to $A$. To see that, consider an arbitrary element $e$ added to $A$, and recall that the value of the set $A$ before the addition of $e$ is $A_e$. Thus, the decrease in $\bE[\rank^*(\Rdist(\vy) \times [1] \mid A \times [k])]$ following the addition of $e$ is
\begin{align*}
	\bE[\rank^*(\Rdist(\vy) \times [1] \mid {}&A_e \times [k])] - \bE[\rank^*(\Rdist(\vy) \times [1] \mid (A_e + e) \times [k])]\\
	={} &
	\rank^*((A_e + e) \times [k]) - \rank^*(A_e \times [k]) \\&+ \bE[\rank^*((\Rdist(\vy) \times [1]) \cup (A_e \times [k])) - \rank^*((\Rdist(\vy) \times [1]) \cup (A_e + e) \times [k])]\\
	={} &
	k \cdot |A_e + e| - k \cdot |A_e| - \bE[\rank^*(e \times [k] \mid \Rdist(\vy) \times [1]) \cup (A_e \times [k])]\\
	\geq{} &
	k - (k - \tau)
	=
	\tau
	,
\end{align*}
where the second equality holds since both $A_e$ and $A_e + e$ are independent in $\cQ$, which implies that their Cartesian products with $[k]$ are independent in $\cQ_*^k$.

To summarize, $\bE[\rank^*(\Rdist(\vy) \times [1] \mid A \times [k])]$ is at most $\|\vy\|_1$ at the beginning of the process, and decreases by at least $\tau$ whenever an element is added to $A$. Therefore, its final value is upper bounded by $\|\vy\|_1 - \tau|A|$. Since the expectation $\bE[\rank^*(\Rdist(\vy) \times [1] \mid A \times [1])]$ is always non-negative, this upper bound must be also non-negative, i.e., it must hold that $\|\vy\|_1 - \tau|A| \geq 0$. The lemma now follows by rearranging this inequality.
\end{proof}

Using the last lemma, we can now prove that \cref{alg:lfold-DCRS} always terminates. The proof of the next corollary is similar to the proof of \cref{lem:terminate}, but is somewhat more involved.

\begin{corollary} \label{cor:fold_terminate}
\cref{alg:lfold-DCRS} always terminates.
\end{corollary}
\begin{proof}
\cref{alg:lfold-DCRS} includes only three loops that do not repeat a pre-determined number of times. One loop starts on \cref{line:fold_T_remove_condition}. Each iteration of this loop removes an element from $\Iset_i$, and since $\Iset_i$ is a finite set, this loop can only repeat a finite number of times. The next loop starts on \cref{line:fold_T_add_condition}. Each iteration of this loop picks an element $e \not \in \Iset_i$ and adds it to $\Iset_i$, which increases the size of $\Iset_i$. Since the size of each set $\Iset_i$ cannot exceed the size of the ground set, this loop can only repeat a finite number of times. The third loop starts on \cref{line:fold_add_condition}. Each iteration of this loop picks an element $e \not \in \Span_\cQ(S_i) \supseteq S_i$ and adds it to $S_i$, which increases the size of $S_i$. Each iteration of this loop can also invoke $\Reset(i)$, which decreases the size of $S_i$. However, we prove below that this cannot happen more than once for every value of $i$, which proves that this loop must also terminate after a finite number of iterations since the size of each set $S_i$ also cannot exceed the size of the ground set.

Assume towards a contradiction that \Reset$(i)$ is invoked at least twice by a single execution of \cref{alg:lfold-DCRS}. Let us denote by $\Iset_1$ and $\Iset_2$ two consecutive such invocations, and let $\hat{S}_{i}$ denote the value of the set $S_{i}$ just before $\Iset_2$. Immediately after invocation $\Iset_1$, the set $S_i$ is empty. Consider now an arbitrary element $e$ added by \cref{alg:lfold-DCRS} between this point and the invocation $\Iset_2$, and let $S_{i, e}$ denote the set $S_i$ just before the addition of $e$. By the condition of \cref{line:fold_add_condition} of \cref{alg:lfold-DCRS}, the element $e$ obeys the following two properties.
\begin{itemize}
	\item $e \not \in \Span_\cQ(S_{i, e})$, and
	\item $\bE[\rank^*(e \times [k] \mid (\Rdist(\vx|_{\Span_\cQ(S_{i - 1})}) \times [1]) \cup (S_{i} \times [k]))] \leq k - \tau$.
\end{itemize}
The above properties show together that the process of the growth of $S_i$ between the invocations $\Iset_1$ and $\Iset_2$ has all the properties required from the process by \cref{lem:fold_growth_bound} with respect to the vector $\vy = \vx|_{\Span_\cQ(S_{i - 1})}$. Thus, by this lemma, the size of $S_i$ just before the invocation $\Iset_2$ of \Reset($i$) is upper bounded by $\|\vx|_{\Span_\cQ(S_{i - 1})}\|_1/\tau$. However, this contradicts the fact that \cref{alg:lfold-DCRS} invoked \Reset($i$) at this point because the condition on \cref{line:fold_reset_condition} of the algorithm implies that \Reset($i$) is invoked only when the size of $S_i$ is strictly larger than $r \cdot \|\vx|_{\Span_\cQ(S_{i - 1})}\|_1/\tau \geq \|\vx|_{\Span_\cQ(S_{i - 1})}\|_1/\tau$.
\end{proof}

Next, we need a counterpart for \cref{lem:properties}.

\begin{lemma} \label{lem:fold_properties}
\cref{alg:lfold-DCRS} maintains the invariant that, for every $i \in [\ell + 1]$,
\begin{compactenum}[\bf (a)]
	\item $S_i \in \cI_\cQ$,\label{item:fold_S_independent}
	\item $\Span_\cQ(S_i) \subseteq \Span_\cQ(S_{i - 1})$, \label{item:fold_laminar}
	\item $\rank^*(\Iset_{i - 1} \times [1] \mid S_{i} \times [k]) = |\Iset_{i - 1}|$, which implies that $(\Iset_{i - 1} \times [1]) \cup (S_{i} \times [k])$ is independent in $\cQ_*^k$, and \label{item:fold_I_independent}
	\item $\Iset_{i - 1} \times [1] \subseteq \Span^*(S_{i - 1} \times [k]) \setminus \Span^*(S_{i} \times [k])$. \label{item:fold_I_inclusion}
\end{compactenum}
\end{lemma}
\begin{proof}
Notice that all the properties stated by the lemma are trivially satisfied when our algorithm is initialized, i.e., when $S_i = \emptyset$ for every $i \in [\ell]$ and $\Iset_i = \emptyset$ for every $i \in \{0\} \cup [\ell]$. We show below that no operation of \cref{alg:matroid-DCRS} can violate the properties of the lemma, assuming these properties held at all times prior to this operation, and thus, the properties are maintained throughout the execution of \cref{alg:lfold-DCRS}. For the purpose of this proof, we view an execution of the procedure {\Reset} as single operation.

The proofs that Properties~\eqref{item:fold_S_independent} and~\eqref{item:fold_laminar} are not violated by any operation are almost identical to the arguments used in the proof of \cref{lem:properties} to show that \cref{alg:matroid-DCRS} maintains Properties~\eqref{item:S_independent} and~\eqref{item:laminar} of \cref{lem:properties}. The sole modification that needs to be done in both arguments is that the references to $\Span$ should be replaced with references to $\Span_Q$. The proofs that Properties~\eqref{item:fold_I_independent} and~\eqref{item:fold_I_inclusion} are not violated by any operation of \cref{alg:lfold-DCRS} are given below.

\paragraph{Property~\eqref{item:fold_I_independent}: $\rank^*(\Iset_{i - 1} \times [1] \mid S_{i} \times [k]) = |\Iset_{i - 1}|$.}

The set $\Iset_{i-1}$ is modified in multiple places in \cref{alg:lfold-DCRS}. In most of these places the modification consists only of removal of elements, which cannot cause the property to be violated. The sole exception for that is \cref{line:fold_add_T} of \cref{alg:lfold-DCRS}. This line adds elements to $\Iset_{i-1}$. However, the condition on \cref{line:fold_T_add_condition} of the algorithm makes sure that for every element $e$ added by \cref{line:fold_add_T} to $\Iset_{i-1}$, the pair $(e, 1)$ is not spanned (prior to the addition) by $(S_i \times [k]) \cup (\Iset_{i-1} \times [1])$, and therefore, the addition of $e$ to $\Iset_i$ increases $\rank^*(\Iset_{i - 1} \times [1] \mid S_{i} \times [k])$ by $1$, and thus, preserves the property.

Let us now consider the places in \cref{alg:lfold-DCRS} that can modify the sets $S_i$. There are two such places. One place is the {\Reset} procedure, which only removes elements from $S_i$, and thus, cannot make the property to become violated. The other place is \cref{line:fold_add} of \cref{alg:lfold-DCRS}. Consider the situation when this line adds an element $e$ to $S_{i}$ in a certain iteration of loop starting on \cref{line:fold_add_condition} of \cref{alg:lfold-DCRS}, and let $S_i$ and $\Iset_{i-1}$ denote their values just before this addition. Since we already proved that $S_i$ always belongs to $\cI_Q$, including after the addition of $e$, we get $\rank^*((S_i + e) \times [k]) = |(S_i + e) \times [k]|$. Let us now denote by $v_1, v_2, \ldots, v_{|\Iset_{i-1}|}$ the elements of $\Iset_{i-1}$ in an arbitrary order. Then,
\begin{align*}
	\rank^*(\Iset_{i-1} \times [1] \mid (S_{i} &{}+ e) \times [k])\\
	={} &
	\sum_{j = 1}^{|\Iset_i{i-1}|} \rank^*((v_j, 1) \mid ((S_{i} + e) \times [k]) \cup (\{v_1, v_2, \dotsc, v_{j - 1}\} \times [1]))\\
	\geq{} &
	\sum_{j = 1}^{|\Iset_{i-1}|} \rank^*((v_j, 1) \mid ((S_{i} + e) \times [k]) \cup ((\Iset_{i-1} - v_j) \times [1]))
	=
	\sum_{j = 1}^{|\Iset_{i-1}|} 1
	=
	|\Iset_{i-1}|
	,
\end{align*}
where the inequality follows from the submodularity of $\rank^*$, and the second equality holds since the loop on \cref{line:fold_T_remove_condition} terminated before the addition of $e$ to $S_i$. 

\paragraph{Property~\eqref{item:fold_I_inclusion}: $\Iset_{i-1} \times [1] \subseteq \Span^*(S_{i - 1} \times [k]) \setminus \Span^*(S_{i} \times [k])$.}

We have already proved the property that $\rank^*(\Iset_{i-1} \times [1] \mid S_i \times [k]) = |I_{i - 1}|$, which in particular implies that $\Iset_{i-1} \times [1]$ does not contain any elements of $\Span^*(S_i \times [k])$. Thus, we concentrate on proving that $\Iset_{i - 1} \times [1]$ remains a subset of $\Span(S_{i - 1} \times [k])$. The set $\Iset_{i-1}$ is modified in multiple places in \Cref{alg:lfold-DCRS}. In most of these places the modification consists only of removal of elements, which cannot cause a violation for the inclusion $\Iset_{i - 1} \times [1] \subseteq \Span^*(S_{i - 1} \times [k])$. The sole exception for that is \cref{line:fold_add_T} of \cref{alg:lfold-DCRS}. This line adds elements to $\Iset_{i - 1}$. However, the condition on \cref{line:fold_T_add_condition} of the algorithm makes sure that an element $e$ is added by \cref{line:fold_add_T} to $\Iset_{i - 1}$ only when $e \in \Span_\cQ(S_{i - 1})$, which by \Cref{obs:equivalence_multi_k}, implies $(e, 1) \in \Span^*(S_{i - 1} \times [k])$.

Let us now consider the places in \cref{alg:lfold-DCRS} that can modify the set $S_{i - 1}$. There are two such places. One place is the {\Reset} procedure. Notice that whenever this procedure modifies $S_{i - 1}$, it sets both $S_{i - 1}$ and $\Iset_{i - 1}$ to be empty, which makes the inclusion $\Iset_{i-1} \times [1] \subseteq \Span^*(S_{i - 1} \times [k])$ trivial. The other place is \cref{line:fold_add} of \cref{alg:lfold-DCRS}, but this line adds an element to $S_{i - 1}$, and thus, cannot make the inclusion $\Iset_{i - 1} \times [1] \subseteq \Span^*(S_{i - 1} \times [k])$ stop holding.
\end{proof}

\begin{corollary} \label{cor:fold_feasible}
The set $\Iset = \cup_{i = 0}^{\ell} \Iset_i$ returned by \cref{alg:lfold-DCRS} is always an independent in $\cM$ subset of $R^t_d$ (and thus, also of $R^t$).
\end{corollary}
\begin{proof}
For every $i \in \{0\} \cup [\ell]$, \cref{line:fold_non-R_remove} of \cref{alg:lfold-DCRS} removes from $\Iset_{i}$ every element that does not belong to $R_d$, and then the condition on \cref{line:fold_T_add_condition} of the algorithm explicitly makes sure that only elements of $R_d$ are added to $\Iset_i$. Thus, by the end of the iteration of \cref{alg:lfold-DCRS}, $\Iset_i \subseteq R_d$. Since this is true for every $i \in \{0\} \cup [\ell]$, $\Iset = \cup_{i = 0}^{\ell} \Iset_i$ is also a subset of $R_d = R^t_d$. The rest of this proof is devoted to showing that $\Iset$ is also independent in $\cM$.

For every $i \in \{0\} \cup [\ell]$ and integer $i + 1 \leq i' \leq \ell$, Parts~\eqref{item:fold_laminar} and~\eqref{item:fold_I_inclusion} of \cref{lem:fold_properties} guarantee that $\Iset_{i'} \times [1] \subseteq \Span^*(S_{i'} \times [k]) \subseteq \Span^*(S_{i + 1} \times[k])$. Thus, by the submodularity of the rank function, we have
\[
	\rank^*(\Iset_i \times [1] \mid \cup_{i' = i + 1}^{\ell} (\Iset_{i'} \times [1]))
	\geq
	\rank^*(\Iset_i \times [1] \mid \Span^*(S_{i + 1} \times [k]))
	=
	\rank^*(\Iset_i \times [1] \mid S_{i + 1} \times [k])
	,
\]
where the equality holds since $\rank^*(A \cup \Span(B)) = \rank^*(A \cup B)$ for every two sets $A$ and $B$. Using this inequality, we can now get
\begin{align*}
	\rank^*(I \times [1])
	={} &
	\sum_{i = 0}^\ell \rank^*(\Iset_i \times [1] \mid \cup_{i' = i + 1}^{\ell} (\Iset_{i'} \times [1]))\\
	\geq{} &
	\sum_{i = 0}^\ell \rank^*(\Iset_i \times [1] \mid S_{i + 1} \times [k])
	=
	\sum_{i = 0}^\ell |\Iset_i \times [1]|
	=
	|\Iset \times [1]|
	,
\end{align*}
where the second equality holds by Part~\eqref{item:fold_I_independent} of \cref{lem:fold_properties}, and the last equality holds since Part~\eqref{item:fold_I_inclusion} of \cref{lem:fold_properties} implies that the sets $\Iset_1, \Iset_2, \dotsc, I_\ell$ are disjoint.

The above inequality shows that $\Iset \times [1]$ is independent in $\cQ_*^k$. In other words, $I \times [1]$ can be discomposed into $k$ disjoint sets $A^*_1, A^*_2, \dotsc, A^*_k$ such that each one of them is independent in $\cQ_*$. Let us now define, for every $i \in [k]$, $A_i = \{u \in \cN \mid (u, 1) \in A^*_i\}$. By the definition of $\cQ_*$, the sets $A_1, A_2, \dotsc, A_k$ are all independent in $\cQ$. Furthermore, every element of $\Iset$ belongs to one of these sets, and thus, $\Iset$ is independent in $\cQ^k = \cM$.
\end{proof}

The next two subsections analyze the balance and recourse of \cref{alg:lfold-DCRS}. They are similar to \cref{ssc:balance} and \cref{sec:recourse}, but require some modifications. Technically, it is also necessary to prove that \cref{alg:lfold-DCRS} is monotone, but we omit this proof since it is very similar to the corresponding proof regarding \cref{alg:matroid-DCRS} from \cref{sec:matroid-dcrs-monotone}.

\subsection{Balance}

In this section we analyze the balance of \cref{alg:lfold-DCRS}. We begin with the following observation.
\begin{observation} \label{obs:wrong_rang_penalty}
For every set $A \subseteq \cN$, $\rank_\cM(\Span_\cQ(A)) \leq k \cdot \rank_\cQ(A)$.
\end{observation}
\begin{proof}
Let $B$ be a base of $\Span_\cQ(A)$ with respect to $\cM$. Since $B$ is independent in $\cM$, there must be a partition of $B$ into $k$ sets $B_1, B_2, \dotsc, B_k$ such that all the sets in this partition are independent in $\cQ$. Notice now that for every $j \in [k]$, $B_j$ is a subset of $\Span_\cQ(A)$, and thus, $|B_j| = \rank_\cQ(B_j) \leq \rank_\cQ(\Span_\cQ(A)) = \rank_\cQ(A)$. The observation now follows since the definition of $B$ implies that $\sum_{j = 1}^k |B_j| = |B| = \rank_\cM(\Span_\cQ(A))$.
\end{proof}

Using the last observation, we can now prove the following counterpart of \cref{lem:set_size_bound}.
\begin{lemma} \label{lem:fold_set_size_bound}
When \cref{alg:lfold-DCRS} terminates the processing of $\vx^t$, $|S_i| \leq (rbk/\tau)^{i - j} \cdot |S_j|$ for every $i \in [\ell]$ and integer $0 \leq j \leq i$. In particular, $|S_\ell| \leq (rbk/\tau)^\ell \cdot |S_0| = (rbk/\tau)^\ell \cdot \rank(\cN) \leq 1$.
\end{lemma}
\begin{proof}
For every $i \in [\ell]$,
\[
	\norm{\vx|_{\Span_\cQ(S_{i - 1})}}_1
	\leq
	b \cdot \rank_\cM(\Span_\cQ(S_{i - 1}))
	\leq
	bk \cdot \rank_\cQ(S_{i - 1})
	=
	bk \cdot |S_{i - 1}|
	,
\]
where the first inequality holds since $\vx \in b \cdot \cP_\cM$, and the second inequality follows from \cref{obs:wrong_rang_penalty}.
\Cref{line:fold_early_reset_condition,line:fold_reset_condition} of \cref{alg:lfold-DCRS} now guarantee together that when iteration $i$ of the loop starting on \cref{line:fold_main_S_loop} of the algorithm terminates, we have the inequality 
\[
	|S_i|
	\leq
	\frac{r \cdot \|\vx|_{\Span_\cQ(S_{i - 1})}\|_1}{\tau}
	\leq
	\frac{r \cdot (bk \cdot |S_{i - 1}|)}{\tau}
	=
	\frac{rbk}{\tau} \cdot |S_{i - 1}|
	.
\]
Since $S_i$ and $S_{i - 1}$ do not change after this point in the processing of $\vx^t$ by \cref{alg:lfold-DCRS}, the last inequality holds also when this processing terminates. Furthermore, by combining this inequality for different values of $i$, we get for every $i \in [\ell]$ and integer $0 \leq j \leq i$ that
\[
	|S_i|
	\leq
	\Big(\frac{rbk}{\tau}\Big)^{i - j} \cdot |S_j|
	,
\]
which completes the proof of the first part of the lemma. To see why the second part of the lemma holds as well, notice that
\[
	|S_\ell|
	\leq
	\Big(\frac{rbk}{\tau}\Big)^\ell \cdot |S_0|
	=
	\Big(\frac{rbk}{\tau}\Big)^\ell \cdot \rank(\cN)
	\leq
	\Big(\frac{rbk}{\tau}\Big)^{\log_{\tau / (rbk)} \rank(\cN)} \cdot \rank(\cN)
	=
	1
	.
	\qedhere
\]
\end{proof}

Following is a corollary of \cref{lem:fold_set_size_bound}. In this corollary, and in the rest of the analysis of \cref{alg:lfold-DCRS}, we use a definition of $i(e)$ similar to the definition used \cref{sec:matroid}, but based on $\Span_\cQ$ instead of $\Span_\cM$. In other words, for every element $e \in \cN$, $i(e)$ is the maximal integer such that $e \in \Span_\cQ(S_{i(e)})$. Despite the change in the definition of $i(e)$, it still holds that $0 \leq i(e) \leq \ell$ for every element $e$ of $\cQ$.

\begin{corollary} \label{cor:fold_prob_span}
For every element $e \in \cN$, $\bE[\rank^*(e \times [k] \mid (\Rdist(\vx|_{\Span_\cQ(S_{i(e)})}) \times [1]) \cup (S_{i(e) + 1} \times [k]))] > k - \tau$. 
\end{corollary}
\begin{proof}
Consider first the case that $i(e) < \ell$, in this case we know that in iteration number $i(e) + 1$ of the loop starting on \cref{line:fold_main_S_loop} of \cref{alg:lfold-DCRS}, the internal loop on \cref{line:fold_add_condition} terminated, which implies that $e$ does not obey the condition of this internal loop. However, by the definition of $i(e)$, $e \in \Span_\cQ(S_{i(e)}) \setminus \Span_\cQ(S_{i(e) + 1})$, and thus, the part of the condition that $e$ violates must be the inequality $\bE[\rank^*(e \times [k] \mid (\Rdist(\vx|_{\Span_\cQ(S_{i(e)})}) \times [1]) \cup (S_{i(e) + 1} \times [k]))] \leq k - \tau$.

The rest of this proof is devoted to the case of $i(e) = \ell$. In this case, $S_{i(e) + 1} = S_{\ell + 1} = \emptyset$ and
\[
	\|\vx|_{\Span_\cQ(S_{i(e)})}\|_1
	=
	\|\vx|_{\Span_\cQ(S_{\ell})}\|_1
	\leq
	b \cdot \rank_\cM(\Span_\cQ(S_{\ell}))
	\leq
	bk \cdot \rank(S_{\ell})
	=
	bk \cdot |S_{\ell}|
	\leq
	bk
	,
\]
where the first inequality holds since $\vx \in b \cdot \cP_\cM$, the second inequality holds by \cref{obs:wrong_rang_penalty} and the last inequality follows from \cref{lem:fold_set_size_bound}. Thus,
\begin{align*}
	\bE[\rank^*(&e \times [k] \mid (\Rdist(\vx|_{\Span_\cQ(S_{i(e)})}) \times [1]) \cup (S_{i(e) + 1} \times [k]))]\\
	={} &
	\bE[\rank^*(e \times [k] \mid \Rdist(\vx|_{\Span_\cQ(S_{i(e)})}) \times [1])]
	\geq
	\rank^*(e \times [k]) - \bE[|\Rdist(\vx|_{\Span_\cQ(S_{i(e)})}) \times [1]|]\\
	={} &
	\rank^*(e \times [k]) - \|\vx|_{\Span_\cQ(S_{i(e)})}\|_1
	\geq
	k - bk
	>
	k - \tau
	,
\end{align*}
where the penultimate inequality uses the fact that since $e$ is not a self-loop in $\cM = \cQ^k$, $\{e\}$ is independent in $\cQ$, and thus, $\{e\} \times [1]$ is independent in $\cQ_*^k$.
\end{proof}

We now need the following known result.

\begin{theorem}[{\cite[Theorem 3]{DBLP:conf/innovations/AlonGPRWW025}}]
    \label{thm:bicrit_conc}
        Let $f\colon \{0,1\}^\cN \rightarrow \mathbb{R}$ be a monotone $1$-Lipschitz function. Let $X_1 \sim \ber(p_1), \ldots, X_n \sim \ber(p_n)$ be i.i.d.\ Bernoulli random variables, and for any $s \geq 0$, let $X_1^{(s)} \sim \ber(e^{-s} p_1), \ldots, X_n \sim \ber(e^{-s} \cdot p_n)$ be Bernoullis with scaled down parameters. 
        Then, for any $s \in (0,1]$ and $t > 0$,
        \[
            \Pr[f(X^{(s)}_1, \ldots, X^{(s)}_n) \geq \bE[f(X_1, \ldots, X_n)] + t] \leq e^{-st}.
        \]
    \end{theorem}

Using this theorem, we can prove the balance of our algorithm.

\begin{proposition} \label{prop:balance_fold}
\cref{alg:lfold-DCRS} has a balance of $d(1 - d^{k - \tau})$.
\end{proposition}
\begin{proof}
Consider an arbitrary element $e \in \cM$. We need to prove that $\prob{e \in \Iset^t \mid e \in R} \geq x_ed(1 - d^{k - \tau})$. By \cref{cor:fold_prob_span},
\begin{align} \label{eq:chernoff_based_bound}
	\Pr[(e, 1) \in \Span^*(&((R_d \cap \Span_Q(S_{i(e)}) - e) \times [1]) \cup (S_{i(e) + 1} \times [k]))]\\\nonumber
	={} &
	\Pr[(e, 1) \in \Span^*(((\Rdist(d \cdot \vx|_{\Span_Q(S_{i(e)})} - e)) \times [1]) \cup (S_{i(e) + 1} \times [k]))]\\\nonumber
	={} &
	\Pr[\rank^*(\{e\} \times [k] \mid ((\Rdist(d \cdot \vx|_{\Span_Q(S_{i(e)})}) - e) \times [1]) \cup (S_{i(e) + 1} \times [k])) = 0]\\\nonumber
	\leq{} &
	\Pr[\rank^*(\{e\} \times [k] \mid (\Rdist(d \cdot \vx|_{\Span_Q(S_{i(e)})}) \times [1]) \cup (S_{i(e) + 1} \times [k])) = 0]\\\nonumber
    \leq{} &
	\Pr\big[\rank^*(\{e\} \times [k] \mid (\Rdist(d \cdot \vx|_{\Span_Q(S_{i(e)})}) \times [1]) \cup (S_{i(e) + 1} \times [k])) \\\nonumber&\mspace{36mu}\leq \bE[\rank^*(\{e\} \times [k] \mid (\Rdist(\vx|_{\Span_\cQ(S_{i(e)})}) \times [1]) \cup (S_{i(e) + 1} \times [k]))] - (k - \tau)\big] \\\nonumber
	\leq{} &
	e^{-(-\ln d) \cdot (k - \tau)}
	=
	d^{k - \tau}
	,
\end{align}
where the first equality holds since one can verify that while the values of the sets $\Iset_0, \Iset_1, \dotsc, \Iset_{\ell}$ depend on the value of $R_d$, the construction of the sets $S_1, S_2, \dotsc, S_\ell$ (and thus, also $i(e)$) is independent of the value of $R_d$; the second equality holds since for every set $S \subseteq (\cN - e) \times [k]$ and two integers $i, j \in [k]$, $(e, i) \in \Span^*(S)$ if and only if $(e, j) \in \Span^*(S)$; the first inequality follows from the submodularity of $\rank^*$; and the last inequality follows from \cref{thm:bicrit_conc} because $-\rank^*(\{e\} \times [k] \mid (S \times [1]) \cup (S_{i(e) + 1} \times [k]))$ is a monotone $1$-Lipschitz function of $S \subseteq \cN$.

Combining this observation with the fact that Part~\eqref{item:fold_I_inclusion} of \cref{lem:fold_properties} guarantees that $\Iset_{i(e)} \times [1]$ is a subset of $\Span^*(S_{i(e)} \times [k])$ and \cref{cor:fold_feasible} guarantees that $\Iset_{i(e)} \subseteq \cup_{i = 0}^\ell \Iset_i \subseteq R_d$, we get
\begin{align*}
	\Pr[e \in R_d &\text{ and } (e, 1) \not \in \Span^*(((\Iset_{i(e)} - e) \times [1]) \cup (S_{i(e) + 1} \times [k]))]\\
	\geq{} &
	\Pr[e \in R_d \text{ and } (e, 1) \not \in \Span^*(((R_d \times [1]) \cap \Span^*(S_{i(e)} \times [k]) - (e, 1)) \cup (S_{i(e) + 1} \times [k]))]\\
	={} &
	\Pr[e \in R_d \text{ and } (e, 1) \not \in \Span^*(((R_d \cap \Span_Q(S_{i(e)}) - e) \times [1]) \cup (S_{i(e) + 1} \times [k]))]\\
	={} &
	\Pr[e \in R_d] \cdot \Pr[(e, 1) \not \in \Span^*(((R_d \cap \Span_Q(S_{i(e)}) - e) \times [1]) \cup (S_{i(e) + 1} \times [k]))]\\
	\geq{} &
	dx_e \cdot (1 - d^{k - \tau})
	,
\end{align*}
where the first equality uses \cref{obs:equivalence_multi_k}, the second equality holds since the membership of each element in $R_d$ is independent, and the last inequality holds due to Inequality~\eqref{eq:chernoff_based_bound}.

To complete the proof of the proposition, it only remains to observe that the event that both $e \in R_d$ and $e \not \in \Span^*(((\Iset_{i(e)} - e) \times [1]) \cup (S_{i(e) + 1} \times [k]))$ implies that $e \in \Iset_{i(e)}$, and thus, $e \in \cup_{i = 0}^\ell \Iset_i = \Iset^t$. Assume towards a contradiction that this is not the case, then
\[
	(e, 1)
	\not \in
	\Span^*(((\Iset_{i(e)} - e) \times [1]) \cup (S_{i(e) + 1} \times [k]))
	=
	\Span^*((\Iset_{i(e)} \times [1]) \cup (S_{i(e) + 1} \times [k]))
	,
\]
and additionally, by the definition of $i(e)$, $e \in \Span_\cQ(S_{i(e)})$. This contradicts the fact that in iteration number $i(e) + 1$ of the loop starting on \cref{line:fold_main_T_loop} of \cref{alg:lfold-DCRS}, the internal loop on \cref{line:fold_T_add_condition} terminated without adding $e$ to $\Iset_{i(e)}$.
\end{proof}

\subsection{Recourse}

In this section, we analyze the recourse of \cref{alg:lfold-DCRS}. We begin with the following counterpart of \cref{lem:remaining}. The proof of this lemma follows the same line of arguments as the proof of \cref{lem:remaining}, but now we need to use \cref{lem:fold_growth_bound} and \cref{lem:fold_set_size_bound} instead of \cref{lem:growth_bound} and \cref{lem:set_size_bound}.

\begin{lemma} \label{lem:fold_remaining}
When \cref{alg:lfold-DCRS} terminates, $\sum_{i = 1}^\ell |S_\ell| \leq \frac{\Inc(T)}{\tau - rbk}$.
\end{lemma}

We can also get the following counterpart of \cref{lem:removed}.

\begin{lemma} \label{lem:fold_removed}
The total number of elements \Cref{alg:lfold-DCRS} removes from the sets $S_1, S_2, \dotsc, S_\ell$ satisfies 
    \[
        \sum_{t \in [T]} \sum_{i \in [\ell]} \abs{S^{t-1}_i \setminus S^t_i } \leq \frac{\ell \cdot \Dec(T)}{(\tau - bk)(1 - 1/r)}
        .
    \]
\end{lemma}
\begin{proof}
The proof of this lemma follows the same line of argument as the proof of \cref{lem:removed}, with every use of \cref{lem:fold_growth_bound} replaced with a use of \cref{lem:growth_bound}. One complication, however, is that we can no longer bound $\|\vx|_{\Span_\cQ(S_i)}\|_1$ by $b \cdot |S_i|$. Instead, we need to use \cref{obs:wrong_rang_penalty} to bound this expression by $bk \cdot |S_i|$.
\end{proof}

Replacing \cref{lem:remaining} and \cref{lem:removed} with \cref{lem:fold_remaining} and \cref{lem:fold_removed} in the proof of \cref{cor:S_i_additions}, we get the following counterpart.

\begin{corollary} 
\label{cor:fold_S_i_additions}
    The total number of additions of elements to the sets $S_1, S_2, \dotsc, S_\ell$ satisfies
    \[
        \sum_{t \in [T]} \sum_{i \in [\ell]} \abs{S^t_i \setminus S^{t-1}_i} 
        \leq \frac{\Inc(T)}{\tau - rbk} + \frac{\ell \cdot \Dec(T)}{(\tau - bk)(1 - 1/r)}
        .
    \]
\end{corollary}

Our next goal is to bound the number of changes in the sets $\Iset_0, \Iset_2, \dotsc, \Iset_{\ell}$. 
\begin{lemma} \label{lem:fold_T_removals}
The total number of removals from the output set is bounded by
\[
        \expect*{\sum_t \abs{\Iset^{t-1} \setminus \Iset^{t}} }
        \leq \frac{k \cdot \Inc(T)}{\tau - rbk} + \frac{2k\ell \cdot \Dec(T)}{(\tau - bk)(1 - 1/r)} + d \cdot \sum_{t} \abs{\Rset^t \setminus \Rset^{t-1}}
        .
    \]
\end{lemma}
\begin{proof}
\cref{alg:lfold-DCRS} removes elements from the sets $\Iset_0, \Iset_2, \dotsc, \Iset_{\ell}$ only in the following three places.
\begin{itemize}
	\item The loop starting on \cref{line:fold_T_remove_condition} of \cref{alg:lfold-DCRS} removes elements from $\Iset_i$. The condition of this loop guarantees that the removal of these elements does not decrease $\rank^*((\Iset_i \times [1]) \cup ((S_{i+1} + e) \times [k]))$, where $e$ is the element is added to $S_i$ immediately after the loop by \cref{line:fold_add}. Let us denote by $D$ the set of elements removed by a particular execution of this loop. Since Property~\eqref{item:fold_I_independent} of \cref{lem:fold_properties} guarantees that $(\Iset_i \times [1]) \cup (S_{i+1} \times [k])$ is independent in $\cQ_*^k$ at all times during the execution of \cref{alg:lfold-DCRS},
	\begin{align*}
		|(\Iset_i \times [1]) \cup (S_{i+1} \times [k])|
		={} &
		\rank^*((\Iset_i \times [1]) \cup (S_{i+1} \times [k]))\\
		\leq{} &
		\rank^*((\Iset_i \times [1]) \cup ((S_{i+1} + e) \times [k]))\\
		={} &
		\rank^*(((\Iset_i \setminus D) \times [1]) \cup ((S_{i+1} + e) \times [k]))\\
		={} &
		|((\Iset_i \setminus D) \times [1]) \cup ((S_{i+1} + e) \times [k])|
		.
	\end{align*}
	Notice that the last inequality can hold only when $|D| \leq k$, and thus, the number of elements removed by the loop on \cref{line:fold_T_remove_condition} of \cref{alg:lfold-DCRS} is upper bounded by $k$ times the number of elements added to the sets $S_1, S_2, \dotsc, S_\ell$ by \cref{line:fold_add}, which is at most $\frac{\Inc(T)}{\tau - rbk} + \frac{\ell \cdot \Dec(T)}{(\tau - bk)(1 - 1/r)}$ by \cref{cor:fold_S_i_additions}.
	\item \cref{line:fold_non-R_remove} removes elements of $\Iset_i$ that do not belong to $R_d$. Since \cref{alg:lfold-DCRS} adds elements to $\Iset_i$ only on \cref{line:fold_add_T}, which appears after \cref{line:fold_non-R_remove}, every element removed from $\Iset_i$ by \cref{line:fold_non-R_remove} must be in $\Iset_i \setminus R_d$ already at the beginning of the iteration of \cref{alg:lfold-DCRS}. However, \cref{cor:fold_feasible} guarantees that $\Iset_i \setminus R_d$ was empty when the previous iteration of \cref{alg:lfold-DCRS} terminated. Thus, every element removed by \cref{line:fold_non-R_remove} must have been removed from $R_d$ between the iterations (i.e., it is an element of $R^{t - 1}_d \setminus R^t_d$), and thus, the number of such removals is bounded by $\sum_{t} \abs{\Rset_d^t \setminus \Rset_d^{t-1}} \leq d \cdot \sum_{t} \abs{\Rset^t \setminus \Rset^{t-1}}  $. 
	\item Whenever the procedure $\Reset$ removes all the elements of a set $\Iset_i$, it also removes all the elements of $S_{i}$. By Parts~\eqref{item:fold_I_independent} and~\eqref{item:fold_I_inclusion} of \cref{lem:fold_properties}, $S_{i - 1} \times [k]$ and $\Iset_i \times [1]$ are both independent in $\cQ_*^k$ and $\Iset_i \times [1] \subseteq \Span^*(S_{i - 1} \times [k])$. Together, these facts imply
	\[
		|\Iset_i|
		=
		\rank^*(\Iset_i \times [1])
		\leq
		\rank^*(\Span^*(S_{i - 1} \times [k]))
		=
		\rank^*(S_{i - 1} \times [k])
		=
		k|S_{i - 1}|
		.
	\]
	Therefore, the total number of elements removed by $\Reset$ from the sets $\Iset_1, \dotsc, \Iset_{\ell}$ is upper bounded by $k$ times the number of elements it removes from the sets $S_1, S_2, \dotsc, S_{\ell}$ during this processing, which is at most $\frac{\ell \cdot \Dec(T)}{(\tau - bk)(1 - 1/r)}$ by \cref{lem:fold_removed}.
\end{itemize}

The lemma now follows by adding up the above bounds on the (expected) number of elements removed from the sets $\Iset_1, \Iset_2, \dotsc, \Iset_{\ell + 1}$ during the processing of the first $m$ update steps by the loop on \cref{line:fold_T_remove_condition} of \cref{alg:lfold-DCRS}, the procedure {\Reset}, and \cref{line:fold_non-R_remove} of \cref{alg:lfold-DCRS}.
\end{proof}

We are now ready to bound the recourse of our algorithm.
\begin{proposition} \label{prop:recours_bound_fold}
The total expected recourse of \cref{alg:lfold-DCRS} is at most 
\[
    O\left( \frac{k\ell r}{(\tau - rbk)(r-1)} \right) \sum_{t\in[T]} \expect{\abs{\Rset^t \setminus \Rset^{t-1}}}.
\]
In the increase-only case, this improves to 
\[ 
    O\left( \frac{k}{\tau - rbk} \right) \sum_{t} \expect{\abs{\Rset^t \setminus \Rset^{t-1}}}. 
\]
\end{proposition}
\begin{proof} 
Just as in the proof of \cref{prop:recourse_bound}, the number of total changes in $I = \cup_{i = 0}^\ell \Iset_i$ is at most twice the number of removals from this set plus whatever remains at the end. Furthermore, the number for elements that remain in $I$ when \cref{alg:lfold-DCRS} terminates is upper bounded by $|R_d| \leq |R|$, and thus, is at most $\Inc(T)$. Combining these observations with \cref{lem:fold_T_removals}, we get
\begin{align*}
    \sum_{t} \expect{\abs{\Iset^t \setminus \Iset^{t-1}}} 
    &\leq 2 \cdot \sum_{t} d \cdot \expect{\abs{\Rset^t \setminus \Rset^{t-1}}} + \left(\frac{2k}{\tau - rbk} + 1\right) \cdot \Inc(T) + \frac{4k\ell \cdot \Dec(T)}{(\tau - bk)(1 - 1/r)} \\
    &\leq 2 \cdot \sum_{t} \expect{\abs{\Rset^t \setminus \Rset^{t-1}}} + O\left( \frac{k\ell r}{(\tau - rbk)(r-1)} \right) \cdot(\Inc(T) + \Dec(T)) \\
    &= O\left( \frac{k\ell r}{(\tau - rbk)(r-1)} \right) \cdot \sum_{t} \expect{\abs{\Rset^t \setminus \Rset^{t-1}}}
    .
\end{align*}
Reevaluating this bound when $\Dec(T) = 0$, we get 
\begin{align*}
    \sum_{t} \expect{\abs{\Iset^t \setminus \Iset^{t-1}}} &\leq \left(\frac{2k}{\tau - rbk} + 1\right) \Inc(T) + 2d \sum_{t} \expect{\abs{\Rset^t \setminus \Rset^{t-1}}} \\
    &\leq O\left( \frac{k}{\tau - rbk} \right) \sum_{t} \expect{\abs{\Rset^t \setminus \Rset^{t-1}}} 
    . \qedhere
\end{align*}
\end{proof}

\subsection{Setting the Parameters}
To conclude, we obtain a guarantee analogous to \Cref{thm:mat_dcrs_formal}:
\begin{theorem}
\label{thm:kfold-mat_dcrs_formal}
    Let $b \in(0, 1 - 1/k)$, and let $\eps > 0$ be small enough so that $(1 + \eps)b + \eps < 1 - 1/k$. There is an assignment of values to the parameters of \Cref{alg:lfold-DCRS}, which makes this algorithm a monotone $(b, c)$-CRS for any TCOS and for every $t \in [T]$, where 
    \[
        c = 1 - \frac{1 + \ln(k(1 - b))}{k(1 - b - 2 \eps)}
        . 
    \]
    Furthermore, its expected recourse is bounded by 
    \[
        \sum_{t = 1}^T \expect*{\abs{\Iset^t \triangle \Iset^{t-1}}} \leq z \cdot \sum_{t = 1}^T \expect*{\abs{\Rset^t \triangle \Rset^{t-1}}}
        ,
    \]
    for $z = O( \eps^{-3} \cdot \log \rank(\cN))$.
    For monotonously increasing $\{\vx^t\}$, the recourse improves to $z = O(\eps^{-1})$.
\end{theorem}
\begin{proof}
We choose $r = 1 + \eps$, $\tau = k\left((1 + \eps)b + \eps\right)$, and $d = 1 - \frac{\ln(k - \tau)}{(k - \tau)}$. Notice that $k - \tau > 1$, and therefore, $d \in (0, 1)$. This choice of parameters yields
\[
	\ell
	=
	\lceil \log_{\tau / (rbk)} \rank(\cN)\rceil
	=
	\lceil \log_{1 + \eps/(rb)} \rank(\cN)\rceil
	\leq
	\lceil \log_{1 + \eps/(2b)} \rank(\cN)\rceil
	=
	O(b\eps^{-1} \log \rank(\cN))
	,
\]
and thus, by \cref{prop:recours_bound_fold}, leads to a recourse of
\begin{align*}
    O\left( \frac{rk\ell}{(\tau - rbk)(r-1)} \right) 
    &= O\left( \frac{(1 + \eps)k\ell}{k \cdot \left(\left((1 + \eps)b + \eps\right) - (1 + \eps)b\right) \cdot \left((1 + \eps) - 1\right)}\right) \\
    &= O\left(\frac{(1 + \eps)\ell}{\eps^2}\right) = O(\ell/\eps^2) \\
    &= O( (b/\eps^3) \cdot \log \rank(\cN))
\end{align*}
In the incremental case, the recourse improves to
\begin{align*}
    O\left( \frac{k}{\tau - rbk} \right) &= O\left( \frac{1}{(1 + \eps)b + \eps - (1 + \eps)b} \right) 
    = O\left(\eps^{-1}\right)
    . 
\end{align*}
By \cref{prop:balance_fold}, the balance implied by the above choice of values for the parameters is
\begin{align*}
	d(1 - d^{k - \tau})
	={} &
	\bigg(1 - \frac{\ln(k - \tau)}{k - \tau}\bigg) \cdot \bigg[1 - \bigg(1 - \frac{\ln(k - \tau)}{k - \tau}\bigg)^{k - \tau}\bigg]
	\geq
	\bigg(1 - \frac{\ln(k - \tau)}{k - \tau}\bigg) \cdot (1 - e^{-\ln(k - \tau)})\\
	={} &
	\bigg(1 - \frac{\ln(k - \tau)}{k - \tau}\bigg) \cdot \bigg(1 - \frac{1}{k - \tau}\bigg)
	\geq
	1 - \frac{1 + \ln(k - \tau)}{k - \tau}\\
	\geq{} &
	1 - \frac{1 + \ln(k(1 - b))}{k(1 - b - 2\eps)}. \qedhere
\end{align*}
\end{proof}

For $k \geq 2$, an interesting choice is $b = 1 - \sqrt{\frac{\ln k}{k}}$ (which is indeed smaller than $1 - 1/k$). For this value of $b$, we may choose for \cref{thm:kfold-mat_dcrs_formal} any positive $\eps$ smaller than
\[
	\frac{1 - b - 1/k}{1 + b}
	=
	\frac{\sqrt{k\ln k} - 1}{2k - \sqrt{\ln k}}
	,
\]
and thus, we can choose $\eps = \frac{\sqrt{\ln k / k}}{16}$. The balance implied by \cref{thm:kfold-mat_dcrs_formal} for this choice of $\eps$ is
\begin{align*}
	1 - \frac{1 + \ln(k(1 - b))}{k(1 - b - 2\eps)}
	={} &
	1 - \frac{1 + \ln(\sqrt{k \ln k})}{\sqrt{k \ln k}/8}
	=
	1 - O\bigg(\sqrt{\frac{\log k}{k}}\bigg)
	,
\end{align*}
and the recourse is $O(\eps^{-3} \cdot \log \rank(\cN)) = O((k/\log k)^{3/2} \cdot \log \rank(\cN))$ in general and $O(\eps^{-1}) = O(\sqrt{k/\log k})$ for monotonically increasing $\vx^t$. Using \cref{rmk:remove_b}, we now get the DCRS whose parameters are reported by \cref{thm:lfold-union_dcrs_informal}, and applying \Cref{thm:low_rec_submod} to this DCRS gives the approximation and recourse stated in \Cref{tab:app}.

%% file: adversary.tex
\section{Impossibility Against Adaptive Adversaries}
\label{sec:adversary-lb}

The results we present for solving fully dynamic problems with low competitive recourse are against a \emph{non-adaptive} adversary, meaning that the sequence of element additions/deletions and changes to the objective function are adversarially chosen but fixed at the outset.
The same is true for the competitive recourse results of \cite{BBLS23} and Buchbinder, Levin, and Yang \cite{buchbindercompetitively}.

We show this non-adaptivity assumption is necessary by presenting an instance of the online problem from \cref{thm:low_rec_submod} for which the competitive recourse must be at least $\Omega(n)$.
We also adapt this instance to show that no $o(n)$-TCOS exists for adaptive online adversaries, which shows that the idea of DCRSs is probably not useful against this kind of adversaries even for applications other than dynamic packing problems.

\begin{example}[$1$-uniform matroid with element deletions]
\label{ex:adaptive-lower-bound-1-uniform}
    Consider the $1$-uniform matroid on the $n$-element ground set $\cN$ (i.e., a set is independent in the matroid if it contains up to a single element) and an objective function that simply returns the number of elements in the solution. In the first time step, all the elements are in the active set, but at each subsequent time step one of the elements is removed from the active set, until this set contains only a single element (notice that $T = n$ in this instance). It will be useful to denote by $\sigma$ a permutation over $\cN$ so that $\sigma(t - 1)$ is the element removed at time step $t$, and $\sigma(n)$ is the element that remain in the active set until the very end.
\end{example}

Assuming we want to output at every time step a solution that provides $\beta$-approximation for some constant $\beta > 0$, then the algorithm must output at each time step a non-empty set with probability at least $\beta$. Thus, the hindsight-optimal solution is to use the solution $\{\sigma(n)\}$ at all time steps, which leads to a recourse of $1$ (this is non-zero due to the cost of getting to this solution from the empty solution). The adversary controls the permutation $\sigma$, and its goal is to make the recourse of any algorithm as high as possible compared to this $1$. When the adversary is non-adaptive, it can use a uniformly random permutation $\sigma$, and one can verify that this would lead to a logarithmic recourse for any algorithm guaranteeing $\beta$-competitiveness, which shows that at least one of the logarithmic factors in each one of the competitive recourse bounds of \Cref{tab:app} is necessary. 

The following shows that when the adversary can adaptively choose the element $\sigma(t)$ based on the solution $\Iset^{t-1}$ of the algorithm at time step $t - 1$, any $\Omega(1)$-competitive algorithm incurs much more recourse.
\begin{theorem}
\label{thm:adaptive-adversary-lb}
    Any dynamic $\beta$-competitive algorithm $\cA$ for maximization subject to a $1$-uniform matroid constraint against an adaptive adversary incurs worst-case competitive recourse $\Omega(\beta \cdot n)$.
\end{theorem}
\begin{proof}
    We demonstrate this via \Cref{ex:adaptive-lower-bound-1-uniform}.
    As we have established, the hindsight-optimal algorithm chooses $\Iset^1 = \ldots = \Iset^n = \{\sigma(n)\}$, which leads to a recourse of $1$.

		After each round $t - 1$, the adversary needs to set $\sigma(t - 1)$. When $\Iset^{t-1}$ is non-empty, an adaptive adversary can set $\sigma(t - 1)$ to be the single element of $\Iset^{t-1}$, which forces $\cA$ to choose $\Iset^t$ that is either empty or contains an element other than $\sigma(t - 1)$. In either case, the recourse of the algorithm increases by at least $1$ in such a round.
		
		Note now that since $\cA$ is a $\beta$-competitive algorithm, the fraction of the rounds in which $\Iset^t$ is non-empty must be at least $\beta$, and thus, by the above logic, the recourse of the algorithm must be at least $\beta n - 1 = \Omega(\beta \cdot n)$.
\end{proof}

The above instance, together with the relaxation $\cP = \{\vx \in [0, 1]^\cN : \|\vx\|_1 \leq 1\}$, demonstrates also the impossibility of non-trivial TCOS for adaptive adversaries.
\begin{theorem}
    If $\cR$ is an $\alpha$-TCOS against an adaptive adversary for $1$-uniform matroid and the polytope $\cP$ relaxing it, then $\alpha = \Omega(n)$.
\end{theorem}
\begin{proof}
    Given a non-empty set $S \subseteq \cN$, let us denote $\vx(S) := \frac{1}{\abs{S}} \cdot \sum_{e \in S} \mathbbm{1}_e$, where $\mathbbm{1}_e$ is the standard basis vector for $e$. For example, $x(\cN) = (\nf{1}{n}, \ldots, \nf{1}{n})$, and for every $e \in \cN$, $x(\{e\}) = \mathbbm{1}_e$. Note that $\vx(S) \in \cP$ whenever $S \neq \emptyset$.
    
    Assume now that the adversary maintains a permutation $\sigma$ as in \cref{ex:adaptive-lower-bound-1-uniform}. Specifically, in each round $t$, the adversary passes to the algorithm $\cR$ the vector $\vx^t \defeq \vx(\cN \setminus \sigma_{< t})$, where $\sigma_{< t}$ is the set of elements that appear in the permutation $\sigma$ prior to position $t$. The algorithm $\cR$ then has to return a set $R^t$ 
    that is an independent rounding of $\vx^t$.
    
    Based on $R^t$, the adversary chooses some element of $\cN \setminus \sigma_{< t}$ as $\sigma(t)$. If $R^t$ is empty, then the choice of $\sigma(t)$ by the adversary can be arbitrary, but when $|R^t| \geq 1$, the adversary chooses $\sigma(t)\in R^t$, which it can do since it is adaptive. Notice that when this happens for $1 \leq t < n - 1$, the set $R^{t + 1}$ cannot contain the chosen element of $R^t$, and thus, $R^t \triangle R^{t + 1}$ is nonempty. Hence, for every round $1 \leq t < n - 1$,
    \[
    	\bE[|R^t \triangle R^{t + 1}|]
    	\geq
    	\Pr[R^t \neq \emptyset]
    	=
    	1 - \bigg(1 - \frac{1}{n - t + 1}\bigg)^{n - t + 1}
    	\geq
    	1 - e^{-1}
    	\enspace.
    \]
    
    In contrast, for every round $1 \leq t < n - 1$,
    \[
    	\norm{\vx(\cN \setminus \sigma_{< t}) - \vx(\cN \setminus \sigma_{< t + 1})}_1
    	=
    	\frac{1}{n - t + 1} + (n - t) \cdot \bigg(\frac{1}{n - t} - \frac{1}{n - t + 1}\bigg)
    	=
    	\frac{2}{n - t + 1}
    	\enspace,
    \]
    and thus, $\expect{\abs{R^t \triangle R^{t + 1}}}$ is larger than $\norm{\vx(\cN \setminus \sigma_{< t}) - \vx(\cN \setminus \sigma_{< t + 1})}_1$ by a factor of $\Theta(n - t)$. Recall now that since $\cR$ is an $\alpha$-TCOS, $\alpha$ must be at least the above ratio for every $1 \leq t < n - 1$; thus $\alpha = \Omega(n)$.
\end{proof}

%% file: tcos-variance.tex
\section{Comparison of TCOS Recourse Variance}
\label{app:tcos}

For a given sequence $\{\vx^t\}_t$ with total $\ell_1$ movement $r$, both the threshold TCOS and the Markovian TCOS are $1$-TCOS, and therefore incur expected recourse $r$ over their respective sequences $\{\Rset^t\}_t$.
However, the Markovian TCOS produces a more predictable total recourse, with variance $O(r)$ regardless of the input sequence of fractional vectors.
By contrast, for all $T$ and $r = T\cdot (1-\Omega(1))$, there are sequences $\{\vx^t\}_t$ on which the threshold TCOS exhibits total recourse with variance $\Omega(rT)$. 

In this appendix we prove both these results. We begin with Example~\ref{ex:threshold-tcos-high-variance}, which demonstrates the high variance the threshold TCOS can have. Then, the majority of the appendix is devoted to bounding the variance of the Markovian TCOS.

\begin{example}
\label{ex:threshold-tcos-high-variance}
    For fixed $T$ and fractional recourse $r < T$, consider a ground set containing a single element $e$ and a sequence of vectors given by $x_e^t := \delta \cdot \ind{\text{$t$ is even}}$, where $\delta := \nf{r}{T}$.
    Clearly $\sum_t \norm{\vx^{t-1} - \vx^t}_1 = r$.
\end{example}
Given \Cref{ex:threshold-tcos-high-variance}, the threshold TCOS chooses a single $\lambda_e \sim U([0,1])$ at the outset. 
If $\lambda_e \in (0,\delta)$, then its total integral recourse on this sequence is $R = T$; otherwise, it is $R = 0$.
Thus, for $r = T(1-\Omega(1))$,
\[
    \var(R) = \expect*{R^2} - (\expect{R})^2 = T^2 \cdot \delta - (T \cdot \delta)^2 = T^2 \cdot \frac{r}{T} - T^2 \cdot \frac{r^2}{T^2} = rT - r^2 = \Omega(rT).
\]

We now turn to the Markovian TCOS.
We write the recourse of a sequence $\{\Rset^t\}$ as $R = \sum_t \sum_{e \in \cN} Y_e^t$, where $Y_e^t := \ind{e \in \Rset^{t-1} \symdif \Rset^t}$ is an indicator for the event that element $e$ is in the symmetric difference of $\Rset^{t-1}$ and $\Rset^t$.
Next, recall that $\delta_e^t := x_e^t - x_e^{t-1}$, and observe that $\expect{Y_e^t} = \abs{\delta_e^t}$
because the Markovian TCOS is a $1$-TCOS.

By the definition of the Markovian TCOS, the random decisions for distinct $e, e' \in \cN$ are independent, and thus, the indicators $Y_e^t$ and $Y_{e'}^{t'}$ are independent Bernoulli random variables.
If $Y_e^t$ and $Y_e^{t'}$ were also independent for fixed $e$ and $t \neq t'$, then we would be done: variances for independent random variables add, and hence, we would have $\var(R) = \sum_t \sum_{e \in \cN} \var(Y_e^t) = \sum_t \sum_e \delta_e^t (1 - \delta_e^t) \leq \sum_t \sum_e \delta_e^t = r$. 
But this does not hold: if $x_e^0 = 0$, $x_e^1 = \nf{1}{2}$, and $x_e^2 = 0$, then $Y_e^1$ and $Y_e^2$ are equal, and thus positively correlated.

To address this, we split each variable $Y^t_e$ into upward and downward parts. More formally, $Y^{t\uparrow}_e \defeq \ind{e \in \Rset^t \setminus \Rset^{t - 1}}$ and $Y^{t\downarrow}_e \defeq \ind{e \in \Rset^t \setminus \Rset^{t - 1}}$. Clearly, $Y^t_e = Y^{t\uparrow}_e + Y^{t\downarrow}_e$. Additionally, Corollary~\ref{cor:markovian-negative-time-correlation} below shows that these indicators exhibit negative correlations.
We now set about proving this. 

\begin{lemma}
\label{lem:markovian-positive-time-correlation}
    For all sequences $\{x_t\}_t$, time steps $1 \leq t \leq t' \leq T$ and elements $e$, the events $e \in \Rset^t$ and $e \in \Rset^{t'}$ are positively correlated in the sense that $\prob{e \in \Rset^{t} \wedge e \in \Rset^{t'}} \geq \prob{e \in \Rset^{t}} \cdot \prob{e \in \Rset^{t'}} = x^t_e x^{t'}_e$.
\end{lemma}
\begin{proof}
We prove the lemma by induction on $t' - t$. If $t' - t = 0$, then the claim holds since
\[
	\prob{e \in \Rset^{t} \wedge e \in \Rset^{t'}} = \prob{e \in \Rset^{t}} = x^t_e
	\geq
	(x^t_e)^2 = \prob{e \in \Rset^{t}} \cdot \prob{e \in \Rset^{t}}
	\enspace.
\]
Assume now that $t' - t > 0$. We need to consider a few cases. If $\delta^{t'}_e = 0$, then by the induction hypothesis,
\[
	\prob{e \in \Rset^t \wedge e \in \Rset^{t'}}
	=
	\prob{e \in \Rset^t \wedge e \in \Rset^{t' - 1}}
	\geq
	x^t_e x^{t' - 1}_e
	=
	x^t_e x^{t'}_e
	.
\]
Consider now the case that $\delta^{t'}_e > 0$. Then,
\begin{align*}
	\prob{e \in \Rset^t \wedge e \in \Rset^{t'}}
	={} &
	\prob{e \in \Rset^t \wedge e \in \Rset^{t' - 1}} + \prob{e \in \Rset^t \wedge e \not \in \Rset^{t' - 1}} \cdot \frac{\delta^{t'}_e}{1 - x^{t' - 1}_e}\\
	={} &
	\prob{e \in \Rset^t \wedge e \in \Rset^{t' - 1}} \cdot \bigg(1 - \frac{\delta^{t'}_e}{1 - x^{t' - 1}_e}\bigg) + \prob{e \in \Rset^t} \cdot \frac{\delta^{t'}_e}{1 - x^{t' - 1}_e}\\
	\geq{} &
	x^t_e x^{t' - 1}_e \cdot \bigg(1 - \frac{\delta^{t'}_e}{1 - x^{t' - 1}_e}\bigg) + x^t_e \cdot \frac{\delta^{t'}_e}{1 - x^{t' - 1}_e}
	=
	x^t_e (x^{t'-1}_e + \delta^{t'}_e)
	=
	x^t_e x^{t'}_e
	.
\end{align*}
Similarly, when $\delta^{t'}_e < 0$,
\begin{align*}
	\prob{e \in \Rset^t \wedge e \in \Rset^{t'}}
	={} &
	\prob{e \in \Rset^t \wedge e \in \Rset^{t' - 1}} \cdot \bigg(1 + \frac{\delta^{t'}_e}{x_e^{t' - 1}}\bigg)\\
	\geq{} &
	x^t_e x^{t' - 1}_e \cdot \bigg(1 + \frac{\delta^{t'}_e}{x_e^{t' - 1}}\bigg)
	=
	x^t_e (x^{t'-1}_e + \delta^{t'}_e)
	=
	x^t_e x^{t'}_e
	.
	\qedhere
\end{align*}
\end{proof}

\begin{corollary}
\label{cor:markovian-negative-time-correlation}
    For all sequences $\{x_t\}_t$, time steps $1 \leq t < t' \leq T$, and elements $e$, the events $Y^{t\uparrow}_e$ and $Y^{t'\uparrow}_e$ are negatively correlated. 
\end{corollary}
\begin{proof}
If either $\delta^t_e$ or $\delta^{t'}_e$ is non-positive, then both sides of the inequality $\prob{Y^{t\uparrow}_e \wedge Y^{t'\uparrow}_e} \leq \prob{Y^{t\uparrow}_e} \cdot \prob{Y^{t'\uparrow}_e}$ are equal to $0$. Otherwise,
\begin{align*}
	\prob{Y^{t\uparrow}_e \wedge Y^{t'\uparrow}_e}
	={} &
	\prob{Y^{t'\uparrow}_e \mid Y^{t\uparrow}_e} \cdot \prob{Y^{t\uparrow}_e}
	=
	\prob{Y^{t'\uparrow}_e \mid e \in \Rset^t} \cdot \prob{Y^{t\uparrow}_e}\\
	={} &
	\prob{Y^{t'\uparrow}_e \mid e \not \in \Rset^{t' - 1}} \cdot \prob{e \not \in \Rset^{t' - 1} \mid e \in \Rset^t} \cdot \prob{Y^{t\uparrow}_e}\\
	={} &
	\frac{\delta^{t'}_e}{1 - x^{t' - 1}_e} \cdot \frac{\prob{e \in \Rset^t} - \prob{e \in \Rset^{t' - 1} \wedge e \in \Rset^t}}{\prob{e \in \Rset^t}} \cdot \prob{Y^{t\uparrow}_e}\\
	\leq{} &
	\frac{\delta^{t'}_e}{1 - x^{t' - 1}_e} \cdot \Big(1 - \prob{e \in \Rset^{t' - 1}}\Big) \cdot \prob{Y^{t\uparrow}_e}
	=
	\delta^{t'}_e \cdot \prob{Y^{t\uparrow}_e}
	=
	\prob{Y^{t\uparrow}_e} \cdot \prob{Y^{t'}_e}
	,
\end{align*}
where the second and third equalities use the fact that the construction of $\Rset^{t + 1}$ by the Markovian TCOS depends on $\Rset^t$, but not on the content of previous sets, and the inequality holds by Lemma~\ref{lem:markovian-positive-time-correlation}.
\end{proof}

With this in hand, we are now ready to bound the variance of Markovian TCOS. 
\begin{theorem}
\label{thm:markovian-TCOS-low-variance}
    For all $T$ and all input sequences $\{\vx^t\}_t$ with $\sum_t \norm{x^{t-1} - x^t}_1 = r$, the recourse $R$ of the Markovian TCOS on $\{\vx^t\}_t$ satisfies $\var(R) = O(r)$.
\end{theorem}
\begin{proof}
Observe that $Y_e^{t\uparrow}$ and $Y_{e'}^{t'\uparrow}$ are independent whenever $e \neq e'$, and thus, $\Cov(Y_e^t, Y_{e'}^{t'}) = 0$. Additionally, by \Cref{lem:markovian-positive-time-correlation}, whenever $t \neq t'$,
\begin{align*}
	\Cov(Y_e^{t\uparrow}, Y_e^{t'\uparrow})
	={} &
	\bE[Y_e^{t\uparrow} \cdot Y_e^{t'\uparrow}] - \bE[Y_e^{t\uparrow}] \cdot \bE[Y_e^{t'\uparrow}]\\
	={} &
	\prob{Y_e^{t\uparrow} = 1 \wedge Y_e^{t'\uparrow} = 1} - \prob{Y_e^{t\uparrow} = 1} \cdot \prob{Y_e^{t'\uparrow} = 1}
	\leq
	0
	.
\end{align*}
Thus, letting $\bar \delta_e^t \defeq \min(0, \: \delta_e^t)$, 
\begin{align*}
	\var\bigg(\sum_{t \in [T]} \sum_{e \in \cN} Y_e^{t\uparrow}\bigg)
	={} &
	\sum_{t \in [T]} \sum_{t' \in [T]} \sum_{e \in \cN} \sum_{e' \in \cN} \Cov(Y_e^{t\uparrow}, Y_{e'}^{t'\uparrow})\\
	\leq{} &
	\sum_{t \in [T]} \sum_{e \in \cN} \var(Y_e^{t\uparrow})
	=
	\sum_{t \in [T]} \sum_{e \in \cN}	\bar \delta_e^t(1 - \bar \delta_e^t) \leq \sum_t \sum_e \bar \delta_e^t = r.
\end{align*}

Next, we observe that, since $\Rset^0 = \emptyset$,
\begin{align*}
	\sum_{t \in [T]} \sum_{e \in \cN} (Y_e^{t\uparrow} - Y_e^{t\downarrow})
	={} &
	\sum_{t \in [T]} \sum_{e \in \cN} (\ind{e \in \Rset^t \setminus \Rset^{t - 1}} - \ind{e \in \Rset^{t - 1} \setminus \Rset^{t}})\\
	={} &
	\sum_{t \in [T]}  (|\Rset^t \setminus \Rset^{t - 1}| - |\Rset^{t - 1} \setminus \Rset^{t}|)
	=
	\sum_{t \in [T]}  (|\Rset^t| - |\Rset^{t - 1}|)
	=
	|\Rset^T|
	.
\end{align*}
Combining the above results yields
\begin{align*}
	\var(R)
	={} &
	\var\bigg(2 \cdot \sum_{t \in [T]} \sum_{e \in \cN} Y_e^{t\uparrow} - |\Rset^T|\bigg)\\
	={} &
	4 \cdot \var\bigg(\sum_{t \in [T]} \sum_{e \in \cN} Y_e^{t\uparrow}\bigg) + \var(|\Rset^T|) - 4 \cdot \Cov\bigg(\sum_{t \in [T]} Y_e^{t\uparrow}, |\Rset^T|\bigg)\\
	\leq{} &
	4r + \sum_{e \in \cN} x^T_e(1 - x^T_e) + 4r
	\leq
	9r
	,
\end{align*}
where the penultimate inequality holds since $Y_e^{t\uparrow} \geq 0$ and $\ind{e \in \Rset^T} \geq 0$, and so
\begin{align*}
	\Cov\bigg(\sum_{t \in [T]} \sum_{e \in \cN} Y_e^{t\uparrow}, |\Rset^T|\bigg)
	&=
	\Cov\bigg(\sum_{t \in [T]} \sum_{e \in \cN} Y_e^{t\uparrow}, \sum_{e' \in \cN} \ind{e' \in \Rset^T}\bigg)
    \\
	&=
	\sum_{t \in [T]} \sum_{e \in \cN} \sum_{e' \in \cN} \Cov\bigg(Y_e^{t\uparrow}, \ind{e' \in \Rset^T}\bigg)
	=
	\sum_{t \in [T]} \sum_{e \in \cN} \Cov\bigg(Y_e^{t\uparrow}, \ind{e \in \Rset^T}\bigg) \\
    &\geq 
	- \sum_{t \in [T]} \sum_{e \in \cN} \bE[Y_e^{t\uparrow}] \cdot \bE[\ind{e \in \Rset^T}]
	\\
    &\geq
	-\sum_{t \in [T]} \sum_{e \in \cN} \bE[Y_e^{t\uparrow}]
	=
	-\sum_{t \in [T]} \sum_{e \in \cN} \max\{0, \delta^t_e\}
	\geq
	-r
	. \qedhere
\end{align*}
\end{proof}

%% file: polytime.tex
\section{Matroid DCRSs in Polynomial Time}
\label{sec:polytime}

Initial applications of (offline) contention resolution schemes provided strong motivation for development of polynomial-time CRSs \cite{DBLP:journals/siamcomp/ChekuriVZ14}, and efficient (or, at least, polynomial-time) implementations have been of continued interest.
\cite[Theorem 4.3]{DBLP:journals/siamcomp/ChekuriVZ14} obtain $\eps$-instance-optimal CRSs for matroids via a somewhat involved approach of sampling constraints to equip the ellipsoid algorithm with a weak separation oracle, but remark that simpler (greedy) matroid CRSs can be made to run in polynomial time via sampling to estimate span probabilities and appealing to concentration \cite[Theorem 4.14]{DBLP:journals/siamcomp/ChekuriVZ14}.

Our matroid DCRS is not quite so simple, but our argument will more closely resemble this second approach in that we will estimate probabilities that elements are spanned via sampling and concentration, argue that the event of failing to produce an accurate estimate is rare, and adapt the construction of the DCRS to still provide balance and bounded expected recourse---even subject to this occasional failure.

Our modified DCRS is \Cref{alg:matroid-DCRS-quantized}, and we argue that it has the desired properties provided that the sequence of $\vx^t$ it faces is quantized; that is, each $x_e^t$ takes on values that are within a carefully chosen subset of $[0,1]$. 
We complement this poly-time matroid DCRS that works for quantized $\vx^t$ with a generic reduction to the quantized setting (\Cref{alg:matroid-DCRS-polytime}).
Given an arbitrary sequence of $\vx^t$ and $\Rset^t$ from an unspecified TCOS, \Cref{alg:matroid-DCRS-polytime} maintains a dynamic rounding, implements a corresponding quantized TCOS that downsamples from $\Rset^t$, simulates the black-box DCRS on this quantized arrival sequence, and modifies its output to create a sequence $\{\Iset^t\}$, all while incurring an arbitrarily small overhead in the balance and an $O(1)$ multiplicative overhead in the expected recourse.

\input{alg_matroid_DCRS_polytime.tex}

\Cref{alg:matroid-DCRS-polytime} accomplishes this by lazily maintaining a rounding $\bar \vx^t$ of $\vx^t$ such that (i) $\bar\vx^t \leq \vx^t$ for all $t$ coordinate-wise, and (ii) $\sum_{t} \norm{\bar\vx^t - \bar\vx^{t-1}} \leq \sum_{t} \norm{\vx^t - \vx^{t-1}}$. 
The values $x^t_e$ are rounded to values $\bar x^t_e \in \cB_\delta$, which we define for $\delta \in \{1/i:\: i\in\mathbb{N}\}$ to be 
\begin{equation}
	\label{eq:def-buckets-new}
	\cB_\delta \defeq \{0,\: \delta,\: 2\delta,\: \ldots,\: 1-\delta,\: 1\}
	 .
\end{equation}
It will be useful to let $j(x) \defeq \lfloor x/\delta \rfloor$ for $x \in [0,1]$, so that $j(x)\cdot \delta = \max\{b \in \cB_\delta : b \leq x\}$.

\begin{lemma}
\label{lem:polytime-reduction-to-quantized-dcrs}
    For any $b \in (0,1)$ and $\eps > 0$, if \Cref{alg:matroid-DCRS-quantized} is a $(b,\: 1-b-\eps,\: z)$-DCRS, then for appropriate choice of values for the parameters $\gamma$ and $\delta$ of \Cref{alg:matroid-DCRS-polytime}, this algorithm is a $(b,\: 1-b-2\eps,\: O(z))$-DCRS.
\end{lemma}
Furthermore, it is apparent that this reduction runs in polynomial time.
\begin{observation}
    For the choices of $\gamma$ and $\delta$ used in the proof of \Cref{lem:polytime-reduction-to-quantized-dcrs}, \Cref{alg:matroid-DCRS-polytime} uses time $\poly(n, \eps^{-1})$ per step $t$ on top of the time required for\Cref{alg:matroid-DCRS-quantized} .
\end{observation}

We also prove the following result.

\begin{lemma} \label{lem:quantized_guarantee}
    Provided that all arriving $\vx^t \in \cB_\delta$, \Cref{alg:matroid-DCRS-quantized} is a $(b,\: 1-b-\eps,\: O(\eps^{-3} \cdot \log \rank(\cN)))$-DCRS.
    Moreover, it performs $\poly(n, \eps^{-1})$ operations per step $t$ in expectation.
\end{lemma}

Taken together, for any $b \in (0,1)$ and $\eps > 0$ this establishes a $(b,\: 1-b-\eps,\: O(\eps^{-3} \cdot \log \rank(\cN)))$-DCRS that runs in time $\poly(n, \eps^{-1})$.

\subsection{Reduction to the Quantized Case (\texorpdfstring{\Cref{alg:matroid-DCRS-polytime}}{Algorithm~\ref*{alg:matroid-DCRS-polytime}})}

In this section we prove \Cref{lem:polytime-reduction-to-quantized-dcrs}. We note that the procedure $\textsc{Greedy}_\cM(S, \sigma)$ used by \Cref{alg:matroid-DCRS-polytime} is the standard greedy algorithm, i.e. an algorithm that constructs an independent set $I$ of $\cM$ by starting with the empty set, scanning the elements of $S$ in the order specified by $\sigma$, and, for every scanned element, adding it to $I$ if doing so preserves the independence of $I$.

As a claim about the continuity of CRSs with respect to the fractional input vector $\vx$, \Cref{lem:polytime-reduction-to-quantized-dcrs} is quite intuitive: it says that for $\vx$ and $\bar\vx$ that are close in an appropriate sense, any $(b,c)$-CRS $\pi_{\bar{\vx}}$ can be made to serve as a (slightly worse) CRS for $\vx$.
This is not as trivial as one might hope because it is necessary to preserve (multiplicative) selectability for elements with very small $x_e^t$. 
For DCRSs there is the additional challenge of showing that this black-box reduction respects the expected recourse for this particular rounding of the $\vx$ vector.

\begin{proof}[Proof of \Cref{lem:polytime-reduction-to-quantized-dcrs}]
    We prove the lemma for polynomially small choice of the parameters $\gamma$ and $\delta$, which we will specify as we go. 
    In order to show the claim, we must show that \Cref{alg:matroid-DCRS-polytime} meets the definition of a DCRS (\Cref{def:dcrs}). 
    We establish these properties in order.

    First, that the $\Iset^t$ returned on \Cref{line:pt-dcrs-return} is independent in $\cM$ is immediate.
    
    Second, we address balance. 
    We case on whether $x_e^t \leq \gamma$.
    If so, then $e \in \cN_L$ and 
    \begin{align}
    	\prob{e \in \Iset^t \mid e \in \Rset^t} 
    	&\geq \prob{e \in \Iset_L^t \mid e \in \Rset^t} \notag \\
    	&\geq 1 - \sum_{e' \neq e} \prob{e' \in \Iset_L^t} \notag \\
        &\geq 1 - \sum_{e' \neq e} \prob{e \in R^t} \cdot \prob{x_e^t \leq \lambda_e^L} \notag\\
    	&\geq 1 - n \cdot 2 \gamma
    	 ,
    \end{align}
    by \cref{line:pt-dcrs-greedysmall}. 
    For any $\gamma \leq \eps \cdot n^{-1}/2$ this guarantees that the balance of \Cref{alg:matroid-DCRS-polytime} with respect to every element $e$ with $x_e^t \leq \gamma$ is at least $1 - \eps \geq 1 - b - 2\eps$.
    
    Otherwise $x_e^t > \gamma$. Our rounding guarantees $x_e^t < \bar x_e^t + 2 \delta$, and thus,
	provided that $\delta \leq \gamma/4$, 
     we get in this case that $\bar x_e^t \geq \gamma/2$.
    Since $\bar \vx^t \leq \vx^t$ coordinate-wise for all $t$, clearly if $\vx^t \in b\cdot P_\cM$ then $\bar\vx^t \in b\cdot P_\cM$ also.
    Further, we claim that the $\{\bar \Rset^t\}$ maintained by \Cref{alg:matroid-DCRS-polytime} is both a TCOS for $\{\bar\vx^t\}$ (in the sense that $\prob{e \in \bar R^t} = \bar x_e^t$ for all $e$ and $t$) and satisfies $\bar \Rset^t \subseteq \Rset^t$ for all $t$.
    
    Therefore if the algorithm being simulated (\Cref{alg:matroid-DCRS-quantized}) is a $(b, 1-b-\eps, z)$-DCRS (by assumption), then \Cref{alg:matroid-DCRS-polytime} returns a $\bar \Iset^t \subseteq \bar \Rset^t \subseteq \Rset^t$ such that
    \[
        \prob{e \in I_H^t \mid e \in \bar \Rset^t} \geq 1 - b - \eps 
         .
    \]
    
	 This lets us conclude that for elements $e$ with $\bar x_e^t \geq \gamma/2$, 
    \begin{align}
    	\prob{e \in \Iset^t \mid e \in \Rset^t} 
    	&\geq \prob{e \in \Iset^t_H \mid e \in \bar\Rset^t} \cdot \prob{e \in \bar \Rset^t \mid e \in \Rset^t} - \prob{\Iset_L^t \not \subseteq \{e\} \mid e \in \Rset^t} \notag \\
    	&\geq \prob{e \in \Iset^t_H \mid e \in \bar\Rset^t} \cdot \bigg(1 - \frac{2\delta}{x_e^t}\bigg) - \prob{R^t \cap \cN_L \not \subseteq \{e\}} \notag \\
    	&\geq (1 - b - \eps) \cdot \bigg(1-\frac{4 \delta}{\gamma}\bigg)
    	- n \cdot 2 \gamma  \label{eq:bigx-selectable-bound}
    	 ,
    \end{align}
    where the second inequality holds because the set $R^t \cap \cN_L \setminus \{e\}$ contains $\Iset_L^t \setminus \{e\}$ and is independent of $[e \in \Rset^t]$ and because the facts that $\{\bar\Rset^t\}$ is a TCOS for $\{\bar \vx^t\}$ that downsamples from $\{\Rset^t\}$ and $\bar x_e^t \leq x_e^t$ together imply that $\prob{e \in \bar \Rset^t \mid e \in \Rset^t} = \bar x_e^t/x_e^t \geq 1 - \frac{2\delta}{x_e^t}$.
    Inequality~\eqref{eq:bigx-selectable-bound} can be made to imply $\prob{e \in \Iset^t \mid e \in \Rset^t} \geq 1 - b - 2\eps$ via choices of $\gamma$ and $\delta$ that are $1/\poly(n, 1/\eps)$.

	\medskip
    We now address recourse. 
    We will use the fact that, by \Cref{line:pt-dcrs-finalgreedyloop,line:pt-dcrs-independent} of \Cref{alg:matroid-DCRS-polytime},
    \begin{align}
    	\sum_t \abs{\Iset^t \symdif \Iset^{t-1}} 
    	&\leq 2 \sum_t \abs{\Iset_L^t \symdif \Iset_L^{t-1}} + \sum_t \abs{\Iset_H^t \symdif \Iset_H^{t-1}} 
    	  .
        \label{eq:pt-I-movement-bound}
    \end{align}
    (This is because, for fixed $\Iset_H^t$, a single change to $\Iset_L^t$ affects at most one membership change among the subsequent $e \in \Iset_H^t$ in $\Iset^t$.)
    We will bound these two terms separately.
    
    For the first term, changes can happen either in $\Rset$ or in $\cN_L \cap \Rset$, for fixed $\Rset$.
    Therefore
		\[
		    	\sum_t \abs{\Iset_L^t \symdif \Iset_L^{t-1}} 
    	\leq \sum_t \abs{\Rset^t \symdif \Rset^{t-1}}
    	+ \sum_t \sum_{e \in \Rset^t \cap \Rset^{t-1}} \mspace{-18mu} \ind{e \in \cN_L^t \symdif \cN_L^{t-1}} 
 	      .
    \]
		Taking expectation, the last inequality yields 
		\begin{align}
		\sum_t \expect{\abs{\Iset_L^t \symdif \Iset_L^{t-1}} } 
		& \leq \sum_t \expect{\abs{\Rset^t \symdif \Rset^{t-1}}} 
		\notag\\&\mspace{100mu}+ \sum_t \sum_{e \mid \min\{x^t_e, x^{t - 1}_e\} \leq 2\gamma} \mspace{-36mu}  \prob{e \in \cN_L^t \symdif \cN_L^{t-1} \mid e \in \Rset^t \cap \Rset^{t-1}} \prob{e \in \Rset^t \cap \Rset^{t-1}} \notag \\
		& \leq \sum_t \expect{\abs{\Rset^t \symdif \Rset^{t-1}}} 
		+ \sum_t \sum_e \frac{\abs{x_e^t - x_e^{t-1}}}{\gamma} \cdot 2\gamma \notag \\
		& \leq 3 \sum_t \expect{\abs{\Rset^t \symdif \Rset^{t-1}}} 
    	  ,
        \label{eq:pt-IL-movement-bound}
    \end{align}
    where the last inequality holds since $\{R^t\}$ are generated by a TCOS for $\{\vx^t\}$ (\Cref{obs:TCOS-1-impossibility}).

    We now bound the second term. Since our inner DCRS (\Cref{alg:matroid-DCRS-quantized}) is a DCRS,
		\begin{align}
        \sum_{t \in [T]} \expect{\abs{\Iset_H^t \symdif \Iset_H^{t-1}}} 
        \leq z \cdot \sum_{t \in [T]} \expect{\abs{\bar\Rset^t \symdif \bar \Rset^{t-1}}}.
        \label{eq:pt-dcrs-constant-run-IH-recourse}
    \end{align}
    The challenge is to relate the $\bar \Rset$ movement to $\Rset$ movement. Let us denote by $J_e \subseteq [T]$ the time steps for which $\bar x_e^{t} \neq \bar x_e^{t-1}$, and let $j_1, j_2, \dotsc$ denote the times in $J_e$ in increasing order. Then,
    \begin{align}
    	\sum_t \expect{\abs{\bar R^t \symdif \bar R^{t-1}}}
    	& = \sum_e \left( \sum_{t \in J_e} \prob{e \in \bar R^t \symdif \bar R^{t-1}} + \sum_{t \in [T] \setminus J_e} \mspace{-9mu} \prob{e \in \bar R^t \symdif \bar R^{t-1}} \right) \label{eq:pt-IH-movement-bound}
         .
    \end{align}
    We tackle the first term first. 
    For fixed $e$ and $t_j \in J_e$, for \emph{increasing} distinct values $\bar x_e^{t_j}$, $\bar x_e^{t_{j+1}}$, 
    we have
    \begin{align}
        \prob{e \in \bar \Rset^{t_j} \setminus \bar \Rset^{t_{j-1}}} 
        &\leq \prob{e \in \Rset^{t_{j - 1}} \setminus \bar \Rset^{t_{j - 1}}} + \prob{e \in \Rset^{t_j} \setminus \Rset^{t_{j - 1}}} \notag \\
        & \leq 2\delta + \prob{e \in \Rset^{t_j} \setminus \Rset^{t_{j-1}}} \label{eq:pt-rounded-up-movement1}
         . 
    \end{align}
    Next note that our lazy quantization (\Cref{line:pt-dcrs-roundx1,line:pt-dcrs-roundx2}) ensures that if $\bar x_e^{t_j} > \bar x_e^{t_{j-1}}$ then $x_e^{t_j} > x_e^{t_{j-1}} + \delta$, and moreover $x_e^{t_j} - x_e^{t_{j-1}} + \delta > \bar x_e^{t_j} - \bar x_e^{t_{j-1}} \geq 2\delta$.
    Combining this with \eqref{eq:pt-rounded-up-movement1},  
    and using the TCOS guarantee that $\Rset \sim \Rdist(\vx)$ at every step, we obtain
    \begin{align}
        \prob{e \in \bar \Rset^{t_j} \setminus \bar \Rset^{t_{j-1}}} 
        &\leq 2(x_e^{t_j} - x_e^{t_{j-1}}) + \prob{e \in \Rset^{t_j} \setminus \Rset^{t_{j-1}}} \notag \\
        &\leq 3 \prob{e \in \Rset^{t_j} \setminus \Rset^{t_{j-1}}}
        \label{eq:pt-rounded-up-movement3}
         . 
    \end{align}
		For fixed $e$ and $t_j \in J_e$, for \emph{decreasing} distinct values $\bar x_e^{t_j}$, $\bar x_e^{t_{j+1}}$, 
    we have
    \begin{align*}
        \prob{e \in \bar \Rset^{t_j} \setminus \bar \Rset^{t_{j-1}}} 
        &\leq \prob{e \in \Rset^{t_{j - 1}} \setminus \bar \Rset^{t_{j - 1}}} + \prob{e \in \Rset^{t_j} \setminus \Rset^{t_{j - 1}}} \notag \\
        & = (x^{t_{j-1}}_e - \bar x^{t_{j-1}}_e) + \prob{e \in \Rset^{t_j} \setminus \Rset^{t_{j-1}}}
         . 
    \end{align*}
		If $\bar x_e^{t_j} < \bar x_e^{t_{j-1}}$, then $x_e^{t_j}$ must also be smaller than $\bar x_e^{t_{j-1}}$, and thus, 
		    \begin{align}
        \prob{e \in \bar \Rset^{t_j} \setminus \bar \Rset^{t_{j-1}}} 
        &\leq (x_e^{t_{j - 1}} - x_e^{t_{j}}) + \prob{e \in \Rset^{t_j} \setminus \Rset^{t_{j-1}}} \notag \\
        &\leq 2 \prob{e \in \Rset^{t_j} \symdif \Rset^{t_{j-1}}}
        \label{eq:pt-rounded-up-movement4}
         . 
    \end{align}
		
    Summing over $e \in \cN$, $t_j \in J_e$, we obtain
    \begin{align}
				\sum_e \sum_{t_j \in J_e}\prob{e \in \bar \Rset^{t_j} \setminus \bar \Rset^{t_{j-1}}}
				\leq 3 \sum_e \sum_{t_j \in J_e} \prob{e \in \Rset^{t_j} \symdif \Rset^{t_{j-1}}}
        \leq 3 \sum_t \expect{\abs{\Rset^t \symdif \Rset^{t-1}}} 
        \label{eq:pt-rounded-up-movement5}
         . 
    \end{align}
    Finally, since total increase upper bounds total decrease,
    \begin{equation}
        \expect*{\sum_e \sum_{t_j \in J_e} \ind{e \in \bar R^t \symdif \bar R^{t-1}}}
        \leq 2 \cdot \sum_e \sum_{t_j \in J_e}\prob{e \in \bar \Rset^{t_j} \setminus \bar \Rset^{t_{j-1}}}
        \leq 6 \cdot \sum_t \expect{\abs{\Rset^t \symdif \Rset^{t-1}}} 
        \label{eq:pt-rounded-up-movement6}
         . 
    \end{equation}

    It remains only to address the second half of \eqref{eq:pt-IH-movement-bound}, and bound the movement in $\bar R^t$ when $\bar x_e^t$ remains fixed. Consider an element $e$ and time $t$ such that $\bar x^t_e = \bar x^{t -1}_e$, and let $\bar x_e$ be the common value of these two coordinates. Then,
    \begin{align*}
        \prob{e \in \bar\Rset^t \symdif \bar\Rset^{t-1}} 
        \mspace{-20mu}&\mspace{20mu}\leq \prob{e \in \Rset^t \symdif \Rset^{t-1}} + \prob{e \in \bar\Rset^t \symdif \bar\Rset^{t-1} \mid e \in \Rset^t\cap \Rset^{t-1}} \prob{e \in \Rset^t\cap \Rset^{t-1}} \\
        \leq{}& \prob{e \in \Rset^t \symdif \Rset^{t-1}} + \bigg[\prob*{\frac{\bar x^t_e}{x_e^t} \leq \lambda_e < \frac{\bar x^{t - 1}_e}{x_e^{t-1}}} + \prob*{\frac{\bar x^t_e}{x_e^t} > \lambda_e \geq \frac{\bar x^{t-1}_e}{x_e^{t-1}}}\bigg]  \cdot \min(x_e^t, x_e^{t-1}) \\
				={} &
				\prob{e \in \Rset^t \symdif \Rset^{t-1}} + \left|\frac{\bar x_e}{x_e^t} - \frac{\bar x_e}{x_e^{t-1}}\right|  \cdot \min(x_e^t, x_e^{t-1}) 
         ,
    \end{align*}
    by \Cref{line:pt-dcrs-trim-R}. Since $\Rset^t \sim \Rdist(\vx^t)$ for each $t$, we also have
    \[
        \left|\frac{\bar x_e}{x_e^t} - \frac{\bar x_e}{x_e^{t-1}}\right| \cdot \min(x_e^t, x_e^{t-1}) 
        = \frac{\abs{x_e^t - x_e^{t-1}}}{x_e^t x_e^{t-1}} \cdot \bar x_e \cdot \min(x_e^t, x_e^{t-1}) \leq \abs{x_e^t - x_e^{t-1}}
				\leq
				\prob{e \in R^t \symdif R^{t - 1}}
         .
    \]
    and so summing over all $e \in \cN$ and $t \in [T] \setminus J_e$ yields
    \begin{equation}
        \sum_{e \in \cN} \sum_{t \in [T] \setminus J_e} \mspace{-9mu} \prob{e \in \bar\Rset^t \symdif \bar\Rset^{t-1}} 
				\leq
				\sum_{e \in \cN} \sum_{t \in [T] \setminus J_e} 2 \cdot \prob{e \in \Rset^t \symdif \Rset^{t-1}}
				\leq
				2 \cdot \sum_t \expect{\abs{\Rset^t \symdif \Rset^{t-1}}} 
          .
        \label{eq:pt-dcrs-constant-run-Rbar-recourse}
    \end{equation}
    Combining \eqref{eq:pt-dcrs-constant-run-Rbar-recourse} with \eqref{eq:pt-dcrs-constant-run-IH-recourse} and \eqref{eq:pt-rounded-up-movement6} and \eqref{eq:pt-IH-movement-bound}, and plugging it and \eqref{eq:pt-IL-movement-bound} into \eqref{eq:pt-I-movement-bound}, we obtain
    \[
        \sum_t \abs{\Iset^t \symdif \Iset^{t-1}} \leq (6 + 8z) \cdot \sum_{t} \expect{\abs{\Rset^t \symdif \Rset^{t-1}}} = O(z) \cdot \sum_{t} \expect{\abs{\Rset^t \symdif \Rset^{t-1}}}
         .
        \qedhere
    \]
\end{proof}

\subsection{Quantized Polynomial-Time DCRSs (\texorpdfstring{\Cref{alg:matroid-DCRS-quantized}}{Algorithm~\ref*{alg:matroid-DCRS-quantized}})}
Our approach here is to make the probability of failing to obtain a sufficiently accurate Monte Carlo estimate of $\prob{e \in \Span(\Rdist(\vx(i-1)) \cup S_{i})}$ on \Cref{line:empirical_add_prob_pt} small. 
Ideally, we would be able to show that if \Cref{alg:matroid-DCRS-quantized} enters this failure mode, then it exits the mode (e.g. by performing the appropriate reset) after boundedly many steps. 
In general, it is not clear why this should be the case. 
But if the arriving $\vx^t$ are quantized according to \eqref{eq:def-buckets-new}, then each step corresponds to (inverse) polynomial $\vx$ movement. 
Hence after polynomially many steps, we can always afford to do a full reset, (likely) exiting the failure mode.

Since \Cref{alg:matroid-DCRS-quantized} skips iterations in which $\vx^t = \vx^{t - 1}$, we can assume for notational simplicity that no such iterations exist.
Let $\cE^t$ be the event that all estimates $\bar p_e$ computed in the processing of step $t$ of \Cref{alg:matroid-DCRS-quantized} obey $\abs{p_e - \bar p_e} < \eps'$.
Conditioned on $\cE^t$, we begin by furnishing a slightly more careful version of the termination argument (\Cref{lem:terminate}) for \Cref{alg:matroid-DCRS} at a step $t$.
\begin{lemma}[Step runtime] 
\label{lem:pt-dcrs-quant-bdd-good-ops}
    For the parameter choices made in the proof of \Cref{thm:mat_dcrs_formal}, if $\cE^t$ holds then the execution of a single step $t$ of \Cref{alg:matroid-DCRS-polytime} takes $O(n^3)$ operations and $O(n^3)$ estimates of $\bar p_e$.
    
    Moreover, if the failure probability of an individual $\bar p_e$ estimate satisfies $\delta \leq C \cdot n^{-2}$ for a sufficiently small constant $C$, then the \emph{expected} runtime of step $t$ of \Cref{alg:matroid-DCRS-quantized} is $O(n^3)$ (even if $\cE^t$ is not guaranteed to hold).
\end{lemma}
\begin{proof}
    The loop at \Cref{line:T_add_condition_pt} iterates at most $n$ times. 
    This is also true of the loop at \Cref{line:add_condition_pt}: provided each \Reset($i$) is called at most once, it can add at most $2n$ elements in total to $S_i$ and therefore completes at most $2n$ iterates.
    (Clearly the hard reset on \Cref{line:hard_reset_pt} is called at most once, unconditionally.)
    Assuming $\cE^t$ holds, the same argument as for \Cref{lem:terminate} suffices to show no level is reset twice, so long as $\frac{1}{\tau - \eps'} < \frac{r}{\tau}$ and the process of constructing $S_i$ meets the conditions of \Cref{lem:growth_bound} with span probability bound $\tau - \eps'$.
    This holds for our choices of $r$ and $\tau$ (\Cref{line:r_tau_settings_pt}) so long as $\eps' \leq \eps^2/2$.

    All together, at most $O(n^2)$ lines are executed per step $t$, and \Cref{alg:matroid-DCRS-quantized} computes at most $n \cdot \ell n \leq n^3$ individual $\bar p_e$ estimates between \Cref{line:empirical_add_prob_pt} and \Cref{line:empirical_add_prob_recompute_pt} per step.

    To argue about the expected runtime even if $\cE^t$ is not guaranteed to hold, note that the only way in which the above argument can fail is if \Reset() is called more than once per $i$. 
    Still, if $\delta < 1/4n^2$, then after each reset, iteration $i$ has probability of $\Omega(1)$ of making it up to satisfying \Cref{lem:growth_bound} without a failed estimate, and therefore terminating and moving past $i$. Thus, the expected number of resets done for a single $i$ value is only a constant, which implies that the expected number of operations before success of iteration $i$ is $O(n^2)$.
    Summing these over $i \in [\ell]$ establishes the claim.
\end{proof}

Our hard resets (\Cref{line:hard_reset_pt}) allow us to define \emph{epochs}: each epoch $u$ is a sequence of time steps $T_u \subseteq [T]$ ending with a hard reset. 
An \emph{epoch failure} is the event $\Lambda_u$ that a failure $\neg \cE^t$ occurs at some step $t \in T_u$.

\begin{proof}[Proof of \Cref{lem:quantized_guarantee}]
    \Cref{lem:pt-dcrs-quant-bdd-good-ops} shows that \Cref{alg:matroid-DCRS-quantized} runs in (expected) polynomial time when $M(n,\eps')$ is polynomial.
    It remains to show that \Cref{alg:matroid-DCRS-quantized} is a DCRS for appropriate choices of $\delta$, $M(n,\eps')$, and $\eps'$.
    
    First, the assumption that $\vx \in \cB_\delta$ lets us bound the number of steps per epoch:
    \[
        \abs{T_u} \leq \frac{\rank(\cN)}{\delta}.
    \]
    (Recall from the proof of \Cref{lem:polytime-reduction-to-quantized-dcrs} that $\delta = 1/\poly(n)$.)
    This in turn enables a bound on epoch failure: 
    \begin{equation}
        \prob{\Lambda_u} \leq \prob{\bar p_e \text{ failure}} \cdot n^3 \cdot \abs{T_u} \leq \poly(n, \eps^{-1}) \cdot \exp({-M(n, \eps')\cdot\eps'^2})
         ,
    \label{eq:epoch-fail-bound}
    \end{equation}
    where we use Hoeffding's inequality to relate the probability that $\abs{\bar p_e - p_e} \geq \eps'$ to the (as-of-yet unspecified) number of samples $M(n, \eps')$ used to estimate $\bar p_e$.

    If $u$ is an epoch without failure, the arguments from \Cref{sec:matroid} show that \Cref{alg:matroid-DCRS-quantized} exhibits balance $1-\tau - \eps'$, and asymptotically equivalent expected recourse (even with the hard resets).
    The failure upper bound above holds regardless of $\vx^t$ and $\Rset^t$, and so if $u(t)$ is the epoch containing $t$, then the balance of any $e$ in round $t$ is at least
    \[
        (1-\tau - \eps')(1-\prob{\Lambda_u})
         ,
    \]
    and the expected recourse (as a function of the $z$ derived in \Cref{thm:mat_dcrs_formal}) satisfies
    \begin{align*}
        \sum_t \expect{\abs{\Rset^t \symdif \Rset^{t-1}}} \leq 2 z \cdot (1-\prob{\Lambda_u}) \cdot \expect{\abs{\vx^t - \vx^{t-1}}} + \prob{\Lambda_u} \cdot n \cdot \abs{T_u}
         ,
    \end{align*}
    because the recourse is at most $n$ per step. 
    Choosing $M(n, \eps')$ to be a sufficiently large polynomial suffices to make these balance and recourse guarantees the same as the ones guaranteed by \Cref{thm:mat_dcrs_formal}.
\end{proof}

%% file: alg_matroid_DCRS_polytime.tex
\begin{algorithm}[th]
	\caption{\textsc{MatroidDCRSPolytime} }
	\label{alg:matroid-DCRS-polytime}
	\begin{algorithmic}[1]
		\Require matroid $\cM$, $b \in (0,1)$, $\eps \in (0, 1-b)$, online sequences $\{\vx^t\} \subseteq b \cdot P_{\cM}$ and TCOS output $\{\Rset^t\}_t$, discretization $\cB_\delta$, $\gamma$.
		\State{Sample $\lambda_e \sim \mathrm{Unif}([0,1])$ for each $e\in \cN$.}
		\State{Sample $\lambda_e^L \sim \mathrm{Unif}([\gamma, 2\gamma])$ for each $e\in \cN$.}
		\State{$\sigma \gets$ fixed ordering of $\cN$}
		\State{$\bar \vx^0 \gets 0$ and $\bar R^0 \gets \emptyset$}
		\For{$t = 1, 2, \ldots$}
		      \State{$\bar \Rset^t \gets \Rset^t$}
    		\For{$e \in \cN$}:
        		\State{$i \gets \bar x_e^{t-1}/\delta$ and $i' \gets j(x_e^t)$}
        		\If{$i' - i = 1$ \label{line:pt-dcrs-roundx1-condition} } 
        		      \State{$\bar x_e^t \gets \bar x_e^{t-1}$ \label{line:pt-dcrs-roundx1}}
        		\Else{}
        		      \State{$\bar x_e^t \gets i' \cdot \delta$\label{line:pt-dcrs-roundx2}}
        		\EndIf
        		\If{$e \in \Rset^t$ and $\lambda_e  x_e^t > \bar x_e^t $ \label{line:pt-dcrs-trim-R}}
        		      \State{$\bar\Rset^t \gets \bar\Rset^t \setminus \{e\}$}
        		\EndIf
    		\EndFor
    		\State{$\cN_L \gets \{e \in \cN: x_e^t \leq \lambda_e^L\}$ \label{line:pt-dcrs-smalleles} }
    		\State{$I_L^t \gets \textsc{Greedy}_\cM(\Rset^t \cap \cN_L, \sigma)$ \label{line:pt-dcrs-greedysmall} }
    		\State{$I_H^t \gets \textsc{MatroidDCRSQuantized}(\cM, \eps, b, \bar\vx^{\leq t}, \bar R^{\leq t})$\hspace{5mm}} \algorithmiccomment{\textit{\Cref{alg:matroid-DCRS-quantized}}}
    		\State{$I^t \gets I_L^t$}
    		\For{$e \in I_H^t$ in $\sigma$ order \label{line:pt-dcrs-finalgreedyloop}}
    		\State{If $I^t\cup \{e\} \in \cI$, then $I^t \gets I^t\cup \{e\}$ \label{line:pt-dcrs-independent} }
		\EndFor
		\State{\textbf{return} $I^t$ \label{line:pt-dcrs-return}}
		\EndFor
	\end{algorithmic}
\end{algorithm}

\begin{algorithm}[th]
	\caption{\textsc{MatroidDCRSQuantized} }
	\label{alg:matroid-DCRS-quantized}
	\begin{algorithmic}[1]
        \Require matroid $\cM$, $b \in (0,1)$, $\eps \in (0, 1-b)$, online sequences $\{\vx^t\} \subseteq b \cdot P_{\cM}$ and TCOS output $\{\Rset^t\}_t$.
        \State{$r \gets 1+\eps/(1 + b)$ and $\tau \gets b + \eps$. \label{line:r_tau_settings_pt} }
        \State{$\sigma \gets$ arbitrary ordering of $\cN$.}
        \State{Initialize $S^0_i \gets \emptyset$, $\Iset^0_i \gets \emptyset$ for all $i \in [\ell]$.}
        \State{$X \gets 0$.}
        \For{$t = 1, 2, \ldots$}
            \State{$\vx \gets \vx^t$, $\Rset \gets \Rset^t$, $S_i \gets S_i^{t-1}$ and $\Iset_i\leftarrow \Iset_i^{t-1}$ for all $i \in [\ell]$.} 
						\State{If $\vx^t = \vx^{t-1}$, skip the rest of this iteration.}
            \State{$X \gets X + \norm{\vx^t - \vx^{t-1}}_1$.}
            \If{$X \geq b\cdot\rank(\cN)$ \label{line:hard_reset_pt} }
                \State{\Reset($0$).}
                \State{$X \gets 0$.}
            \EndIf
    		\For{$i = 1, 2, \ldots, \ell$ \label{line:main_S_loop_pt}}
                \algorithmiccomment{\textit{Move likely-spanned elements to $S_i$, update $\Iset_{i-1}$, \Reset if necessary.}}
                \If{$|S_i| > r \cdot \|\vx(i-1)\|_1 / \tau$\label{line:early_reset_condition_pt}}
                   \State{\Reset($i$).}
                \EndIf
                \State{$\bar p_e \gets$ average of $\ind{e \in \Span(\Rset \cup S_{i})}$ over $M(t, \eps')$ samples $\Rset \sim \Rdist(\vx(i-1))$, for every $e \in \cN$. \label{line:empirical_add_prob_pt}}
                \While{there is $e \in \Span(S_{i - 1}) \setminus \Span(S_i)$ obeying $\bar p_e \geq \tau - \eps'$ \label{line:add_condition_pt}}
                    \State Add $e$ to $S_{i}$ and remove $e$ from $\Iset_{i-1}$ if $e \in \Iset_{i-1}$.\label{line:add_pt}
                    \If{$S_i \cup \Iset_{i-1}$ is no longer independent\label{line:T_remove_condition_pt}}
                        \State{There exists $e' \in \Iset_{i-1}$ in the unique cycle in $S_i \cup \Iset_{i-1}$; remove $e'$ from $\Iset_{i-1}$.\label{line:T_remove_pt}}
                    \EndIf
		              \If{$|S_i| > r \cdot \|\vx(i-1)\|_1 / \tau$\label{line:reset_condition_pt}}
                      \State \Reset($i$).
                    \EndIf
                    \State{Recompute $\bar p_e$ from $M(t, \eps')$ fresh samples for all $e \in \cN$. \label{line:empirical_add_prob_recompute_pt}}
                \EndWhile
            \EndFor
            \For{$i = 0,1,  \ldots, \ell$ \label{line:main_T_loop_pt}}
                \algorithmiccomment{\textit{Add new elements from $\Rset$ to $\Iset_i$}.}
            	\State Update $\Iset_i \gets \Iset_i \cap \Rset$. \label{line:non-R_remove_pt}
    	         \While{there is an $e \in \Rset \cap (\Span(S_{i}) \setminus (\Span(S_{i+1} \cup \Iset_{i})))$ \label{line:T_add_condition_pt}}
                    \State Add the $\sigma$-first such $e$ to $\Iset_i$.\label{line:add_T_pt}
                \EndWhile
            \EndFor
            \State{Update $\Iset^t_i \leftarrow \Iset_i$ and $S_i^t \leftarrow S_i$ for all $i \in [\ell]$.} 
            \State{Set $I^t \leftarrow \cup_{i = 0}^{\ell} \Iset_i$ and \Return $\Iset^t$.}  \label{line:return_pt}
        \EndFor
        \Statex
        \Procedure{Reset}{$i$}
            \For{level $j=i,\ldots, \ell$}
                \State Set $S_j \gets \emptyset$. \label{line:reset_removal_S_pt}
                \State Set $\Iset_{j} \gets \emptyset$. \label{line:reset_removal_I_pt}
            \EndFor
        \EndProcedure
        \end{algorithmic}
\end{algorithm}

%% file: round-or-sep.tex
\section{Submodular Chasing in Polynomial Time}

\label{sec:round-or-sep}

To efficiently find a vector in the polytope $K_t$ (Equation~\eqref{line:polytope}), it suffices to give a separation oracle (see \cite[Appendix A]{BBLS23} for a polynomial-time reduction from general body chasing to halfspace chasing). Separating over the constraints defining the polytope $P$ depends on the properties of this polytope. For example, if $P$ is a matroid polytope, then it is well known that one can use submodular minimization to identify the most violated constraint. How should one handle the remaining constraints coming from the covering extension $f^*_t$? If there were a way to design a separation oracle for these constraints, then we could also evaluate the covering extension $f^*_t$, but doing so is known to be \textsf{APX}-hard. Thus, there is no hope for a traditional separation oracle. Nevertheless, Gupta and Levin \cite{DBLP:conf/soda/GuptaL20} show that if the multilinear extension $F(\vx)$ is sufficiently bounded away from the global maximum, then there is a randomized separation oracle that outputs a constraint violated by $\vx$ with high probability. The contrapositive of this statement is that if the procedure for finding violated constraints fails, then $\vx$ is already an approximate optimizer of the multilinear extension $F$, and finding such a vector was the goal all along! (They refer to this as the ``round-or-separate'' method.)

The details of how to use round-or-separate for low-recourse \emph{fractional} submodular optimization are formalized in concurrent work by Buchbinder, Naor and Wajc \cite[Theorem 4.2]{buchbinder2025chasingsubmodularobjectivessubmodular}. (We note that Buchbinder, Naor, and Wajc give a tighter analysis of the round-or-separate algorithm of \cite{DBLP:conf/soda/GuptaL20} with improved constant over the original version.)

\begin{theorem}[{\cite[Theorem 4.2]{buchbinder2025chasingsubmodularobjectivessubmodular}}]
\label{thm:fancy_low_rec_submod}
    For any $\eps>0$ and down-closed constraint $\cF$ with separable packing polytope $P \subseteq [0,1]^\cN$, there exists a polynomial-time deterministic $(1-e^{-1} -\eps)$-approximate $O(\nicefrac{1}{\epsilon} \log (\nicefrac{n}{\epsilon}))$-competitive recourse algorithm for the fractional submodular objective chasing problem. 
\end{theorem}

We can use this theorem in place of \cref{thm:postive_bodies} within the proof of \cref{thm:low_rec_submod}.

%% file: acknowledgements.tex
\section{Acknowledgments and AI Disclosure}

The authors would like to thank Yue Yang for early discussions. 

We used language models---in particular, a combination of GPT and Gemini---to generate \Cref{fig:framework_horizontal} and to scrutinize our proofs for errors.

%% file: dblp.bib
@inproceedings{DBLP:conf/soda/FeldmanSZ26,
  author       = {Moran Feldman and
                  Ola Svensson and
                  Rico Zenklusen},
  editor       = {Kasper Green Larsen and
                  Barna Saha},
  title        = {Nearly Tight Sample Complexity for Matroid Online Contention Resolution},
  booktitle    = {Proceedings of the 2026 Annual {ACM-SIAM} Symposium on Discrete Algorithms,
                  {SODA} 2026, Vancouver, BC, Canada, January 11-14, 2026},
  pages        = {4692--4711},
  publisher    = {{SIAM}},
  year         = {2026},
  url          = {https://doi.org/10.1137/1.9781611978971.171},
  doi          = {10.1137/1.9781611978971.171},
  bibsource    = {dblp computer science bibliography, https://dblp.org}
}

@inproceedings{DBLP:conf/icalp/ChekuriGQ15,
  author    = {Chandra Chekuri and
               Shalmoli Gupta and
               Kent Quanrud},
  editor    = {Magn{\'{u}}s M. Halld{\'{o}}rsson and
               Kazuo Iwama and
               Naoki Kobayashi and
               Bettina Speckmann},
  title     = {Streaming Algorithms for Submodular Function Maximization},
  booktitle = {Automata, Languages, and Programming - 42nd International Colloquium,
               {ICALP} 2015, Kyoto, Japan, July 6-10, 2015, Proceedings, Part {I}},
  series    = {Lecture Notes in Computer Science},
  volume    = {9134},
  pages     = {318--330},
  publisher = {Springer},
  year      = {2015},
  url       = {https://doi.org/10.1007/978-3-662-47672-7\_26},
  doi       = {10.1007/978-3-662-47672-7\_26},
  bibsource = {dblp computer science bibliography, https://dblp.org}
}

@article{DBLP:journals/jcss/AwerbuchAPW01,
  author    = {Baruch Awerbuch and
               Yossi Azar and
               Serge A. Plotkin and
               Orli Waarts},
  title     = {Competitive Routing of Virtual Circuits with Unknown Duration},
  journal   = {J. Comput. Syst. Sci.},
  volume    = {62},
  number    = {3},
  pages     = {385--397},
  year      = {2001},
  url       = {https://doi.org/10.1006/jcss.1999.1662},
  doi       = {10.1006/jcss.1999.1662},
  bibsource = {dblp computer science bibliography, https://dblp.org}
}

@inproceedings{DBLP:conf/soda/GuptaK14,
  author    = {Anupam Gupta and
               Amit Kumar},
  editor    = {Chandra Chekuri},
  title     = {Online {Steiner} Tree with Deletions},
  booktitle = {Proceedings of the Twenty-Fifth Annual {ACM-SIAM} Symposium on Discrete
               Algorithms, {SODA} 2014, Portland, Oregon, USA, January 5-7, 2014},
  pages     = {455--467},
  publisher = {{SIAM}},
  year      = {2014},
  url       = {https://doi.org/10.1137/1.9781611973402.34},
  doi       = {10.1137/1.9781611973402.34},
  bibsource = {dblp computer science bibliography, https://dblp.org}
}

@inproceedings{DBLP:conf/waoa/DahlmeierH25,
  author       = {J. Niklas Dahlmeier and
                  D. Ellis Hershkowitz},
  editor       = {Jannik Matuschke and
                  Jos{\'{e}} Verschae},
  title        = {Low Recourse Arborescence Forests Under Uniformly Random Arcs},
  booktitle    = {Approximation and Online Algorithms - 23rd International Workshop,
                  {WAOA} 2025, Warsaw, Poland, September 18-19, 2025, Proceedings},
  series       = {Lecture Notes in Computer Science},
  volume       = {16077},
  pages        = {142--156},
  publisher    = {Springer},
  year         = {2025},
  url          = {https://doi.org/10.1007/978-3-032-06706-7\_10},
  doi          = {10.1007/978-3-032-06706-7\_10},
  bibsource    = {dblp computer science bibliography, https://dblp.org}
}

@inproceedings{DBLP:conf/icml/DuettingFLNZ24,
  author       = {Paul D\"utting and
                  Federico Fusco and
                  Silvio Lattanzi and
                  Ashkan Norouzi{-}Fard and
                  Morteza Zadimoghaddam},
  title        = {Consistent Submodular Maximization},
  booktitle    = {Forty-first International Conference on Machine Learning, {ICML} 2024,
                  Vienna, Austria, July 21-27, 2024},
  publisher    = {OpenReview.net},
  year         = {2024},
  url          = {https://openreview.net/forum?id=AlJkqMnyjL},
  bibsource    = {dblp computer science bibliography, https://dblp.org}
}

@inproceedings{DBLP:conf/sosa/ChekuriSZ24,
  author       = {Chandra Chekuri and
                  Junkai Song and
                  Weizhong Zhang},
  editor       = {Merav Parter and
                  Seth Pettie},
  title        = {Contention Resolution for the \emph{{\(\ell\)}}-fold union
                  of a matroid via the correlation gap},
  booktitle    = {2024 Symposium on Simplicity in Algorithms, {SOSA} 2024, Alexandria,
                  VA, USA, January 8-10, 2024},
  pages        = {396--405},
  publisher    = {{SIAM}},
  year         = {2024},
  url          = {https://doi.org/10.1137/1.9781611977936.36},
  doi          = {10.1137/1.9781611977936.36},
  bibsource    = {dblp computer science bibliography, https://dblp.org}
}

@inproceedings{DBLP:conf/innovations/AlonGPRWW025,
  author       = {Noga Alon and
                  Nick Gravin and
                  Tristan Pollner and
                  Aviad Rubinstein and
                  Hongao Wang and
                  S. Matthew Weinberg and
                  Qianfan Zhang},
  editor       = {Raghu Meka},
  title        = {A Bicriterion Concentration Inequality and Prophet Inequalities for
                  k-Fold Matroid Unions},
  booktitle    = {16th Innovations in Theoretical Computer Science Conference, {ITCS}
                  2025, January 7-10, 2025, Columbia University, New York, NY, {USA}},
  series       = {LIPIcs},
  volume       = {325},
  pages        = {4:1--4:22},
  publisher    = {Schloss Dagstuhl - Leibniz-Zentrum f{\"{u}}r Informatik},
  year         = {2025},
  url          = {https://doi.org/10.4230/LIPIcs.ITCS.2025.4},
  doi          = {10.4230/LIPICS.ITCS.2025.4},
  bibsource    = {dblp computer science bibliography, https://dblp.org}
}

@inproceedings{DBLP:conf/sosa/FuLTTWW022,
  author       = {Hu Fu and
                  Pinyan Lu and
                  Zhihao Gavin Tang and
                  Abner Turkieltaub and
                  Hongxun Wu and
                  Jinzhao Wu and
                  Qianfan Zhang},
  editor       = {Karl Bringmann and
                  Timothy M. Chan},
  title        = {Oblivious Online Contention Resolution Schemes},
  booktitle    = {5th Symposium on Simplicity in Algorithms, SOSA@SODA 2022, Virtual
                  Conference, January 10-11, 2022},
  pages        = {268--278},
  publisher    = {{SIAM}},
  year         = {2022},
  url          = {https://doi.org/10.1137/1.9781611977066.20},
  doi          = {10.1137/1.9781611977066.20},
  bibsource    = {dblp computer science bibliography, https://dblp.org}
}

@inproceedings{DBLP:conf/stoc/DuttingFLNSZ25,
  author       = {Paul D{\"{u}}tting and
                  Federico Fusco and
                  Silvio Lattanzi and
                  Ashkan Norouzi{-}Fard and
                  Ola Svensson and
                  Morteza Zadimoghaddam},
  editor       = {Michal Kouck{\'{y}} and
                  Nikhil Bansal},
  title        = {The Cost of Consistency: Submodular Maximization with Constant Recourse},
  booktitle    = {Proceedings of the 57th Annual {ACM} Symposium on Theory of Computing,
                  {STOC} 2025, Prague, Czechia, June 23-27, 2025},
  pages        = {1406--1417},
  publisher    = {{ACM}},
  year         = {2025},
  url          = {https://doi.org/10.1145/3717823.3718131},
  doi          = {10.1145/3717823.3718131},
  bibsource    = {dblp computer science bibliography, https://dblp.org}
}

@article{DBLP:journals/mp/ChakrabartiK15,
  author    = {Amit Chakrabarti and
               Sagar Kale},
  title     = {Submodular maximization meets streaming: matchings, matroids, and
               more},
  journal   = {Math. Program.},
  volume    = {154},
  number    = {1-2},
  pages     = {225--247},
  year      = {2015},
  url       = {https://doi.org/10.1007/s10107-015-0900-7},
  doi       = {10.1007/s10107-015-0900-7},
  bibsource = {dblp computer science bibliography, https://dblp.org}
}

@inproceedings{DBLP:conf/stoc/AbboudA0PS19,
  author    = {Amir Abboud and
               Raghavendra Addanki and
               Fabrizio Grandoni and
               Debmalya Panigrahi and
               Barna Saha},
  editor    = {Moses Charikar and
               Edith Cohen},
  title     = {Dynamic set cover: improved algorithms and lower bounds},
  booktitle = {Proceedings of the 51st Annual {ACM} {SIGACT} Symposium on Theory
               of Computing, {STOC} 2019, Phoenix, AZ, USA, June 23-26, 2019},
  pages     = {114--125},
  publisher = {{ACM}},
  year      = {2019},
  url       = {https://doi.org/10.1145/3313276.3316376},
  doi       = {10.1145/3313276.3316376},
  bibsource = {dblp computer science bibliography, https://dblp.org}
}

@inproceedings{DBLP:conf/focs/BhattacharyaHN19,
  author    = {Sayan Bhattacharya and
               Monika Henzinger and
               Danupon Nanongkai},
  editor    = {David Zuckerman},
  title     = {A New Deterministic Algorithm for Dynamic Set Cover},
  booktitle = {60th {IEEE} Annual Symposium on Foundations of Computer Science, {FOCS}
               2019, Baltimore, Maryland, USA, November 9-12, 2019},
  pages     = {406--423},
  publisher = {{IEEE} Computer Society},
  year      = {2019},
  url       = {https://doi.org/10.1109/FOCS.2019.00033},
  doi       = {10.1109/FOCS.2019.00033},
  bibsource = {dblp computer science bibliography, https://dblp.org}
}

@inproceedings{DBLP:conf/soda/BhattacharyaHNW21,
  author    = {Sayan Bhattacharya and
               Monika Henzinger and
               Danupon Nanongkai and
               Xiaowei Wu},
  editor    = {D{\'{a}}niel Marx},
  title     = {Dynamic Set Cover: Improved Amortized and Worst-Case Update Time},
  booktitle = {Proceedings of the 2021 {ACM-SIAM} Symposium on Discrete Algorithms,
               {SODA} 2021, Virtual Conference, January 10 - 13, 2021},
  pages     = {2537--2549},
  publisher = {{SIAM}},
  year      = {2021},
  url       = {https://doi.org/10.1137/1.9781611976465.150},
  doi       = {10.1137/1.9781611976465.150},
  bibsource = {dblp computer science bibliography, https://dblp.org}
}

@article{DBLP:journals/siamcomp/ChekuriVZ14,
  author    = {Chandra Chekuri and
               Jan Vondr{\'{a}}k and
               Rico Zenklusen},
  title     = {Submodular Function Maximization via the Multilinear Relaxation and
               Contention Resolution Schemes},
  journal   = {{SIAM} J. Comput.},
  volume    = {43},
  number    = {6},
  pages     = {1831--1879},
  year      = {2014},
  url       = {https://doi.org/10.1137/110839655},
  doi       = {10.1137/110839655},
  bibsource = {dblp computer science bibliography, https://dblp.org}
}

@article{DBLP:journals/siamcomp/GuG016,
  author    = {Albert Gu and
               Anupam Gupta and
               Amit Kumar},
  title     = {The Power of Deferral: Maintaining a Constant-Competitive {Steiner}
               Tree Online},
  journal   = {{SIAM} J. Comput.},
  volume    = {45},
  number    = {1},
  pages     = {1--28},
  year      = {2016},
  url       = {https://doi.org/10.1137/140955276},
  doi       = {10.1137/140955276},
  bibsource = {dblp computer science bibliography, https://dblp.org}
}

@inproceedings{DBLP:conf/esa/AssadiS21,
  author    = {Sepehr Assadi and
               Shay Solomon},
  editor    = {Petra Mutzel and
               Rasmus Pagh and
               Grzegorz Herman},
  title     = {Fully Dynamic Set Cover via Hypergraph Maximal Matching: An Optimal
               Approximation Through a Local Approach},
  booktitle = {29th Annual European Symposium on Algorithms, {ESA} 2021, September
               6-8, 2021, Lisbon, Portugal (Virtual Conference)},
  series    = {LIPIcs},
  volume    = {204},
  pages     = {8:1--8:18},
  publisher = {Schloss Dagstuhl - Leibniz-Zentrum f{\"{u}}r Informatik},
  year      = {2021},
  url       = {https://doi.org/10.4230/LIPIcs.ESA.2021.8},
  doi       = {10.4230/LIPIcs.ESA.2021.8},
  bibsource = {dblp computer science bibliography, https://dblp.org}
}

@article{DBLP:journals/siamdm/ImaseW91,
  author    = {Makoto Imase and
               Bernard M. Waxman},
  title     = {Dynamic {Steiner} Tree Problem},
  journal   = {{SIAM} J. Discret. Math.},
  volume    = {4},
  number    = {3},
  pages     = {369--384},
  year      = {1991},
  url       = {https://doi.org/10.1137/0404033},
  doi       = {10.1137/0404033},
  bibsource = {dblp computer science bibliography, https://dblp.org}
}

@inproceedings{DBLP:conf/focs/GuptaL20,
	author    = {Anupam Gupta and
	Roie Levin},
	editor    = {Sandy Irani},
	title     = {Fully-Dynamic Submodular Cover with Bounded Recourse},
	booktitle = {61st {IEEE} Annual Symposium on Foundations of Computer Science, {FOCS}
	2020, Durham, NC, USA, November 16-19, 2020},
	pages     = {1147--1157},
	publisher = {{IEEE}},
	year      = {2020},
	url       = {https://doi.org/10.1109/FOCS46700.2020.00110},
	doi       = {10.1109/FOCS46700.2020.00110},
	bibsource = {dblp computer science bibliography, https://dblp.org}
}

@inproceedings{DBLP:conf/soda/GuptaL20,
	author    = {Anupam Gupta and
	Roie Levin},
	editor    = {Shuchi Chawla},
	title     = {The Online Submodular Cover Problem},
	booktitle = {Proceedings of the 2020 {ACM-SIAM} Symposium on Discrete Algorithms,
	{SODA} 2020, Salt Lake City, UT, USA, January 5-8, 2020},
	pages     = {1525--1537},
	publisher = {{SIAM}},
	year      = {2020},
	url       = {https://doi.org/10.1137/1.9781611975994.94},
	doi       = {10.1137/1.9781611975994.94},
	bibsource = {dblp computer science bibliography, https://dblp.org}
}

@inproceedings{DBLP:conf/icalp/FeldmanLNSZ22,
	author    = {Moran Feldman and
	Paul Liu and
	Ashkan Norouzi{-}Fard and
	Ola Svensson and
	Rico Zenklusen},
	editor    = {Mikolaj Bojanczyk and
	Emanuela Merelli and
	David P. Woodruff},
	title     = {Streaming Submodular Maximization Under Matroid Constraints},
	booktitle = {49th International Colloquium on Automata, Languages, and Programming,
	{ICALP} 2022, July 4-8, 2022, Paris, France},
	series    = {LIPIcs},
	volume    = {229},
	pages     = {59:1--59:20},
	publisher = {Schloss Dagstuhl - Leibniz-Zentrum f{\"{u}}r Informatik},
	year      = {2022},
	url       = {https://doi.org/10.4230/LIPIcs.ICALP.2022.59},
	doi       = {10.4230/LIPIcs.ICALP.2022.59},
	bibsource = {dblp computer science bibliography, https://dblp.org}
}

@inproceedings{avadhanula23fully,
  author       = {Vashist Avadhanula and
                  Andrea Celli and
                  Riccardo Colini{-}Baldeschi and
                  Stefano Leonardi and
                  Matteo Russo},
  editor       = {Brian Williams and
                  Yiling Chen and
                  Jennifer Neville},
  title        = {Fully Dynamic Online Selection through Online Contention Resolution Schemes},
  booktitle    = {Thirty-Seventh {AAAI} Conference on Artificial Intelligence, {AAAI}
                  2023, Thirty-Fifth Conference on Innovative Applications of Artificial
                  Intelligence, {IAAI} 2023, Thirteenth Symposium on Educational Advances
                  in Artificial Intelligence, {EAAI} 2023, Washington, DC, USA, February
                  7-14, 2023},
  pages        = {6693--6700},
  publisher    = {{AAAI} Press},
  year         = {2023},
  url          = {https://doi.org/10.1609/aaai.v37i6.25821},
  doi          = {10.1609/AAAI.V37I6.25821},
  bibsource    = {dblp computer science bibliography, https://dblp.org}
}

@inproceedings{feldman16online,
  author       = {Moran Feldman and
                  Ola Svensson and
                  Rico Zenklusen},
  editor       = {Robert Krauthgamer},
  title        = {Online Contention Resolution Schemes},
  booktitle    = {Proceedings of the Twenty-Seventh Annual {ACM-SIAM} Symposium on Discrete
                  Algorithms, {SODA} 2016, Arlington, VA, USA, January 10-12, 2016},
  pages        = {1014--1033},
  publisher    = {{SIAM}},
  year         = {2016},
  url          = {https://doi.org/10.1137/1.9781611974331.ch72},
  doi          = {10.1137/1.9781611974331.CH72},
  bibsource    = {dblp computer science bibliography, https://dblp.org}
}

@inproceedings{DBLP:conf/icalp/BanihashemBGHJM25,
  author       = {Kiarash Banihashem and
                  Leyla Biabani and
                  Samira Goudarzi and
                  MohammadTaghi Hajiaghayi and
                  Peyman Jabbarzade and
                  Morteza Monemizadeh},
  editor       = {Keren Censor{-}Hillel and
                  Fabrizio Grandoni and
                  Jo{\"{e}}l Ouaknine and
                  Gabriele Puppis},
  title        = {Dynamic Algorithms for Submodular Matching},
  booktitle    = {52nd International Colloquium on Automata, Languages, and Programming,
                  {ICALP} 2025, Aarhus, Denmark, July 8-11, 2025},
  series       = {LIPIcs},
  pages        = {19:1--19:21},
  publisher    = {Schloss Dagstuhl - Leibniz-Zentrum f{\"{u}}r Informatik},
  year         = {2025},
  url          = {https://doi.org/10.4230/LIPIcs.ICALP.2025.19},
  doi          = {10.4230/LIPICS.ICALP.2025.19}
}

@inproceedings{DBLP:conf/stoc/ChenP22,
  author       = {Xi Chen and
                  Binghui Peng},
  editor       = {Stefano Leonardi and
                  Anupam Gupta},
  title        = {On the complexity of dynamic submodular maximization},
  booktitle    = {{STOC} '22: 54th Annual {ACM} {SIGACT} Symposium on Theory of Computing,
                  Rome, Italy, June 20 - 24, 2022},
  pages        = {1685--1698},
  publisher    = {{ACM}},
  year         = {2022},
  url          = {https://doi.org/10.1145/3519935.3519951},
  doi          = {10.1145/3519935.3519951},
  bibsource    = {dblp computer science bibliography, https://dblp.org}
}

@inproceedings{DBLP:conf/nips/LattanziMNTZ20,
  author       = {Silvio Lattanzi and
                  Slobodan Mitrovic and
                  Ashkan Norouzi{-}Fard and
                  Jakub Tarnawski and
                  Morteza Zadimoghaddam},
  editor       = {Hugo Larochelle and
                  Marc'Aurelio Ranzato and
                  Raia Hadsell and
                  Maria{-}Florina Balcan and
                  Hsuan{-}Tien Lin},
  title        = {Fully Dynamic Algorithm for Constrained Submodular Optimization},
  booktitle    = {Advances in Neural Information Processing Systems 33: Annual Conference
                  on Neural Information Processing Systems 2020, NeurIPS 2020, December
                  6-12, 2020, virtual},
  year         = {2020},
  url          = {https://proceedings.neurips.cc/paper/2020/hash/9715d04413f296eaf3c30c47cec3daa6-Abstract.html},
  bibsource    = {dblp computer science bibliography, https://dblp.org}
}

@article{DBLP:journals/mor/HarshawKFK22,
  author       = {Christopher Harshaw and
                  Ehsan Kazemi and
                  Moran Feldman and
                  Amin Karbasi},
  title        = {The Power of Subsampling in Submodular Maximization},
  journal      = {Math. Oper. Res.},
  volume       = {47},
  number       = {2},
  pages        = {1365--1393},
  year         = {2022},
  url          = {https://doi.org/10.1287/moor.2021.1172},
  doi          = {10.1287/MOOR.2021.1172},
  bibsource    = {dblp computer science bibliography, https://dblp.org}
}


%% file: refs.bib
@misc{buchbinder2025chasingsubmodularobjectivessubmodular,
      title={Chasing Submodular Objectives, and Submodular Maximization via Cutting Planes}, 
      author={Niv Buchbinder and Joseph (Seffi) Naor and David Wajc},
      year={2025},
      eprint={2511.13605},
      archivePrefix={arXiv},
      primaryClass={cs.DS},
      url={https://arxiv.org/abs/2511.13605}, 
}

@BOOK{Schrijver-book,
  title = {Combinatorial optimization. {P}olyhedra and efficiency},
  publisher = {Springer-Verlag},
  year = {2003},
  author = {Schrijver, Alexander},
  volume = {24},
  pages = {xxxviii+1881},
  series = {Algorithms and Combinatorics},
  address = {Berlin},
  isbn = {3-540-44389-4},
  mrclass = {90-02 (05-02 52B55 68Q25 68R10 90C27 90C35 90C57)},
  mrnumber = {MR1956924 (2004b:90004a)},
  mrreviewer = {Alexander I. Barvinok}
}

@inproceedings{buchbindercompetitively,
	title={Competitively Consistent Clustering},
	author={Buchbinder, Niv and Levin, Roie and Yang, Yue},
	booktitle={Forty-second International Conference on Machine Learning},
    year = {2025}
}

@inproceedings{BBLS23,
  author       = {Sayan Bhattacharya and
                  Niv Buchbinder and
                  Roie Levin and
                  Thatchaphol Saranurak},
  title        = {Chasing Positive Bodies},
  booktitle    = {64th {IEEE} Annual Symposium on Foundations of Computer Science, {FOCS} 2023},
  pages        = {1694--1714},
  publisher    = {{IEEE}},
  year         = {2023}
}

@inproceedings{BLP22,
  author    = {Sayan Bhattacharya and
               Silvio Lattanzi and
               Nikos Parotsidis},
  title     = {Efficient and Stable Fully Dynamic Facility Location},
  booktitle   = {Thirty-sixth Conference on Neural Information Processing Systems
(NeurIPS)},
  year      = {2022},
  }

@inproceedings{KLS23,
  author       = {Ravishankar Krishnaswamy and
                  Shi Li and
                  Varun Suriyanarayana},
  editor       = {Barna Saha and
                  Rocco A. Servedio},
  title        = {Online Unrelated-Machine Load Balancing and Generalized Flow with
                  Recourse},
  booktitle    = {Proceedings of the 55th Annual {ACM} Symposium on Theory of Computing,
                  {STOC} 2023, Orlando, FL, USA, June 20-23, 2023},
  pages        = {775--788},
  publisher    = {{ACM}},
  year         = {2023},
  url          = {https://doi.org/10.1145/3564246.3585222},
  doi          = {10.1145/3564246.3585222},
  bibsource    = {dblp computer science bibliography, https://dblp.org}
}

@article{vondrak2007submodularity,
	title={Submodularity in combinatorial optimization},
	author={Vondr{\'a}k, Jan},
	journal={PhD thesis, Charles University, Prague, Czech Republic},
	year={2007}
}

@inproceedings{GKS14,
  author    = {Anupam Gupta and
               Amit Kumar and
               Cliff Stein},
  title     = {Maintaining Assignments Online: Matching, Scheduling, and Flows},
  booktitle = {Proceedings of the Twenty-Fifth Annual {ACM-SIAM} Symposium on Discrete
               Algorithms, {SODA} 2014},
  pages     = {468--479},
  publisher = {{SIAM}},
  year      = {2014}
}

@inproceedings{GKKP17,
  author    = {Anupam Gupta and
               Ravishankar Krishnaswamy and
               Amit Kumar and
               Debmalya Panigrahi},
  title     = {Online and dynamic algorithms for set cover},
  booktitle = {Proceedings of the 49th Annual {ACM} {SIGACT} Symposium on Theory
               of Computing, {STOC} 2017, Montreal, QC, Canada, June 19-23, 2017},
  pages     = {537--550},
  publisher = {{ACM}},
  year      = {2017}
}

@article{B,
author = {C. Burch},
title = {Machine learning in metrical task systems and other on-line problems},
journal = {Ph.D. Thesis, published as CMU Tech Report CMU-CS-00-135.}
}

@inproceedings{BBM,
 author = {Y. Barta and B. Bollob\'{a}s and M. Mendel},
 title = {A Ramsey-Type Theorem for Metric Spaces and its Applications for Metrical Task Systems and Related
 Problems},
 booktitle = {Proceedings of the 42nd IEEE symposium on Foundations of Computer Science},
 year = {2001},
 pages = {396}
 }

@article{avin2020dynamic,
	title={Dynamic balanced graph partitioning},
	author={Avin, Chen and Bienkowski, Marcin and Loukas, Andreas and Pacut, Maciej and Schmid, Stefan},
	journal={SIAM Journal on Discrete Mathematics},
	volume={34},
	number={3},
	pages={1791--1812},
	year={2020},
	publisher={SIAM}
}

@inproceedings{avin2016online,
	title={Online balanced repartitioning},
	author={Avin, Chen and Loukas, Andreas and Pacut, Maciej and Schmid, Stefan},
	booktitle={Distributed Computing: 30th International Symposium, DISC 2016, Paris, France, September 27-29, 2016. Proceedings},
	pages={243--256},
	year={2016},
	organization={Springer}
}

@article{azar2023competitive,
	title={Competitive Vertex Recoloring: (Online Disengagement)},
	author={Azar, Yossi and Machluf, Chay and Patt-Shamir, Boaz and Touitou, Noam},
	journal={Algorithmica},
	pages={1--27},
	year={2023},
	publisher={Springer}
}

@inproceedings{brodal1999dynamic,
	title={Dynamic representations of sparse graphs},
	author={Brodal, Gerth St{\o}lting and Fagerberg, Rolf},
	booktitle={Algorithms and Data Structures: 6th International Workshop, WADS’99 Vancouver, Canada, August 11--14, 1999 Proceedings 6},
	pages={342--351},
	year={1999},
	organization={Springer}
}

@inproceedings{bhattacharya2022simple,
  author    = {Sayan Bhattacharya and
               Thatchaphol Saranurak and
               Pattara Sukprasert},
  title     = {Simple Dynamic Spanners with Near-Optimal Recourse Against an Adaptive
               Adversary},
  booktitle = {30th Annual European Symposium on Algorithms (ESA)},
  series    = {LIPIcs},
  volume    = {244},
  pages     = {17:1--17:19},
  year      = {2022},
  }

@article{baswana2012fully,
	title={Fully dynamic randomized algorithms for graph spanners},
	author={Baswana, Surender and Khurana, Sumeet and Sarkar, Soumojit},
	journal={ACM Transactions on Algorithms (TALG)},
	volume={8},
	number={4},
	pages={1--51},
	year={2012},
	publisher={ACM New York, NY, USA}
}

@inproceedings{assadi2018fully,
	title={Fully dynamic maximal independent set with sublinear update time},
	author={Assadi, Sepehr and Onak, Krzysztof and Schieber, Baruch and Solomon, Shay},
	booktitle={Proceedings of the 50th Annual ACM SIGACT Symposium on theory of computing},
	pages={815--826},
	year={2018}
}

@article{solomon2020improved,
	title={Improved dynamic graph coloring},
	author={Solomon, Shay and Wein, Nicole},
	journal={ACM Transactions on Algorithms (TALG)},
	volume={16},
	number={3},
	pages={1--24},
	year={2020},
	publisher={ACM New York, NY, USA}
}

@inproceedings{sawlani2020near,
	title={Near-optimal fully dynamic densest subgraph},
	author={Sawlani, Saurabh and Wang, Junxing},
	booktitle={Proceedings of the 52nd Annual ACM SIGACT Symposium on Theory of Computing},
	pages={181--193},
	year={2020}
}

@inproceedings{guo2020facility,
	title={On the facility location problem in online and dynamic models},
	author={Guo, Xiangyu and Kulkarni, Janardhan and Li, Shi and Xian, Jiayi},
	booktitle={Approximation, Randomization, and Combinatorial Optimization. Algorithms and Techniques (APPROX/RANDOM 2020)},
	year={2020},
	organization={Schloss Dagstuhl-Leibniz-Zentrum f{\"u}r Informatik}
}

@inproceedings{lkacki2015power,
	title={The power of dynamic distance oracles: Efficient dynamic algorithms for the {Steiner} tree},
	author={{\L}{\k{a}}cki, Jakub and O{\'c}wieja, Jakub and Pilipczuk, Marcin and Sankowski, Piotr and Zych, Anna},
	booktitle={Proceedings of the forty-seventh annual ACM symposium on Theory of computing},
	pages={11--20},
	year={2015}
}

@inproceedings{bera2022new,
	title={A new dynamic algorithm for densest subhypergraphs},
	author={Bera, Suman K and Bhattacharya, Sayan and Choudhari, Jayesh and Ghosh, Prantar},
	booktitle={Proceedings of the ACM Web Conference 2022},
	pages={1093--1103},
	year={2022}
}
